\UseRawInputEncoding
\documentclass[10pt,authoryear]{article}
\usepackage{appendix}
\usepackage{setspace}
\usepackage[font=small,  margin=2cm]{caption}
\usepackage[tbtags]{amsmath}
\usepackage{amsthm,amssymb,amsfonts}
\usepackage{natbib}
\usepackage{multirow}
\usepackage{lscape}
\usepackage{subfigure}
\usepackage{makecell}
\usepackage{exscale}
\usepackage{booktabs}
\usepackage{array}
\usepackage{fullpage}
\usepackage{url}
\usepackage{algorithm}
\usepackage{algpseudocode}
\usepackage{bm}
\usepackage{smile}
\usepackage{mathtools}
\usepackage{wrapfig}
\usepackage{lipsum}
\usepackage{threeparttable}
\usepackage{mathrsfs}
\usepackage{relsize}
\usepackage{dsfont}
\usepackage{multirow}
\usepackage{apalike}
\usepackage[usenames,dvipsnames,svgnames,table]{xcolor}
\usepackage[colorlinks=true,linkcolor=blue,urlcolor=blue,citecolor=blue]{hyperref}

\newcommand{\bargmin}{\mathop{\mathrm{arg\ min}}}
\newcommand{\bargmax}{\mathop{\mathrm{arg\ max}}}
\def\tr{{\mathop{\text{\rm Tr}}}}
\def\vec{{\mathop{\text{\rm Vec}}}}
\numberwithin{equation}{section}
\numberwithin{theorem}{section}
\numberwithin{corollary}{section}
\numberwithin{definition}{section}
\renewcommand{\baselinestretch}{2}

\newtheorem{assumption}{Assumption}
\newtheorem{assumptionalt}{Assumption}[assumption]
\newenvironment{assumptionp}[1]{
  \renewcommand\theassumptionalt{#1}
  \assumptionalt
}{\endassumptionalt}

\begin{document}

\title{\LARGE  A new non-parametric Kendall's tau for matrix-valued elliptical observations}

	\author{
	Yong He\footnotemark[1],~
	Yalin Wang\footnotemark[1]
,~Long Yu\footnotemark[2],~
	Wang Zhou\footnotemark[3],~
	Wen-Xin Zhou\footnotemark[4]
	}
\renewcommand{\thefootnote}{\fnsymbol{footnote}}
\footnotetext[1]{Institute for Financial Studies, Shandong University, China; e-mail:{\tt heyong@sdu.edu.cn}}
\footnotetext[2]{Shanghai University of Finance and Economics, China. e-mail:{\tt fduyulong@163.com}}
\footnotetext[3]{National University of Singapore, Singapore. e-mail:{\tt wangzhou@nus.edu.sg }}
\footnotetext[4]{University of California, San Diego, USA. e-mail:{\tt wez243@ucsd.edu}}

\maketitle

In this article, we first propose generalized row/column matrix Kendall's tau  for matrix-variate observations that are ubiquitous in areas such as finance and medical imaging. For a random matrix following a matrix-variate elliptically contoured distribution,  we show that the eigenspaces of
the proposed row/column matrix Kendall's tau  coincide with those of the row/column scatter matrix respectively, with the same descending order of the eigenvalues. We perform eigenvalue decomposition to the generalized row/column matrix Kendall's tau  for recovering the loading spaces of the matrix factor model. We also propose to estimate the pair of the factor numbers by exploiting the eigenvalue-ratios of the  row/column matrix Kendall's tau.  Theoretically,  we derive the convergence rates of the estimators for loading spaces, factor scores and common components, and prove the consistency of the estimators for the factor numbers without any moment constraints on the idiosyncratic errors.
Thorough simulation studies are conducted to show  the higher degree of robustness of the proposed
estimators over the existing ones. Analysis of a financial dataset of asset returns and a medical imaging dataset associated with COVID-19 illustrate
the empirical
usefulness of the proposed method.

\vspace{0.5em}


\section{Introduction}
The spatial/multivariate Kendall's tau, as a powerful nonparametric/robust tool for analyzing random vectors, is first introduced in \cite{choi1998multivariate} for testing independence, and henceforth has been used for estimating covariance matrices, principal components and factor models in both low and high-dimensions.
However, it remains unclear how to extend such a tool to analyze matrix-variate data, which are now ubiquitous in areas   such as finance, macroeconomics and biology. The naive vectorization technique ignores the intrinsic matrix structure, and thus is less efficient for analyzing matrix-variate data. In the current work, we introduce a new concept, named matrix Kendall's tau, for analyzing random matrices, and use this tool to estimate the loading spaces of the matrix elliptical factor model introduced in \cite{Li2022Manifold}. Although we focus on the application of this new tool to factor models, the proposed matrix Kendall's tau has wide applications for analyzing matrix-variate data including  but not limited to performing 2-Dimensional Principal Component Analysis (PCA) and estimating separable covariance matrices robustly \citep{yu2021central,ke2019user}.
\subsection{Literature Review on Related Works}
 Multivariate rank-based statistics are widely discussed in robust statistics, see for example  \cite{Tyler1987distribution,Oja2010Multivariate,hallin2002optimal,hallin2010optimal,han2014scale,zhou2019extreme}, among many others.
Among various  multivariate rank-based statistics, the multivariate Kendall' tau is attractive due to its robustness property in dealing with elliptically contoured random vectors as \cite{marden1999some} and \cite{Croux2002Sign} showed that the population multivariate Kendall's tau
shares the same eigenspace as the scatter/covariance matrix. The elliptical family provides more flexibility in modeling modern
complex data than the Gaussian family in terms of capturing heavy-tail property and
nontrivial tail dependence between variables (variables tend to go to extremes together), which are important in areas such as finance and macroeconomics. It is well known that financial data such as portfolio/asset returns are heavy-tailed with nontrivial tail dependence.  With the aid of multivariate Kendall's tau, \cite{Taskinen2012Robustifying} characterized the robustness and efficiency properties of Elliptical Component Analysis (ECA) in low dimensions, and \cite{Han2018ECA} proposed ECA procedures for both non-sparse and sparse settings in high dimensions. Moreover,  \cite{Fan2018Large} proposed Principal Orthogonal complEment Thresholding (POET) for large-scale covariance matrix estimation based on the approximate elliptical factor model, and \cite{yu2019robust}  studied the estimation of factor number for large-dimensional elliptical factor model. More recently, \cite{He2022large} proposed a Robust Two-Step (RTS) procedure for estimating loading and factor spaces.

Matrix-variate data arise when one observes a group of variables structured in a well defined matrix form, and have been frequently observed in various research areas such as finance, signal processing and medical imaging. The naive vectorization technique vectorizes the matrix observations into long vectors and thus ignores the intrinsic matrix structure. As a result, it may lead to inefficient statistical inference and also suffers from heavy computational burden.
 Modelling  matrix-valued data by Matrix-elliptical family not only enjoys the  flexibility in modeling heavy-tail property and tail dependencies of matrix-variate data, but also maintains the inherent matrix structures. A natural question arises: {\it can we define a similar Kendall's tau matrix for analyzing random elliptical matrices, parallel to the multivariate Kendall's tau for elliptical vectors?} In this work, we give an affirmative answer to this question by proposing a new type of Kendall's tau, named ``matrix Kendall's tau".  In detail, assume $\Xb_{p\times q}$ is a random matrix and $\widetilde\Xb_{p\times q}$ is an independent copy of $\Xb_{p\times q}$. We define the row/column matrix Kendall's tau $\Kb_r,\Kb_c$  as
\begin{equation}
\Kb_{r}=\EE\left(\frac{(\Xb-\tilde{\Xb})(\Xb-\tilde{\Xb})^\top}{\norm{ \Xb-\tilde{\Xb}}_{F}^{2}}\right), \ \ \
\Kb_{c}=\EE\left(\frac{(\Xb-\tilde{\Xb})^\top(\Xb-\tilde{\Xb})}{\norm{ \Xb-\tilde{\Xb}}_{F}^{2}}\right).\nonumber
\end{equation}
The ``matrix Kendall's tau" can be viewed as a generalization of multivariate Kendall' tau to the random matrix setting. When $q=1$, $\Kb_r$ is reduced to the conventional multivariate Kendall' tau proposed by \cite{choi1998multivariate}.
The matrix Kendall's tau is particularly suitable for analyzing matrix-elliptical distributions.  Given a random matrix $\Xb \sim E_{p,q}(\Mb,\bSigma \otimes\bOmega,\psi)$ (see Definition \ref{def:MED}), we show that $\Kb_{r}$ and $\bSigma$ share the same eigenspace while  $\Kb_{c}$ and $\bOmega$ share the same eigenspace,  with the same descending order of the eigenvalues.

A closely related line of research is on large-dimensional factor model,
which is a powerful tool for summarizing and extracting information from large datasets and draws growing attention in the ``big-data" era. In the context of large-dimensional approximate vector factor models, we refer to \cite{bai2002determining}, \cite{stock2002forecasting},
  \cite{bai2003inferential},\cite{onatski09}, \cite{ahn2013eigenvalue}, \cite{fan2013large}, \cite{Trapani2018A}, \cite{barigozzi2018simultaneous}, \cite{barigozzi2020generalized}, \cite{Sahalia2017Using}, \cite{yu2019robust}, \cite{Sahalia2020High}, \cite{Chen2021Quantile} and \cite{He2022large}.
Although (approximate) vector factor models have been extensively studied, they are inadequate to model matrix-variate observations that have drawn growing attention in recent years. \cite{wang2019factor} first proposed the following matrix factor model for matrix series $\{\Xb_t, 1\leq t\leq T\}$:
\begin{equation}\label{equ:matrixfactormodel}
  \underbrace{\Xb_t}_{p_1\times p_2}=\underbrace{\Rb}_{p_1\times k_1}\times \underbrace{\Fb_t}_{k_1\times k_2}\times \underbrace{\Cb^\top}_{k_2\times p_2}+  \underbrace{\Eb_t}_{p_1\times p_2},
\end{equation}
where $\Rb$ is the row factor loading matrix exploiting the variations of $\Xb_{t}$ across the rows, $\Cb$ is the $p_{2}\times k_{2}$ column factor loading matrix reflecting the differences across the columns of $\Xb_{t}$, $\Fb_{t}$ is the common factor matrix for all cells in $\Xb_{t}$, and $\Eb_{t}$ is the idiosyncratic component. This model along with its variant has been studied by \citet{chen2019constrained}, \cite{Chen2020Modeling}, \cite{Gao2021A},
\citet{fan2021},  \cite{He2021Vector}, \cite{he2021online}, \cite{Yu2021Projected}.  \cite{Li2022Manifold} proposed a Matrix Elliptical Factor Model (MEFM), assuming that $\{(\text{Vec}(\Fb_t)^\top,\text{Vec}(\Eb_t)^\top)^\top, t=1,\ldots,T\}$ in (\ref{equ:matrixfactormodel}) follow some vector elliptical distribution. They further proposed  a Manifold Principal Component Analysis (MPCA) method to estimate the loading spaces. However, the MPCA method would result in loading space estimators with a slow $\sqrt{T}$ convergence rate compared to the typical convergence rate of $\sqrt{Tp_1}$ (or $\sqrt{Tp_2}$) for matrix factor model \citep{Chen2020StatisticalIF,Yu2021Projected}. For matrix factor analysis, it is also  critical to determine the pair of factor numbers $(k_1,k_2)$.  \cite{Chen2020StatisticalIF} proposed an $\alpha$-PCA based eigenvalue-ratio method; \cite{Yu2021Projected}  proposed a projection-based iterative eigenvalue-ratio method; \cite{He2021Matrix} proposed a robust iterative eigenvalue-ratio method to estimate the numbers of factors using the Huber loss  \citep{Huber1964Robust}.
All these works are based on  the eigenvalue ratio idea from \cite{ahn2013eigenvalue} with an exception of \cite{He2021Vector} who proposed to  determine the factor numbers  from the perspective of sequential hypothesis testing borrowing idea from \cite{Trapani2018A}.

\subsection{Contributions and Paper Organization}

The contributions of the current work lie in the following aspects. Firstly, we propose a new type of Kendall's tau, named matrix Kendall's tau, that generalizes the multivariate Kendall' tau to the random matrix setting. We show that $\Kb_{r}$ and $\bSigma$ share the same eigenspace while  $\Kb_{c}$ and $\bOmega$ share the same eigenspace,  with the same descending order of the eigenvalues. The sample version of matrix Kendall's tau is a U-statistic with a bounded kernel under operator norm and enjoys the same distribution-free property as multivariate Kendall's tau, which can also be directly extended to analyze high dimensional data. Secondly, we apply this new type of Kendell's tau to the Matrix Elliptical Factor Model (MEFM), and  propose a Matrix-type Robust Two Step (MRTS) method to estimate the loading and factor spaces under MEFM. The proposed estimator is ``unbiased"  and achieves faster convergence rates ($\sqrt{Tp_1}$ or $\sqrt{Tp_2}$) than Manifold Principal Component Analysis (MPCA)  for estimating the loading spaces.
As an illustration, we check the empirical performances of some state-of-the-art methods to the heavy-tailedness of the factor and idiosyncratic errors with a synthetic dataset.
 Figure  \ref{fig:1} shows that the proposed MRTS method performs comparably with the $\alpha$-PCA method by \cite{fan2021}, PE method  by \cite{Yu2021Projected} and RMFA by \cite{He2021Matrix}, and outperforms the MPCA by \cite{Li2022Manifold}. As the tails become heavier, the performance of   $\alpha$-PCA and PE deteriorates quickly, while MRTS, RMFA and MPCA exhibit robustness to different extents, and RMTS always performs the best in heavy-tailed cases. Thirdly, we  provide consistent estimators of the pair of factor numbers without any moment constraints on the underlying distribution. At last, we point out that the matrix Kendall's tau has wide applications to analyze matrix-variate data in addition to the MEFM discussed here, such as robust estimation of separable covariance matrix or robust 2-dimensional PCA.

\begin{figure}[!h]
  \centering
  \includegraphics[width=17cm, height=6cm]{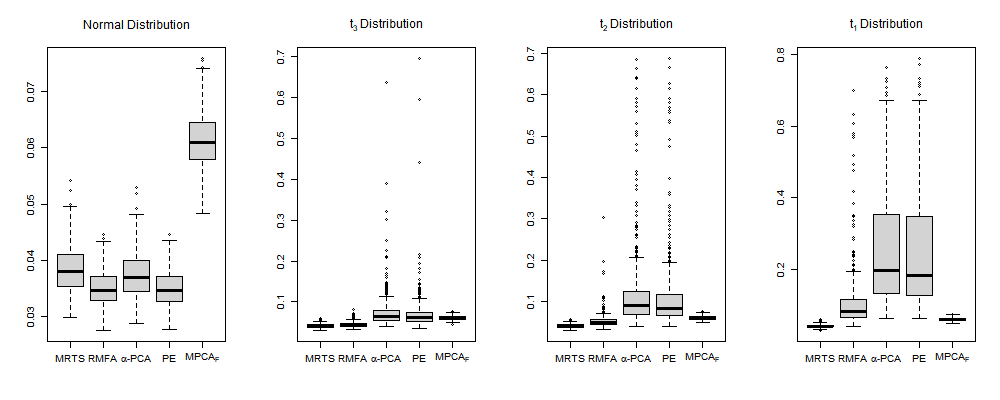}
 \caption{Boxplot of the distance between the estimated row loading space and the true row loading space by MRTS, RMFA, $\alpha$-PCA, PE, MPCA$_F$ methods under different distributions. $p_1=p_2=T=50$.}\label{fig:1}
 \end{figure}
The rest of the article proceeds as follows. In Section 2, we introduce a new type of Kendall's tau and briefly discuss its properties along with its sample version. We then introduce a Matrix-type Robust Two Step (MRTS)  procedure for the MEFM with the aid of matrix Kendall's tau, parallel to the RTS procedure for vector elliptical factor model with the aid of multivariate Kendall's tau  \citep{He2022large}. We further propose robust estimators of the pair of factor numbers by exploiting the eigenvalue-ratios of the sample matrix Kendall's tau.  In Section 3, we investigate the theoretical properties of the proposed estimators for MEFM. In Section 4, thorough simulation studies are conducted to illustrate the advantages of the proposed estimators over the state-of-the-art methods.
In Section 5, we analyze a financial dataset and a medical imaging dataset associated with COVID-19 to illustrate the empirical performance/usefulness of the proposed methods.
 We discuss possible future research directions and conclude the article in Section 6. The proofs of the main theorems and
additional details are collected in the supplementary materials.

Notations adopted throughout the paper are as follows. For any vector $\bmu=(\mu_1,\ldots,\mu_p)^\top \in \RR^p$, let $\|\bmu\|_2=(\sum_{i=1}^p\mu_i^2)^{1/2}$, $\|\bmu\|_\infty=\max_i|\mu_i|$. For a real number $a$, denote  $[a]$ as the largest integer smaller than or equal to $a$,  let $sgn(a) = 1$ if $a\geq 0$ and
$sgn(a) =-1$ if $a<0$ . Let $I(\cdot)$ be the indicator function. Let ${\rm diag}(a_1,\ldots,a_p)$ be a $p\times p$ diagonal matrix, whose diagonal entries are $a_1\ldots,a_p$.   For a matrix $\Ab$, let $\mathrm{A}_{ij}$ (or $\mathrm{A}_{i,j}$) be the $(i,j)$-th entry of $\Ab$, $\Ab_{i,\cdot}$ be the $i$-th row and $\Ab_{\cdot,j}$  be the $j$-th column, $\Ab^\top$ the transpose of $\Ab$, ${\rm Tr}(\Ab)$ the trace of $\Ab$, $\text{rank}(\Ab)$ the rank of $\Ab$ and $\text{diag}(\Ab)$ the diagonal matrix composed of the diagonal elements of $\Ab$. $\Ab \succcurlyeq \mathbf{0}$ means $\Ab$ is a non-negative definite matrix. Denote $\lambda_j(\Ab)$ as the $j$-th largest eigenvalue of a nonnegative definitive matrix $\Ab$, and let $\|\Ab\|$ be the spectral norm of matrix $\Ab$ and $\|\Ab\|_F$ be the Frobenius norm of $\Ab$. For matrices $\Ab$ and $\Bb$, let ${\rm diag}(\Ab,\Bb)$ be a block diagonal matrix with diagonal matrices $\Ab$ and $\Bb$. For two series of random variables, $X_n$ and $Y_n$, $X_n\asymp Y_n$ means $X_n=O_p(Y_n)$ and $Y_n=O_p(X_n)$. For two random variables (vectors) $\bX$ and $\bY$, $\bX\stackrel{d}{=}\bY$ means the distributions of $\bX$ and $\bY$ are the same. The constants $c,C_1,C_2$ in different lines can be nonidentical.

\section{Methodology}
In this section, we introduce the generalized Kendall's tau, i.e., the Matrix Kendall's tau and then propose a Matrix-type Robust Two Step (MRTS) procedure for estimating row/column loading and factor spaces for the Matrix Elliptical Factor Model (MEFM).
\subsection{ Matrix Kendall's tau}
Prior to introducing the Matrix Kendall's tau, we first briefly review the definition of
Matrix Elliptical Distribution (MED). We refer to \cite{Gupta2018Matrix} for further details.
\begin{definition}\label{def:MED}
A random matrix $\Yb$ of size $p \times q$ follows the matrix elliptical distribution if its characteristic function has the form $\varphi_{Y}(\Tb )=\exp[\tr(i\Tb^\top\Mb)]\psi[\tr(\Tb^\top\bSigma\Tb\bOmega)]$ with $\Tb$: $p \times q$, $\Mb$: $p \times q$, $\bSigma$: $p \times p$, $\bOmega$:$q \times q$, $\bSigma \succcurlyeq \mathbf{0}$,$\bOmega \succcurlyeq \mathbf{0}$ and $\psi:[0,\infty)\rightarrow \RR$, briefly denoted as
$\Yb\sim E_{p,q}(\Mb,\bSigma \otimes\bOmega,\psi)$. Assume $\text{rank}(\bSigma)=m$, $\text{rank}(\bOmega)=n$, the random matrix $\Yb \sim E_{p,q}(\Mb,\bSigma \otimes\bOmega,\psi)$ if and only if
\begin{equation}
\Yb \overset{d}{=} r\Ab\Ub\Bb^\top +\Mb, \nonumber
\end{equation}
where $\Mb$ is a deterministic location matrix of size $p \times q$,  $\Ub$ is a random matrix of dimension $m \times n$ and Vec$(\Ub)$ is uniformly distributed on the unit sphere in $\RR^{mn}$, $r$ is a nonnegative random variable independent of $\Ub$, $\bSigma=\Ab\Ab^\top$ and $\bOmega=\Bb\Bb^\top$ are rank factorizations of $\bSigma$ and $\bOmega$.
\end{definition}
Throughout the paper, we call $\bSigma$ as row scatter matrix and $\bOmega$ as column scatter matrix.
We first introduce the matrix Kendall's tau for random matrices, which generalizes the multivariate Kendall's tau for random vectors by \cite{marden1999some}.

\begin{definition}
Let $\Yb \sim E_{p,q}(\Mb,\bSigma \otimes\bOmega,\psi)$ be a continuous random matrix, and $\widetilde\Yb$ is an independent copy, the row/column matrix Kendall's tau are defined as:
\begin{equation}\label{def:maken}
\Kb_{r}=\EE\left(\frac{(\Yb-\tilde{\Yb})(\Yb-\tilde{\Yb})^\top}{\norm{ \Yb-\tilde{\Yb}}_{F}^{2}}\right), \ \ \
\Kb_{c}=\EE\left(\frac{(\Yb-\tilde{\Yb})^\top(\Yb-\tilde{\Yb})}{\norm{ \Yb-\tilde{\Yb}}_{F}^{2}}\right).
\end{equation}
\end{definition}
Both $ \Kb_r$ and $ \Kb_c$ are positive semidefinite (PSD) matrices of trace 1. The following proposition illustrates  the connections between $ \Kb_r$ ($\Kb_c$) and $\bSigma$ ($\bOmega$).
\begin{proposition}\label{pro1}
Let $\Yb \sim E_{p,q}(\Mb,\bSigma \otimes\bOmega,\psi)$ be a continuous random matrix, $\Kb_{r}$ be the row matrix multivariate Kendall's tau. Then if $\text{rank}(\bSigma)=m$, we have
\begin{equation}
\lambda_{j}(\Kb_{r})=\EE\left(\frac{\lambda_{j}(\bSigma)\Ub_{j,\cdot}\bOmega^{*}\Ub_{j,\cdot}^\top}{\lambda_{1}(\bSigma)\Ub_{1,\cdot}\bOmega^{*}\Ub_{1,\cdot}^\top+\cdots+\lambda_{m}(\bSigma)\Ub_{m,\cdot}\bOmega^{*}\Ub_{m,\cdot}^\top}\right),\nonumber
\end{equation}
where $\bOmega^{*}=\Bb^\top\Bb$. In addition, $\Kb_{r}$ and $\bSigma$ share the same eigenspace with the same descending order of the eigenvalues.
Furthermore, let $\Kb_{c}$ be the column matrix Kendall's tau,  if $\text{rank}(\bOmega)=n$, we have
\begin{equation}
\lambda_{j}(\Kb_{c})=\EE\left(\frac{\lambda_{j}(\bOmega)\Ub_{\cdot,j}^\top\bSigma^{*}\Ub_{\cdot,j}}{\lambda_{1}(\bOmega)\Ub_{\cdot,1}^\top\bSigma^{*}\Ub_{\cdot,1}+\cdots+\lambda_{n}(\bOmega)\Ub_{\cdot,n}^\top\bSigma^{*}\Ub_{\cdot,n}}\right),\nonumber
\end{equation}
where $\bSigma^{*}=\Ab^\top\Ab$. In addition, $\Kb_{c}$ and $\bOmega$ share the same eigenspace with the same descending order of the eigenvalues.
Without loss of generality, we assume that $\bOmega^{*}$ and $\bSigma^{*}$ are both diagonal matrices. Otherwise, we can always find orthogonal matrices $\Pb$ and $\Qb$ such that $\Bb^*=\Bb\Pb$ and $\Ab^*=\Ab\Qb$ satisfy $\Bb^{*\top}\Bb^{*}$ and $\Ab^{*\top}\Ab^{*}$ are diagonal, and
$\Yb \overset{d}{=} r\Ab^*(\Qb^\top\Ub\Pb)\Bb^{*\top}+\Mb\overset{d}{=}r\Ab^*\Ub\Bb^{*\top}+\Mb$.
\end{proposition}
Proposition \ref{pro1} shows that, to recover the eigenspace of the covariance matrix $\bSigma$ ($\bOmega$), we can resort to recovering the eigenspace of $\Kb_r$ ($\Kb_c$), which, as will be discussed in the following, can be efficiently estimated by U-statistics.

Let $\Yb_1,\ldots,\Yb_T\in\RR^{p\times q}$  be $T$ independent data points of a random matrix $\Yb \sim E_{p,q}(\Mb,\bSigma \otimes\bOmega,\psi)$, and the definition of the matrix  Kendall's tau in (\ref{def:maken}) motivates one to consider the following sample matrix
Kendall's  tau estimator, which are  second-order matrix U-statistics:
\[
\widehat \Kb_r=\frac{2}{T(T-1)}\sum_{t<t^\prime}\frac{(\Yb_t-\Yb_{t^\prime})(\Yb_t-\Yb_{t^\prime})^\top}{\|\Yb_t-\Yb_{t^\prime}\|_F^2}, \ \
\widehat \Kb_c=\frac{2}{T(T-1)}\sum_{t<t^\prime}\frac{(\Yb_t-\Yb_{t^\prime})^\top(\Yb_t-\Yb_{t^\prime})}{\|\Yb_t-\Yb_{t^\prime}\|_F^2}.
\]
It can be easily derived that $\EE(\widehat \Kb_r)=\Kb_r$ and $\EE(\widehat \Kb_c)=\Kb_c$ and both $\widehat \Kb_r$ and $\widehat \Kb_c$ are positive semidefinite (PSD) matrices of trace 1. The kernels of the U-statistics $k_{\text{MK}}^r(\cdot)$ : $\RR^{p\times q}\times\RR^{p\times q}\rightarrow \RR^{p\times p}$ and $k_{\text{MK}}^c(\cdot)$ : $\RR^{p\times q}\times\RR^{p\times q}\rightarrow \RR^{q\times q}$, defined as
\[
k_{\text{MK}}^r(\Yb_t,\Yb_{t^\prime})=\frac{(\Yb_t-\Yb_{t^\prime})(\Yb_t-\Yb_{t^\prime})^\top}{\|\Yb_t-\Yb_{t^\prime}\|_F^2}, \ \
k_{\text{MK}}^c(\Yb_t,\Yb_{t^\prime})=\frac{(\Yb_t-\Yb_{t^\prime})^\top(\Yb_t-\Yb_{t^\prime})}{\|\Yb_t-\Yb_{t^\prime}\|_F^2},
\]
are bounded under the spectral norm, i.e., $\|k_{\text{MK}}^r(\Yb_t,\Yb_{t^\prime})\|\leq 1, \|k_{\text{MK}}^c(\Yb_t,\Yb_{t^\prime})\|\leq 1$. Intuitively, such a boundedness property
makes the U-statistics $\widehat \Kb_r$  and $\widehat \Kb_c$ more robust to heavy-tailed distributions.
We  also prove that the kernels $k_{\text{MK}}^r(\Yb_t,\Yb_{t^\prime})$, $k_{\text{MK}}^c(\Yb_t,\Yb_{t^\prime})$ are distribution-free kernel, i.e., for any continuous $\Yb \sim E_{p,q}(\Mb,\bSigma \otimes\bOmega,\psi)$,
\[
k_{\text{MK}}^r(\Yb_t,\Yb_{t^\prime})\overset{d}{=}k_{\text{MK}}^r(\Zb_t,\Zb_{t^\prime}), \ \
k_{\text{MK}}^c(\Yb_t,\Yb_{t^\prime})\overset{d}{=}k_{\text{MK}}^c(\Zb_t,\Zb_{t^\prime}),
\]
where $\Zb_t,\Zb_{t^\prime}$ follow matrix normal distribution $\Zb \sim \mathcal{MN}(\mathbf{0},\bSigma,\bOmega)$ (check Lemma \ref{lemfree} for details). Therefore, $\widehat \Kb_r$  and $\widehat \Kb_c$  enjoy the same distribution-free property as  Tyler's
M-estimator \citep{Tyler1987distribution}  and multivariate Kendall's tau \citep{marden1999some}. However, the matrix Kendall's tau  can be directly extended to analyze high dimensional data as multivariate Kendall's tau, while Tyler's M-estimator cannot \citep{Han2018ECA}.
\subsection{Matrix Elliptical Factor Model}
For $p_1 \times p_2$ matrix  sequences $\{\Xb_{t},t=1,\dots,T\}$, the matrix factor model is as follows:
\begin{equation}
\Xb_{t}=\Rb\Fb_{t}\Cb^\top+\Eb_{t} , t=1,\dots,T, \nonumber
\end{equation}
where $\Rb$ is the $p_1 \times k_{1}$ row factor loading matrix, $\Cb$ is the $p_2 \times k_{2}$ column factor loading matrix, $\Fb_{t}$ is the common factor matrix and $\Eb_{t}$ is the idiosyncratic components.
In the definition of MEFM, we assume the factor matrix $\Fb_t$ and noise matrix $\Eb_t$ are from a joint matrix elliptical distribution with location parameter zero, that is,

\begin{equation}\label{eq_joint_ellip_model}
	\left(\begin{array}{c}\vec(\Fb_t) \\ \vec(\Eb_t)\end{array}\right)=r_{t}\left(\begin{array}{cc}\Ib_{k_2}\otimes \Ib_{k_1} & \zero \\ \zero & \bOmega_e^{1/2}\otimes \bSigma_e^{1/2}\end{array}\right) \frac{\bZ_t}{\|\bZ_t\|_2},
\end{equation}
where $\bZ_t$ is a $(k_1k_2+p_1p_2)$-dimensional isotropic Gaussian vector, and $r_t$ is a positive random variable independent of $\bZ_t$. In fact,  $\Fb_t$ and $\Eb_t$ are both matrix elliptically distributed, that is,
\begin{equation*}
	\Fb_t = \frac{r_t}{\|\bZ_t\|_2} \Zb_t^{F}\sim E_{k_1, k_2}(\zero,\Ib_{k_1} \otimes \Ib_{k_2}, \psi_F), \quad \Eb_t = \frac{r_t}{\|\bZ_t\|_2}\bSigma_e^{1/2} \Zb_t^{E}\bOmega_e^{1/2}\sim E_{p_1, p_2}(\zero, \bSigma_e \otimes \bOmega_e, \psi_E) ,
\end{equation*}
where $\bSigma_e$  of size $p_1\times p_1$ and $\bOmega_e$ of size $p_2\times p_2$ are positive-definite matrices , $\Zb_t^F$ is a $k_1\times k_2$ random matrix by stacking the first $k_1\times k_2$ elements of $\bZ_t$, and $\Zb_t^E$ is of size $p_1\times p_2$ by stacking all the elements left.

Under the above assumption, and by the property of the elliptical vectors, we have that the scatter matrix of $\text{Vec}(\Xb_t)$, denoted as $\bSigma_{\text{Vec}(\Xb_t)}$, has the form
\[
\bSigma_{\text{Vec}(\Xb_t)}=(\Rb\Rb^\top)\otimes(\Cb\Cb^\top)+\bSigma_e\otimes \bOmega_e,
\]
or equivalently, $\bSigma_{\text{Vec}(\Xb_t)}$ has a low-rank matrix plus sparse matrix structure. Neglecting the sparse term ($\bSigma_e\otimes \bOmega_e$), one would find that $\Xb\sim E_{p_1,p_2}(\zero, (\Rb\Rb^\top)\otimes(\Cb\Cb^\top), \psi_X)$, which  motivates one to estimate $\Rb$ and $\Cb$ by the leading eigenvectors of the row/column matrix Kendall's tau, as Proposition \ref{pro1} shows that the row/column matrix Kendall's tau matrix shares the same eigenspace with the row/column scatter matrix with the same descending order of the eigenvalues.
\subsection{Matrix-type Robust Two step  Procedure for MEFM}
In this section, we introduce a Matrix-type Robust Two step (MRTS)
procedure for fitting large-dimensional matrix elliptical factor model. In the
first step, we propose to estimate $\Rb$ and $\Cb$ by the eigenvectors of the row/column matrix Kendall's tau. Let $\widehat\Kb^X_r$ and $\widehat\Kb^X_c$ be the
sample matrix Kendall' tau based on the observations $\{\Xb_{t},t=1,\dots,T\}$, i.e.,
\begin{equation}\label{equ:sammaken}
\widehat \Kb_r^X=\frac{2}{T(T-1)}\sum_{t<t^\prime}\frac{(\Xb_t-\Xb_{t^\prime})(\Xb_t-\Xb_{t^\prime})^\top}{\|\Xb_t-\Xb_{t^\prime}\|_F^2}, \ \
\widehat \Kb_c^X=\frac{2}{T(T-1)}\sum_{t<t^\prime}\frac{(\Xb_t-\Xb_{t^\prime})^\top(\Xb_t-\Xb_{t^\prime})}{\|\Xb_t-\Xb_{t^\prime}\|_F^2}.
\end{equation}
As the eigenvectors of the matrix Kendall's tau matrix $\Kb_r^X$ ($\Kb_r^X$) are
identical to the eigenvectors of the scatter matrix $\bSigma$ ($\bOmega$),  we
estimate the factor loading matrix $\Rb$ ($\Cb$) by $\sqrt{p_1}$ ($\sqrt{p_2}$) times the leading $k_1$ ($k_2$)
eigenvectors of $\widehat \Kb_r^X$ ($\widehat \Kb_c^X$).
In detail, let $\{\hat\br_1,\ldots,\hat\br_{k_1}\}$ be the leading $k_1$ eigenvectors of $\widehat \Kb_r^X$,  $\{\hat\bc_1,\ldots,\hat\bc_{k_2}\}$ be the leading $k_2$ eigenvectors of $\widehat \Kb_c^X$, We set  $\widehat \Rb=\sqrt{p_1}(\hat\br_1,\ldots,\hat\br_{k_1})$ and $\widehat \Cb=\sqrt{p_2}(\hat\bc_1,\ldots,\hat\bc_{k_2})$ as the estimators of the factor loading matrices $\Rb$ and $\Cb$ respectively.  The numbers of factors $(k_1,k_2)$ are relatively small compared with $p_1,p_2$ and $T$. We first
assume that $(k_1,k_2)$ are known and fixed. In section \ref{sec:4},  we provide  consistent estimators for  $(k_1,k_2)$.

In the second step, we estimate the factor matrices $\{\Fb_1,\ldots,\Fb_T\}$ by
regressing $\vec{(\Xb_t)}$ on $(\widehat \Cb\otimes \widehat \Rb)$. In detail,  $\vec{(\Fb_t)}$ is estimated by the following least-squares-type estimator
\[
 \vec{(\widehat\Fb_t)}=\bargmin_{\bbeta\in\RR_{k_1k_2}}\|\vec(\Xb_t)-(\widehat \Cb\otimes \widehat \Rb)\bbeta\|^2.
\]
In fact, it can be shown that by solving the above optimization problem would lead to ${\widehat\Fb_t}=\widehat \Rb^\top\Xb_t\widehat \Cb/(p_1p_2)$.

An alternative approach is to extend the RTS procedure for vector elliptical factor model \citep{He2022large} by directly vectorizing matrix-variate observations. This, however, would ignore the intrinsic matrix structure, lead to inefficient statistical inference and bring in heavy computational burden.
For estimating the loading spaces, applying the previous RTS procedure to vectorized data has a complexity of order $O(T^2p_1^2p_2^2)$, while the complexity of the proposed MRTS procedure is $O\big(T^2p_1p_2\max(p_1,p_2)\big)$.

\subsection{Model Selection}
\label{sec:4}

Thus far we have assumed that the pair of  factor numbers $(k_1 ,k_2)$  is known a priori. In practice, both both $k_1$ and $k_2$  are unknown and thus need be estimated before implementing the RMTS
procedure. In  this section, we introduce a robust method for determining the pair of
factor numbers, which does not require any moment constraints and is of independent interest.
Our ``Matrix Kendall's tau Eigenvalue-Ratio" (MKER) method is  motivated by \cite{ahn2013eigenvalue} and \cite{yu2019robust}.

Given matrix observations $\{\Xb_1,\ldots,\Xb_T\}$, let $\widehat\Kb_r^X$ and  $\widehat\Kb_c^X$ be the sample matrix Kendall's tau given in (\ref{equ:sammaken}) with eigenvalues $\lambda_j(\hat \Kb_r^X),j=1,\ldots,p_1$ and $\lambda_j(\hat \Kb_c^X),j=1,\ldots,p_2$, respectively.
For a prescribed maximum number of factors $k_{{\rm max}}$, we construct the Matrix Kendall's tau Eigenvalue Ratio (``MKER") estimators as
	\begin{equation}\label{equ:fn}
	\hat k_1=\bargmax_{1\le j\le k_{{\rm max}}}\frac{\lambda_j(\hat \Kb_r^X)}{\lambda_{j+1}(\hat \Kb_r^X)},\ \ \hat k_{2}=\argmax_{1\le j\le k_{{\rm max}}}\frac{\lambda_j(\hat \Kb_c^X)}{\lambda_{j+1}(\hat \Kb_c^X)}.
	\end{equation}
To ensure that the denominators are away from zero, in practice one may add a positive but asymptotically negligible term to each $\lambda_j(\hat \Kb_r^X)$ and $\lambda_j(\hat \Kb_c^X)$. For instance,  take $\delta_{1}=1/\sqrt{\min{(p_2,T^{1-\epsilon})}}$, $\delta_{2}=1/\sqrt{\min{(p_1,T^{1-\epsilon})}}$ for a small $\epsilon>0$ and let $\hat\lambda_j(\hat \Kb_r^X)=\lambda_j(\hat \Kb_r^X)+c\delta_{1}$, $\hat\lambda_j(\hat \Kb_c^X)=\lambda_j(\hat \Kb_c^X)+c\delta_{2}$ with a small positive constant $c$. We then replace $\lambda_j(\hat \Kb_r^X)$ and $\lambda_j(\hat \Kb_c^X)$ with $\hat\lambda_j(\hat \Kb_r^X)$ and $\hat\lambda_j(\hat \Kb_c^X)$ in (\ref{equ:fn}), respectively. 

\section{Theoretical Analysis}
In this section, we investigate the theoretical properties of the proposed estimators. We first formally introduce the matrix elliptical factor model along with some structural assumptions.

\begin{assumptionp}{A}[Joint Matrix Elliptical Model]\label{joint_elliptical}
	We say the matrix-variate observations $\Xb_t \in \RR^{p_1 \times p_2}$ follow a matrix elliptical factor model if

\begin{equation*}
\Xb_t =\Rb \Fb_t\Cb^{\top}+\Eb_t,\quad t=1,\dots,T,
\end{equation*}
where $(\text{Vec}(\Fb_t)^\top,\text{Vec}(\Eb_t)^\top)^\top$ satisfies (\ref{eq_joint_ellip_model}), and $\bZ_t$'s therein  are {\it i.i.d.} $(k_1k_2+p_1p_2)$-dimensional standard multivariate Gaussian vectors, $r_t$'s are {\it i.i.d.} positive random scalars that are independent of $\bZ_t$'s. We consider the high-dimensional regime that $p_1, p_2 \to \infty$ and $\log(\max(p_1, p_2) ) = o(T)$. In addition, we assume $ r_t = O_p(\sqrt{p_1 p_2})$ as $p_1, p_2 \to \infty$.
\end{assumptionp}

\begin{assumptionp}{B}[Strong Factor Conditions]\label{strong_factor}
We assume $\Rb^{\top} \Rb/p_1 \rightarrow \Vb_1$ and $\Cb^{\top} \Cb/p_2 \rightarrow \Vb_2$ for some positive definite matrices $\Vb_1$ and $\Vb_2$ as $p_1, p_2 \to \infty$. There exist positive constants $c_1,c_2$
such that $c_1\leq\lambda_{k_1}(\Vb_1)<\cdots<\lambda_{1}(\Vb_1)\leq c_2$,
$c_1\leq\lambda_{k_2}(\Vb_2)<\cdots<\lambda_{1}(\Vb_2)\leq c_2$.
\end{assumptionp}

\begin{assumptionp}{C}[Regular Noise Conditions]\label{regular_noise}
We assume there exist positive constants $c_1$ and $C_1$ such that $c_1\leq \lambda_{p_1}(\bSigma_e)\leq \lambda_{1}(\bSigma_e)\leq C_2$, $c_2\leq \lambda_{p_2}(\bOmega_e)\leq \lambda_{1}(\bOmega_e)\leq C_2$ as $p_1, p_2\rightarrow\infty$.

\end{assumptionp}

Assumption \ref{joint_elliptical} is the joint matrix elliptical model assumption adopted in \cite{Li2022Manifold}. {The condition $r_t=O_p(\sqrt{p_1p_2})$ essentially assumes that $r_t=\sqrt{p_1p_2}\times r^*_t$ for some random variable $r^*_t$; see \eqref{eq_joint_ellip_model}. For instance, assume $(\text{Vec}(\Fb_t)^\top,\text{Vec}(\Eb_t)^\top)^\top$ follows a multivariate $t$ distribution with degrees of freedom  $v$. Then we have $r^*_t\overset{d}{=}\sqrt{F_{d,v}}$, where $F_{d,v}$ stands for $F$-distribution with degrees of freedom $d,v$ where $d=(p_1p_2+k_1k_2)$. It is worth noting that  no moment constraint is exerted on $r^*_t$ or $r_t$, thus allowing for Cauchy-type distributions that do not even have finite means.
This differentiates the current work from the literature on matrix factor models such as \cite{fan2021} and \cite{Yu2021Projected}.} Assumptions \ref{strong_factor} and \ref{regular_noise} are standard in large-dimensional matrix factor analysis literature, see for example \cite{fan2021} and \cite{Yu2021Projected}. The eigenvalues of $\Vb_i, i=1,2$ are
assumed to be distinct such that the corresponding eigenvectors are identifiable. We assume strong factor conditions
in Assumption \ref{strong_factor}, indicating that the row and column factors are pervasive along both dimensions.

As our first main result, the following theorem provides the convergence rates of the factor loading matrix estimators $\widehat \Rb$ and $\widehat \Cb$ under both the Frobenious and spectral norms.
\begin{theorem}\label{th1}
Under Assumptions \ref{joint_elliptical}-\ref{regular_noise} with $(k_1, k_2)$ fixed, and $p_1,p_2,T\rightarrow\infty$,  there exist  matrices $\hat{\Hb}_{R}$ and $\hat{\Hb}_{C}$ such that $\hat\Hb_R^\top\Vb_1\hat\Hb_R\rightarrow \Ib_{k_1}$,  $\hat\Hb_C^\top\Vb_2\hat\Hb_C\rightarrow \Ib_{k_2}$ and
\begin{equation}
\frac{1}{p_1}\norm{\hat\Rb-\Rb\hat{\Hb}_{R}}_F^2=O_p\left(\frac{1}{Tp_2}+\frac{1}{p_1^2}\right),\quad
\frac{1}{p_2}\norm{\hat\Cb-\Cb\hat{\Hb}_{C}}_F^2=O_p\left(\frac{1}{Tp_1}+\frac{1}{p_2^2}\right).\nonumber
\end{equation}
Moreover,
\begin{equation}
\frac{1}{p_1}\norm{\hat\Rb-\Rb\hat{\Hb}_{R}}^2=O_p\left(\frac{1}{Tp_2}+\frac{1}{p_1^2}\right),\quad
\frac{1}{p_2}\norm{\hat\Cb-\Cb\hat{\Hb}_{C}}^2=O_p\left(\frac{1}{Tp_1}+\frac{1}{p_2^2}\right).\nonumber
\end{equation}
\end{theorem}
\begin{table*}[!h]
 \begin{center}
  \small
  \addtolength{\tabcolsep}{0pt}
  \caption{A comparison of convergence rates of loading matrices estimators under squared Frobenius norm. ``IE" and ``PE" denote the initial estimator and projection estimator by \cite{Yu2021Projected}, ``$\alpha$-PCA" denotes  the estimator by \cite{fan2021}, MPCA$_F$ denotes the estimator by \cite{Li2022Manifold}. }\label{tab:conver}
   \renewcommand{\arraystretch}{5}
  \scalebox{0.7}{
    \begin{tabular*}{23cm}{ccccccccc}
                 \toprule[1.2pt]
     & &\textbf{MRTS}&\textbf{IE} &\textbf{$\alpha$-PCA}&\textbf{PE}&\textbf{MPCA$_F$}\\
     \hline
 &\textbf{Working model}&Matrix Elliptical Factor Model &Matrix Factor Model &Matrix Factor Model&Matrix Factor Model&Matrix Elliptical Factor model \\
 &$\frac{1}{p_1}\norm{ \hat{\Rb}-\Rb\hat{\Hb}_R}_F^2$&$O_p(\frac{1}{Tp_2}+\frac{1}{p_1^2})$ &$O_p(\frac{1}{Tp_2}+\frac{1}{p_1^2})$& $O_p(\frac{1}{Tp_2}+\frac{1}{p_1})$& $O_p(\frac{1}{Tp_2}+\frac{1}{p_1^2p_2^2}+\frac{1}{T^2p_1^2})$&$O_p(\frac{1}{T}+\frac{1}{\sqrt{p_2}})$\\
 &$\frac{1}{p_2}\norm{ \hat{\Cb}-\Cb\hat{\Hb}_C}_F^2$&$O_p(\frac{1}{Tp_1}+\frac{1}{p_2^2})$ &$O_p(\frac{1}{Tp_1}+\frac{1}{p_2^2})$ & $O_p(\frac{1}{Tp_1}+\frac{1}{p_2})$&$O_p(\frac{1}{Tp_1}+\frac{1}{p_1^2p_2^2}+\frac{1}{T^2p_2^2})$&$O_p(\frac{1}{T}+\frac{1}{\sqrt{p_1}})$\\

    \toprule[1.2pt]

  \end{tabular*}}
 \end{center}
\end{table*}
Assume we vectorize the matrix observations by stacking the columns, model (\ref{eq_joint_ellip_model}) degenerates to the vector elliptical factor model in \cite{He2022large} with loading matrix $(\Cb\otimes \Rb)$. Following the theoretical analysis by \cite{He2022large}, the convergence for the loading matrix under squared Frobenius norm will be $O_p(T^{-1}+(p_1p_2)^{-2})$. Therefore, the matrix elliptical factor model will be more efficient as long as $T\ll\min\{p_1^2,p_2^2\}$, i.e., under ``large dimension versus small sample size" regimes. Table \ref{tab:conver} compares the convergence rates of loading matrices estimators under squared Frobenius norm. The derived rates for MRTS in Theorem \ref{th1} are much faster than those of  the  MPCA$_F$ method by \cite{Li2022Manifold}. Indeed, the derived rates for MRTS match with those of the Initial Estimators (``IE" in Table \ref{tab:conver}) in \cite{Yu2021Projected}, which are slightly better than the convergence rates of $\alpha$-PCA methods proposed by \cite{fan2021}. The rates of the projected estimators (``PE" in Table \ref{tab:conver}) by \cite{Yu2021Projected} are faster under light-tailed cases due to the projection technique. It is  possible to further improve the convergence rates of MRTS  using the projection technique proposed by \cite{Yu2021Projected}. This is beyond the scope of this work, and will be left for future research.

Theorems \ref{th2} and \ref{th3} below show consistency of the factor scores $\widehat \Fb_t$ and the common components $\widehat \Sbb_t=\widehat \Rb\widehat \Fb_t\widehat \Cb^\top$, both under the Frobenius norm.

\begin{theorem}\label{th2}
Under  Assumptions \ref{joint_elliptical}-\ref{regular_noise} with $(k_1, k_2)$ fixed, and $\min \{T,p_1,p_2 \} \rightarrow \infty$, we have
$$\norm{\hat{\Fb}_t-\hat{\Hb}_R^{-1}\Fb_t\hat{\Hb}_C^{-1\top}}_F^2=O_p\bigg(\frac{1}{p_1p_2}\bigg).$$
\end{theorem}

\begin{theorem}\label{th3}
Under  Assumptions \ref{joint_elliptical}-\ref{regular_noise} with $(k_1, k_2)$ fixed, and $\min \{T,p_1,p_2 \} \rightarrow \infty$, we have
$$\frac{1}{p_1p_2}\norm{\hat{\Sbb}_t-\Sbb_t}_F^2=O_p\bigg(\frac{1}{p_1p_2}+\frac{1}{Tp_1}+\frac{1}{Tp_2} \bigg) ,$$
where $\Sbb_t =\Rb  \Fb_t  \Cb^\top$.
\end{theorem}

The convergence rates of factor scores and signal matrices are comparable to those in \cite{He2022large}, \cite{Yu2021Projected} and \cite{Chen2020StatisticalIF}. This is because that as long as the the loadings are estimated accurately, the estimation errors for $\Fb_t$  mainly depend on the idiosyncratic errors $\Eb_t$. Even if under the oracle case that the loading matrices are given, the  estimation error is still of order $(p_1p_2)^{-1}$ under the squared Frobenius norm. The estimation errors from $\hat \Fb_t$ will further  affect the estimation accuracy for the signal matrices $\Sbb_t$.

The following theorem shows  the consistency of the MKER estimators introduced in Section~\ref{sec:4}.
\begin{theorem}\label{k1k2Consistency}
Under Assumptions \ref{joint_elliptical}-\ref{regular_noise} with $k_1, k_2 \geq 1$, we have $\PP(\hat{k}_1=k_1) \rightarrow 1$ and $\PP(\hat{k}_2=k_2) \rightarrow 1$ as $\min \{T,p_1,p_2 \} \rightarrow \infty$.
\end{theorem}

\section{Simulation Study}
\subsection{Data Generation}
We first describe the data generation mechanism of the synthetic dataset to assess the performance of the proposed \textbf{M}atrix-\textbf{R}obust-\textbf{T}wo-\textbf{S}tep (MRTS) method compared with some state-of-the-art methods.

We use similar data-generating models as in \cite{He2021Matrix}. We set $k_{1}=k_{2}=3$, draw the entries of $\Rb$ and $\Cb$ independently from the uniform distribution $\mathcal{U}(-1,1)$, and let
	 \[
\begin{array}{cccc}
	 \text{Vec}(\mathbf{F}_t) = \phi \times \text{Vec}(\mathbf{F}_{t-1}) + \sqrt{1-\phi^2} \times \text{Vec}(\mathbf{\bepsilon}_t), \\
	 \text{Vec}(\mathbf{E}_t) = \psi \times \text{Vec}(\mathbf{E}_{t-1}) + \sqrt{1-\psi^2} \times \text{Vec}(\mathbf{U}_t),
\end{array}
	 \]
where $\phi$ and $\psi$ controls the temporal
correlations, $\text{Vec}(\mathbf{\bepsilon}_t)$ and $\text{Vec}(\Ub_t)$ are jointly generated from elliptical distributions. Before we give the data generating scenarios, we first review the matrix normal distribution and matrix $t$-distribution. In detail, when a random matrix $\mathbf{U}_t$ is from a matrix normal distribution $ \mathcal{MN} (\mathbf{0},\mathbf{U}_E,\mathbf{V}_E) $, then Vec($\mathbf{U}_t$) ${\sim}$ $\mathcal{N}$($\mathbf{0}$, $\mathbf{V}_E$ $\otimes$ $\mathbf{U}_E$). If $\mathbf{U}_t$ is from a matrix $t$-distribution $t_{p_1,p_2} (\nu,\mathbf{M},\mathbf{U}_E,\mathbf{V}_E) $, it has probability density function
	 \[f(\mathbf{U}_t;\nu,\mathbf{M},\mathbf{U}_E,\mathbf{V}_E)
	 =K \times \Big\vert \mathbf{I}_{p_1} \mathbf{U}_E^{-1} (\mathbf{U}_t -\mathbf{M}) \mathbf{V}_E^{-1} (\mathbf{U}_t -\mathbf{M})^\top \Big\vert^{-\dfrac{\nu+p_1+p_2-1}{2}},
	 \]
where $K$ is the regularization parameter. In our simulation study, we set $\Mb=\zero$ and let $\mathbf{U}_E$ and $\mathbf{V}_E$ be matrices with ones on the diagonal, and with $ {1}/{p_1} $ and $ {1}/{p_2} $ as off-diagonal entries, respectively.

\subsection{Estimation error for loading spaces}

In this section, we compare the MRTS method with the RMFA method by \cite{He2021Matrix}, the $\alpha$-PCA ($\alpha=0$) method by \cite{fan2021}, the PE method by \cite{Yu2021Projected} and the MPCA$_{F}$ by \cite{Li2022Manifold}. We consider the following two scenarios:

 \textbf{Scenario A}: $\phi=0$ and $\psi=0$, $T \in \{20,50,100\},p_1=p_2 \in \{ 20,50 \}$, $\big(\text{Vec}(\mathbf{\bepsilon}_t)^\top,\text{Vec}(\Ub_t)^\top\big)^\top$ are generated in the following ways:
     (i)  $\big(\text{Vec}(\mathbf{\bepsilon}_t)^\top,\text{Vec}(\Ub_t)^\top\big)^\top$  are {\it i.i.d.} jointly Gaussian distributions $\mathcal{N}\Big(\mathbf{0},\text{diag}\big(\Ib_{k_2}\otimes \Ib_{k_1},\Ub_{E}\otimes \Vb_{E}\big)\Big)$;
     (ii) $\big(\text{Vec}(\mathbf{\bepsilon}_t)^\top,\text{Vec}(\Ub_t)^\top\big)^\top$ are {\it i.i.d.}  jointly $t$ distributions $t_{\nu}\Big(\mathbf{0},\text{diag}\big(\Ib_{k_2}\otimes \Ib_{k_1},\Ub_{E}\otimes \Vb_{E}\big)\Big)$, $\nu=3,2,1$.

 \textbf{Scenario B}: $\phi=0.1$ and $\psi=0.1$, $T \in \{20,50\},p_1=p_2 \in \{ 20,50 \}$, $\big(\text{Vec}(\mathbf{\bepsilon}_t)^\top,\text{Vec}(\Ub_t)^\top\big)^\top$  are generated in the same ways as in Scenario A.

We adopt a metric between linear spaces which was also utilized in \cite{He2021Matrix}. For two column-wise orthogonal matrices $(\bQ_1)_{p\times q_1}$ and $(\bQ_2)_{p\times q_2}$, we define
\[
\mathcal{D}(\bQ_1,\bQ_2)=\bigg(1-\frac{1}{\max{(q_1,q_2)}}\mbox{Tr}\Big(\bQ_1\bQ_1^{\top}\bQ_2\bQ_2^{\top}\Big)\bigg)^{1/2}.
\]
By the definition of $\mathcal{D}(\bQ_1,\bQ_2)$, we can easily deduce that its value  lies in the interval $[0,1]$, which measures the distance between the column spaces spanned by  $\bQ_1$ and $\bQ_2$. The column spaces spanned by $\bQ_1$ and $\bQ_2$  are the same when $\mathcal{D}(\bQ_1,\bQ_2)=0$, while    are orthogonal when $\mathcal{D}(\bQ_1,\bQ_2)=1$.  The Gram-Schmidt orthonormalizing transformation can be applied when $\bQ_1$ and $\bQ_2$ are not column-orthogonal matrices.

\begin{table}[!h]
	 	\caption{Averaged estimation errors and standard errors of $\mathcal{D} (\hat{\mathbf{R}}, \mathbf{R})$ for Scenario A under joint normal distribution and joint $t$ distribution over 500 replications.}
	 	 \label{tab:main1}\renewcommand{\arraystretch}{1} \centering
	 	\selectfont
	 	\begin{threeparttable}
	 		 \scalebox{0.9}{\begin{tabular*}{18.5cm}{cccccccccccccccccccc}
	 				\toprule[2pt]
	 			    &\multirow{1}{*}{Evaluation}   &\multirow{1}{*}{$T$}
	 				&\multirow{1}{*}{$p_1$}        &\multirow{1}{*}{$p_2$}
                    &\multirow{1}{*}{MRTS}          &\multirow{1}{*}{RMFA}
                    &\multirow{1}{*}{$\alpha$-PCA} &\multirow{1}{*}{PE}
                    &\multirow{1}{*}{MPCA$_{F}$}   \cr
	 			    \cmidrule(lr){8-11} \\
	 				\midrule[1pt]
	 				
	 				 &&&& \multicolumn{5}{c}{\multirow{1}{*}{\textbf{Normal Distribution}}}\\
	 \cmidrule(lr){3-12}
	 				&$\mathcal{D} (\hat{\mathbf{R}}, \mathbf{R})$ &20 &20  &20 &0.1189(0.0363)&0.0922(0.0162)&0.1127(0.0352)&0.0927(0.0166)&0.1407(0.0223)\\
	 				                                           &  &   &50  &50  & 0.0622(0.0075)  &0.0569(0.0059) &0.0596(0.0070) &0.0568(0.0060) &0.0985(0.0093)  \\

   \cmidrule(lr){3-12}
	 				&$\mathcal{D} (\hat{\mathbf{R}}, \mathbf{R})$ &50 &20  &20  &0.0811(0.0227)  &0.0574(0.0089) &0.0774(0.0215) &0.0577(0.0090) &0.0868(0.0121)   \\
	 				                                           &  &   &50  &50  &0.0384(0.0041)  &0.0351(0.0032) &0.0374(0.0039) &0.0350(0.0032) &0.0614(0.0050)  \\
	 				\cmidrule(lr){3-12}
	 				&$\mathcal{D} (\hat{\mathbf{R}}, \mathbf{R})$ &100 &20  &20 &0.0654(0.0208)& 0.0405(0.0065)& 0.0619(0.0202)& 0.0406(0.0067)& 0.0609(0.0088)\\
	 				                                           &  &   &50  &50  &0.0279(0.0032)& 0.0247(0.0021)& 0.0271(0.0031)& 0.0246(0.0022)& 0.0430(0.0031)   \\
	 				
                     \hline
	 				 &&&& \multicolumn{5}{c}{\multirow{1}{*}{\textbf{$t_1$ Distribution}}}\\
  \cmidrule(lr){3-12}
	 				&$\mathcal{D} (\hat{\mathbf{R}}, \mathbf{R})$ &20 &20  &20  &0.1300(0.0357)&0.2268(0.1428)&0.4307(0.1770)&0.4239(0.1893)&0.1395(0.0213)\\
	 				                                           &  &   &50  &50  &0.0688(0.0080)  &0.1336(0.1010) &0.2870(0.1799) &0.2817(0.1866) &0.0992(0.0094)  \\

   \cmidrule(lr){3-12}
	 				&$\mathcal{D} (\hat{\mathbf{R}}, \mathbf{R})$ &50 &20  &20  &0.0880(0.0249)  &0.1881(0.1446) &0.4173(0.1926) &0.4041(0.2075) &0.0861(0.0117)  \\
	 				                                           &  &   &50  &50  &0.0427(0.0045)  &0.1068(0.0814) &0.2672(0.1749) &0.2642(0.1828) &0.0611(0.0047)  \\
	 				\cmidrule(lr){3-12}
	 				&$\mathcal{D} (\hat{\mathbf{R}}, \mathbf{R})$ &100 &20  &20 & 0.0693(0.0239) &0.1781(0.1475)& 0.4319(0.1912)& 0.4251(0.2056)& 0.0609(0.0088)  \\
	 				                                           &  &   &50  &50  &0.0309(0.0037)& 0.0946(0.0846)& 0.2656(0.1723)& 0.2627(0.1797)& 0.0430(0.0033)   \\
	 				
                     \hline
	 				&&&& \multicolumn{5}{c}{\multirow{1}{*}{\textbf{$t_2$ Distribution}}}\\
	 				\cmidrule(lr){3-12}
	 				&$\mathcal{D} (\hat{\mathbf{R}}, \mathbf{R})$ &20 &20  &20  &0.1257(0.0349)&0.1277(0.0517)&0.2645(0.1507)&0.2429(0.1592
                     )&0.1386(0.0208)\\
	 				                                           &  &   &50  &50  &0.0669(0.0076)  &0.0783(0.0179) &0.1472(0.1026) &0.1408(0.1028) &0.0992(0.0087)  \\

   \cmidrule(lr){3-12}
	 				&$\mathcal{D} (\hat{\mathbf{R}}, \mathbf{R})$ &50 &20  &20  &0.0852(0.0231)  &0.0851(0.0256) &0.2125(0.1403) &0.1942(0.1500) &0.0860(0.0118)  \\
	 				                                           &  &   &50  &50  &0.0417(0.0041)  &0.0530(0.0188 ) & 0.1205(0.0987)&0.1158(0.1009) &0.0611(0.0047)  \\
	 				\cmidrule(lr){3-12}
	 				&$\mathcal{D} (\hat{\mathbf{R}}, \mathbf{R})$ &100 &20  &20 &0.0706(0.0249) &0.0624(0.0193)& 0.1944(0.1321)& 0.1712(0.1380) & 0.0613(0.0095)   \\
	 				                                           &  &   &50  &50  &0.0302(0.0036)& 0.0402(0.0238)& 0.1011(0.0840)& 0.0966(0.0860)& 0.0429(0.0033)   \\
	 				
                     \hline
                     &&&& \multicolumn{5}{c}{\multirow{2}{*}{\textbf{$t_3$ Distribution}}}\\
                     \cmidrule(lr){3-12}
	 				&$\mathcal{D} (\hat{\mathbf{R}}, \mathbf{R})$ &20 &20  &20  &0.1272(0.0391)&0.1117(0.0229)&0.1935(0.0978)&0.1679(0.0946)&0.1410(0.0226)  \\
	 				                                           &  &   &50  &50  &0.0657(0.0083)  &0.0676(0.0100) &0.0974(0.0372) &0.0926(0.0365) & 0.0988(0.0092) \\

   \cmidrule(lr){3-12}
	 				&$\mathcal{D} (\hat{\mathbf{R}}, \mathbf{R})$ &50 &20  &20  &0.0840(0.0220)  &0.0711(0.0142) &0.1415(0.0814) &0.1186(0.0753) &0.0868(0.0125)  \\
	 				                                           &  &   &50  &50  &0.0411(0.0046)  &0.0439(0.0063) &0.0752(0.0427) &0.0718(0.0478) &0.0612(0.0049)  \\
	 				\cmidrule(lr){3-12}
	 				&$\mathcal{D} (\hat{\mathbf{R}}, \mathbf{R})$ &100 &20  &20 &0.0686(0.0232)& 0.0507(0.0111) &0.1153(0.0691)& 0.0951(0.0678)& 0.0605(0.0086)   \\
	 				                                           &  &   &50  &50  &0.0296(0.0034)& 0.0312(0.0044)& 0.0601(0.0405)& 0.0571(0.0431)& 0.0429(0.0032)   \\
	 				\bottomrule[2pt]
 				\end{tabular*}}
 			\end{threeparttable}
     \end{table}

\begin{table}[!h]
	 	\caption{Averaged estimation errors and standard errors of $\mathcal{D} (\hat{\mathbf{R}}, \mathbf{R})$ for Scenario B under joint normal distribution and joint $t$ distribution over 500 replications.}
	 	 \label{tab:main2}\renewcommand{\arraystretch}{1} \centering
	 	\selectfont
	 	\begin{threeparttable}
	 		 \scalebox{0.9}{\begin{tabular*}{18.5cm}{cccccccccccccccccccc}
	 				\toprule[2pt]
	 			    &\multirow{1}{*}{Evaluation}   &\multirow{1}{*}{$T$}
	 				&\multirow{1}{*}{$p_1$}        &\multirow{1}{*}{$p_2$}
                    &\multirow{1}{*}{MRTS}          &\multirow{1}{*}{RMFA}
                    &\multirow{1}{*}{$\alpha$-PCA} &\multirow{1}{*}{PE}
                    &\multirow{1}{*}{MPCA$_{F}$}   \cr
	 			    \cmidrule(lr){8-11} \\
	 				\midrule[1pt]
	 				
	 				 &&&& \multicolumn{5}{c}{\multirow{1}{*}{\textbf{Normal Distribution}}}\\
	 \cmidrule(lr){3-12}
	 				&$\mathcal{D} (\hat{\mathbf{R}}, \mathbf{R})$ &20 &20  &20 &0.1194(0.0358)  &0.0930(0.0164) &0.1135(0.0348) &0.0936(0.0168) &0.1412(0.0219)\\
	 				                                           &  &   &50  &50  &0.0627(0.0076)  &0.0574(0.0060) &0.0602(0.0071) &0.0574(0.0060) & 0.0988(0.0095) \\

   \cmidrule(lr){3-12}
	 				&$\mathcal{D} (\hat{\mathbf{R}}, \mathbf{R})$ &50 &20  &20  &0.0816(0.0230)  &0.0580(0.0090) &0.0779(0.0217) &0.0583(0.0091) &0.0872(0.0122)  \\
	 				                                           &  &   &50  &50  &0.0388(0.0041)  &0.0354(0.0032) &0.0377(0.0040) &0.0354(0.0032) &0.0617(0.0049)  \\
	 			\cmidrule(lr){3-12}
	 				&$\mathcal{D} (\hat{\mathbf{R}}, \mathbf{R})$ &100 &20  &20 &0.0656(0.0207)& 0.0409(0.0065)& 0.0621(0.0202)& 0.0410(0.0067)& 0.0612(0.0087)   \\
	 				                                           &  &   &50  &50  &0.0281(0.0033)& 0.0249(0.0022)& 0.0274(0.0031)& 0.0249(0.0022)& 0.0431(0.0031)   \\	
	 				
                     \hline
	 				 &&&& \multicolumn{5}{c}{\multirow{1}{*}{\textbf{$t_1$ Distribution}}}\\
  \cmidrule(lr){3-12}
	 				&$\mathcal{D} (\hat{\mathbf{R}}, \mathbf{R})$ &20 &20  &20  &0.1445(0.0426)  &0.2365(0.1476) &0.4310(0.1766) &0.4242(0.1885) &0.1520(0.0247)\\
	 				                                           &  &   &50  &50  &0.0769(0.0111)  &0.1391(0.1041) &0.2867(0.1795) &0.2813(0.1861) & 0.1061(0.0121) \\

   \cmidrule(lr){3-12}
	 				&$\mathcal{D} (\hat{\mathbf{R}}, \mathbf{R})$ &50 &20  &20  &0.0968(0.0274)  &0.1959(0.1471) &0.4178(0.1927) &0.4049(0.2077) &0.0924(0.0131)  \\
	 				                                           &  &   &50  &50  &0.0484(0.0058)  &0.1123(0.0874) &0.2676(0.1750) &0.2647(0.1830) &0.0651(0.0056)  \\
	 				
	 				\cmidrule(lr){3-12}
	 				&$\mathcal{D} (\hat{\mathbf{R}}, \mathbf{R})$ &100 &20  &20 &0.0759(0.0269)& 0.1857(0.1503)& 0.4320(0.1912)& 0.4252(0.2056)& 0.0658(0.0093)   \\
	 				                                           &  &   &50  &50  & 0.0347(0.0040)& 0.0995(0.0899)& 0.2656(0.1719)& 0.2628(0.1796)& 0.0457(0.0037)  \\
                     \hline
	 				&&&& \multicolumn{5}{c}{\multirow{1}{*}{\textbf{$t_2$ Distribution}}}\\
	 				\cmidrule(lr){3-12}
	 				&$\mathcal{D} (\hat{\mathbf{R}}, \mathbf{R})$ &20 &20  &20  &0.1305(0.0371)  &0.1312(0.0568) &0.2648(0.1503) &0.2435(0.1597) &0.1416(0.0217)\\
	 				                                           &  &   &50  &50  &0.0691(0.0083)  &0.0801(0.0197) &0.1476(0.1027) &0.1412(0.1028) &0.1005(0.0094)  \\

   \cmidrule(lr){3-12}
	 				&$\mathcal{D} (\hat{\mathbf{R}}, \mathbf{R})$ &50 &20  &20  &0.0876(0.0236)  &0.0875(0.0275) &0.2125(0.1398) &0.1946(0.1502) &0.0877(0.0122)  \\
	 				                                           &  &   &50  &50  &0.0432(0.0044)  &0.0545(0.0206) &0.1209(0.0992) &0.1161(0.1013) &0.0622(0.0050)  \\
	 				\cmidrule(lr){3-12}
	 				&$\mathcal{D} (\hat{\mathbf{R}}, \mathbf{R})$ &100 &20  &20 & 0.0721(0.0255)& 0.0643(0.0206)& 0.1946(0.1322)& 0.1712(0.1378)& 0.0624(0.0096)  \\
	 				                                           &  &   &50  &50  &0.0313(0.0037)& 0.0415(0.0255)& 0.1012(0.0840)& 0.0966(0.0860)& 0.0436(0.0034)   \\
	 				
                     \hline
                     &&&& \multicolumn{5}{c}{\multirow{2}{*}{\textbf{ $t_3$ Distribution}}}\\
                     \cmidrule(lr){3-12}
	 				&$\mathcal{D} (\hat{\mathbf{R}}, \mathbf{R})$ &20 &20  &20  &0.1295(0.0401)  &0.1135(0.0237) &0.1941(0.0983) &0.1685(0.0945) &0.1436(0.0233)\\
	 				                                           &  &   &50  &50  &0.0669(0.0086)  &0.0686(0.0104) &0.0978(0.0371) &0.0930(0.0361) &0.0999(0.0097)  \\

   \cmidrule(lr){3-12}
	 				&$\mathcal{D} (\hat{\mathbf{R}}, \mathbf{R})$ &50 &20  &20  &0.0853(0.0227)  &0.0723(0.0150) &0.1420(0.0817) &0.1189(0.0752) &0.0877(0.0128) \\
	 				                                           &  &   &50  &50  &0.0419(0.0047)  &0.0447(0.0068) &0.0755(0.0431) &0.0723(0.0502) &0.0617(0.0051)  \\
	 				\cmidrule(lr){3-12}
	 				&$\mathcal{D} (\hat{\mathbf{R}}, \mathbf{R})$ &100 &20  &20 & 0.0694(0.0230)& 0.0516(0.0117)& 0.1155(0.0689)& 0.0952(0.0677)& 0.0615(0.0084)  \\
	 				                                           &  &   &50  &50  & 0.0302(0.0034)& 0.0318(0.0047)& 0.0602(0.0406)& 0.0572(0.0432)& 0.0434(0.0033)  \\
	 				\bottomrule[2pt]
 				\end{tabular*}}
 			\end{threeparttable}
     \end{table}

Tables \ref{tab:main1} and \ref{tab:main2} show the averaged estimation errors of $\mathcal{D} (\hat{\mathbf{R}}, \mathbf{R})$ with standard errors in parentheses  for joint normal distribution and joint $t$ distribution under Scenarios A and B, respectively. All methods benefit from large dimensions, and for the light-tailed normal distribution, the proposed MRTS method performs comparably with the $\alpha$-PCA method by \cite{fan2021}, PE method  by \cite{Yu2021Projected} and RMFA by \cite{He2021Matrix}, while much better the MPCA by \cite{Li2022Manifold} especially in small $T$ case due to faster convergence rate. What we want to emphasize is that the MRTS method shows great advantage over other methods under heavy-tailed $t$ distributions. This shows that the MRTS performs robustly and much better when the data are heavy-tailed, and performs comparably with others when data are light-tailed. By comparing the results in Table \ref{tab:main1} and Table \ref{tab:main2}, we can see that the proposed MRTS method is not sensitive to the weak temporal dependency of the factors and noises, although we need the independent condition for theoretical analysis.

\subsection{Estimation error for common components}
 In this section, we  compare the performances of the MRTS method with those of the RMFA method, the $\alpha$-PCA method, the PE method and the MPCA$_{F}$ in terms of estimating the common component matrices.
We evaluate the performance of different methods by  the Mean Squared Error, i.e.,
\[
\text{MSE}=\frac{1}{Tp_1p_2}\sum_{t=1}^T\|\hat\Sbb_t-\Sbb_t\|_F^2,
\]
where the $\hat\Sbb_t$ refers to an arbitrary estimate and  $\Sbb_t$ is the true common component matrix at time point $t$.

\begin{table}[!h]
	 	\caption{Mean squared error and its standard under Scenario A over 500 replications.}
	 	 \label{tab:main3}\renewcommand{\arraystretch}{1} \centering
	 	\selectfont
	 	
	 	\begin{threeparttable}
	 		 \scalebox{0.9}{\begin{tabular*}{17cm}{ccccccccccccccccccccccccccccccccccccccccccc}
	 				\toprule[2pt]
	 				&\multirow{2}{*}{Distribution}  &\multirow{2}{*}{$p_1$}
	 				&\multirow{2}{*}{MRTS}  &\multirow{2}{*}{RMFA}
                    &\multirow{2}{*}{$\alpha$-PCA}  &\multirow{2}{*}{PE}
                     &\multirow{2}{*}{MPCA$_{F}$}   \cr
	 				\cmidrule(lr){8-15} \\
	 				\midrule[1pt]
	 				 \multicolumn{8}{c}{\multirow{1}{*}{$T=20,p_2=p_1 $}}\\
	 				\cmidrule(lr){2-17}
	 				&Normal  &20   &0.0453(0.0065) &0.0369(0.0036) &0.0429(0.0060) &0.0371(0.0037) &0.0557(0.0067)  \\
	 				      &  &50  &0.0105(0.0008) &0.0094(0.0007) &0.0099(0.0008) &0.0094(0.0007) & 0.0211(0.0024) \\
	 				
	 				\cmidrule(lr){2-17}
	 				&$t_3$  &20   &0.1385(0.1263) &0.1477(0.2537) &0.3388(0.9222) &0.3224(0.9629) &0.1618(0.1434)\\
	 				&  &50    &0.0304(0.0207) &0.0357(0.0370) &0.0720(0.1233) &0.0704(0.1347) & 0.0580(0.0399) \\
	 			    \hline
	 				\multicolumn{8}{c}{\multirow{1}{*}{$T=50,p_2=p_1 $}}\\
	 				\cmidrule(lr){2-17}
	 				&Normal  &20   &0.0337(0.0042) &0.0283(0.0024) &0.0327(0.0040) &0.0283(0.0024) &0.0355(0.0032)  \\
	 				      &  &50  &0.0064(0.0004) &0.0059(0.0003) &0.0062(0.0004) &0.0059(0.0003) &0.0106(0.0008)  \\
	 				
	 				\cmidrule(lr){2-17}
	 				&$t_3$  &20   &0.1037(0.0803) &0.1009(0.1106) &0.2777(0.9210) &0.2654(0.9586) &0.1059(0.0839)\\
	 				&  &50    &0.0199(0.0147) &0.0243(0.0415) &0.0642(0.2261) &0.0624(0.2367) &0.0315(0.0228)  \\
                    \hline
	 			     \multicolumn{8}{c}{\multirow{1}{*}{$T=100,p_2=p_1 $}}\\
                     \cmidrule(lr){2-17}
                     &Normal  &20   &0.0300(0.0035) &0.0253(0.0017) &0.0292(0.0033) &0.0254(0.0017) &0.0289(0.0019)  \\
	 				      &  &50  &0.0051(0.0003) &0.0047(0.0002) &0.0050(0.0003) &0.0047(0.0002) &0.0071(0.0004)  \\
	 				
	 				\cmidrule(lr){2-17}
	 				&$t_3$  &20   &0.0897(0.0503) &0.0853(0.0906) &0.1964(0.6719) &0.1834(0.7119) &0.0855(0.0509)\\
	 				&  &50    &0.0161(0.0119) &0.0183(0.0246) &0.0515(0.1750) &0.0499(0.1815) &0.0217(0.0153)  \\
	 				\bottomrule[2pt]
	 		\end{tabular*}}
	 	\end{threeparttable}
	 	
	 \end{table}

Table \ref{tab:main3} shows that averaged MSEs with standard errors in parentheses under Scenario A for normal distribution and $t_3$ distribution.
From Table \ref{tab:main3}, we can see that our MRTS method performs very well in all settings. Firstly, the MRTS methods perform comparably with the PE method in the normal case, and  MRTS, RMFA and MPCA$_{F}$ perform better than the other methods in the heavy-tailed case. Especially when $p_{1}$ and $p_{2}$ are large, the advantages of MRTS over the other methods are more obvious.

\subsection{Estimating the numbers of factors}
In this section, we compare the empirical performances of the proposed MKER method with the $\alpha$-PCA based ER method ($\alpha$-PCA-ER) by \cite{fan2021}, the IterER method by \cite{Yu2021Projected}, the IC and ER method based on iTIPUP by \cite{han2022rank} (denote as iTIP-IC and iTIP-EC), the TCorTh by \cite{lam2021rank} and the Rit-ER method by \cite{He2021Matrix} in terms of estimating the pair of factor numbers.
\begin{table}[!h]
	 	\caption{The frequencies of exact estimation and underestimation of the numbers of factors under Scenario A over 500 replications. }
	 	 \label{tab:main4}\renewcommand{\arraystretch}{1} \centering
	 	\selectfont
	 	
	 	\begin{threeparttable}
	 		 \scalebox{0.8}{\begin{tabular*}{20cm}{cccccccccccccccccccccccccccccccccccc}
	 				\toprule[2pt]
	 		   		&\multirow{2}{*}{Distribution}  &\multirow{2}{*}{$p_1$}
                     &\multirow{2}{*}{MKER} & \multirow{2}{*}{Rit-ER}
	 				&\multirow{2}{*}{IterER}  &\multirow{2}{*}{$\alpha$-PCA-ER}
                    &\multirow{2}{*}{iTIP-IC}  &\multirow{2}{*}{iTIP-EC}
                     &\multirow{2}{*}{TCorTh}  \cr
	 				\cmidrule(lr){7-15} \\
	 				\midrule[1pt]
	 				 \multicolumn{10}{c}{\multirow{1}{*}{$T=20,p_2=p_1 $}}\\
	 				\cmidrule(lr){2-17}
	 				&Normal  &20  &0.676(0.046)&0.990(0.000)&0.990(0.000)&0.594(0.070)&0.000(1.000) &0.110(0.386) &0.680(0.048)  \\
	 				      &  &50  &1.000(0.000)&1.000(0.000) &1.000(0.000)&1.000(0.000)&0.000(1.000)&0.174(0.178) &1.000(0.000)  \\

	 				\cmidrule(lr){2-17}
	 				&$t_1$  &20   &0.620(0.066) &0.330(0.256) &0.226(0.470) &0.060(0.774) &0.090(0.370)&0.056(0.538)&0.272(0.380)\\
	 				&  &50    &0.998(0.000)&0.666(0.094) &0.532(0.304)&0.302(0.548)&0.146(0.478)&0.118(0.406)&0.700(0.162)  \\
	 				
	 				\cmidrule(lr){2-17}
                    &$t_2$  &20   &0.638(0.050)&0.706(0.068)&0.634(0.182)&0.228(0.464)&0.018(0.952)&0.078(0.452)&0.386(0.220)\\
	 				&  &50    &1.000(0.000) &0.906(0.002) & 0.880(0.066)&0.680(0.168)&0.012(0.970)&0.094(0.310)&0.928(0.028) \\
	 				\cmidrule(lr){2-17}
	 				&$t_3$  &20   &0.646(0.050)&0.870(0.016)&0.806(0.066)&0.360(0.276)&0.000(1.000)&0.090(0.430)&0.482(0.148)\\
	 				&  &50    &1.000(0.000) &0.906(0.002) &0.880(0.066)&0.680(0.168)&0.012(0.970) &0.094(0.310)&0.928(0.028) \\
	 			    \hline
	 				\multicolumn{10}{c}{\multirow{1}{*}{$T=50,p_2=p_1 $}}\\
	 				\cmidrule(lr){2-17}
	 				&Normal  &20   &0.772(0.014)&1.000(0.000) &1.000(0.000)&0.734(0.032)&0.000(1.000) &0.092(0.276)&0.954(0.000)  \\
	 				      &  &50  &1.000(0.000) &1.000(0.000) &1.000(0.000) &1.000(0.000)&0.000(1.000)&0.168(0.158) &1.000(0.000)  \\

	 				\cmidrule(lr){2-17}
	 				&$t_1$  &20   &0.742(0.012)&0.420(0.186) &0.292(0.466)&0.106(0.740) &0.088(0.194)&0.062(0.552)&0.556(0.054)\\
	 				&  &50    &1.000(0.000)&0.688(0.060) &0.560(0.294)&0.356(0.494) &0.176(0.266)&0.128(0.370)&0.698(0.014) \\
	 				
	 				\cmidrule(lr){2-17}
                    &$t_2$  &20   &0.796(0.020)&0.832(0.022) &0.738(0.124)&0.314(0.338) &0.014(0.946)&0.064(0.468)&0.786(0.014)\\
	 				&  &50    &1.000(0.000)&0.938(0.006) &0.916(0.058)&0.802(0.136) &0.006(0.974)&0.170(0.242)&0.932(0.004)  \\
	 				\cmidrule(lr){2-17}
	 				&$t_3$  &20   &0.760(0.018)&0.956(0.000)&0.914(0.032)&0.530(0.164)&0.000(1.000) &0.090(0.402)&0.876(0.012)\\
	 				&  &50    &1.000(0.000) &0.990(0.000) &0.988(0.006)&0.954(0.030)&0.000(1.000) &0.156(0.204)&0.992(0.000) \\
	 			 \hline
	 				\multicolumn{10}{c}{\multirow{1}{*}{$T=100,p_2=p_1 $}}\\
	 				\cmidrule(lr){2-17}
	 				&Normal  &20   &0.824(0.010)&1.000(0.000)&1.000(0.000)&0.786(0.016)&0.000(1.000)&0.076(0.334)&0.996(0.000)  \\
	 				      &  &50  &1.000(0.000)&1.000(0.000)&1.000(0.000)&1.000(0.000)&0.000(1.000)&0.186(0.130)&1.000(0.000)  \\

	 				\cmidrule(lr){2-17}
	 				&$t_1$  &20   &0.810(0.010)&0.402(0.202)&0.278(0.438)&0.092(0.724)&0.078(0.046)&0.052(0.560)&0.458(0.014)\\
	 				&  &50    &1.000(0.000)&0.702(0.044)&0.568(0.304)&0.348(0.514)&0.196(0.122)&0.134(0.366)&0.468(0.000) \\
	 				
	 				\cmidrule(lr){2-17}
                    &$t_2$  &20   &0.806(0.008)&0.874(0.012)&0.780(0.110)&0.400(0.294)&0.012(0.968)&0.072(0.376)&0.864(0.002)\\
	 				&  &50    &1.000(0.000)&0.964(0.000)&0.938(0.034)&0.852(0.094)&0.010(0.980)&0.166(0.254)&0.900(0.000) \\
	 				\cmidrule(lr){2-17}
	 				&$t_3$  &20   &0.818(0.008)&0.966(0.000)&0.950(0.014)&0.640(0.080)&0.000(1.000)&0.082(0.376)&0.960(0.000)\\
	 				&  &50    &1.000(0.000)&0.992(0.000)&0.992(0.006)&0.968(0.016)&0.000(1.000)&0.196(0.182)&0.980(0.000) \\
	 				\bottomrule[2pt]
	 		\end{tabular*}}
	 	\end{threeparttable}
	 	
	 \end{table}

Table \ref{tab:main4}  presents the frequencies of exact estimation and underestimation over 500 replications under Scenario A by different methods. We set $k_{\max}$ = 8. Under the normal case, we see that the Rit-ER and IterER perform comparably and both perform better than the others. However, as $p_{1}$ and $p_2$ increase, it can be seen that all the methods's performances get better. The proposed MKER method performs robustly and always performs the best for the heavy-tailed $t$ distribution cases. And it can also be seen that as the dimension of $p_1$ and $p_2$ increase, the proportion of exact estimation by MKER has the tendency to converge to 1, which is consistent with our theoretical analysis.

\section{Real Data Examples}
\subsection{COVID-CT dataset}
Coronavirus disease 2019 (COVID-19) is a contagious disease caused by a virus, the severe acute respiratory syndrome coronavirus 2 (SARS-CoV-2). The first known case was identified in Wuhan, China, in December 2019. The disease then spreads worldwide rapidly, leading to the COVID-19 pandemic and causing millions of deaths until now. It is critical to identify
 patients infected with COVID-19 and isolate them. Chest computed tomography (CT) scans can be used to screen the positives from the population.
We use the open-source chest CT dataset called COVID-CT which was ever analyzed in \cite{zhang2022low} and is available from \url{https://github.com/UCSD-AI4H/COVID-CT}.  The dataset is consisted of 2D grayscale images mainly in the format of PNG and JPEG, including 349 COVID-19 positive CT scans and 397 negative CT scans in the early pandemic. Figure \ref{cov:1} illustrates the original CT scans of one COVID-19 positive case (the left panel) and one negative case (the right panel).

\begin{figure}[!h]
  \centering
  \includegraphics[width=11cm, height=4.5cm]{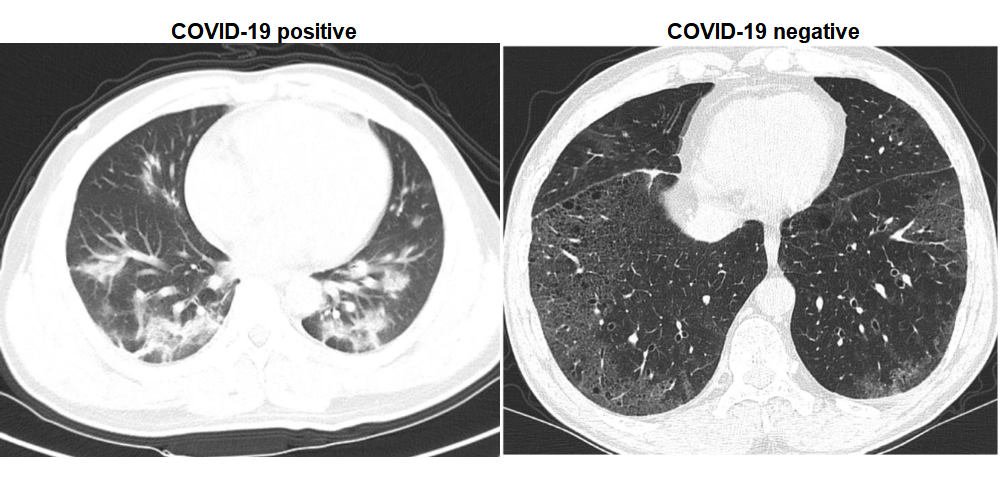}
 \caption{CT images of COVID-19 positive case and negative case.}\label{cov:1}
 \end{figure}

A 2D grayscale image can be represented as a matrix. The color value or gray value of each pixel on a grayscale image, refers to the color depth of a point in a black-and-white image, generally ranging from 0 to 255, with 255 for white and 0 for black. The brighter the original region is, the larger pixel value it has. By contrast, the darker the original region is, the smaller pixel value it has. Each pixel value can be rescaled by dividing 255, therefore a grayscale image is converted into a matrix with all  entries in the interval [0,1].
\begin{figure}[!h]
	  \centering

	   \subfigure[] {
	    \begin{minipage}[htbp]{0.4\textwidth}
	    	\centering
	    	\includegraphics[height=5cm,width=6.9cm]{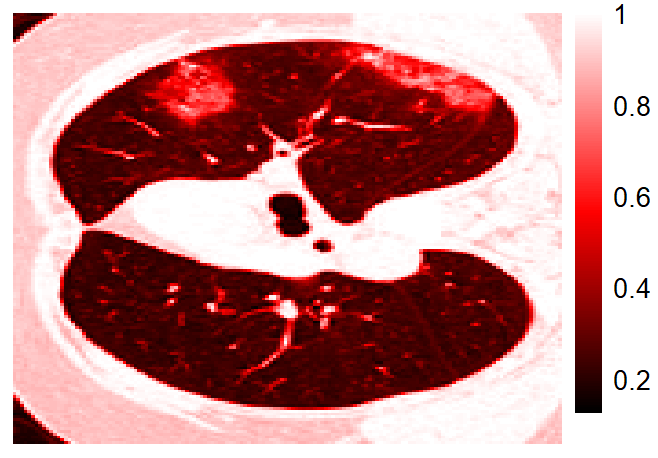}
	    \end{minipage}}
       \subfigure[] {
		\begin{minipage}[htbp]{0.4\textwidth}
			\centering
			\includegraphics[height=5cm,width=6.9cm]{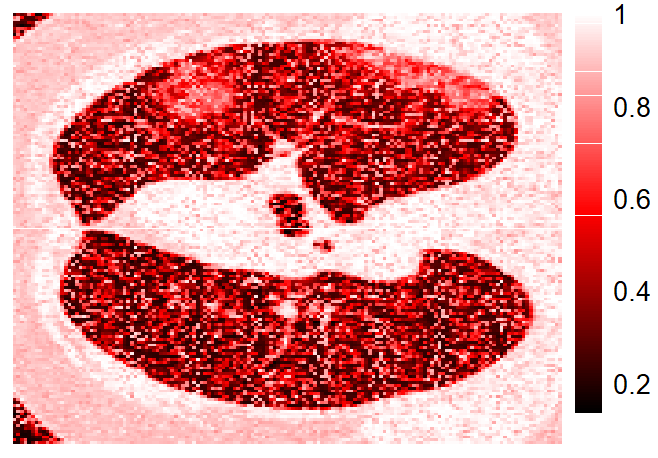}
		\end{minipage} }

      \subfigure[] {
		\begin{minipage}[htbp]{0.4\textwidth}
			\centering
			\includegraphics[height=5cm,width=6.9cm]{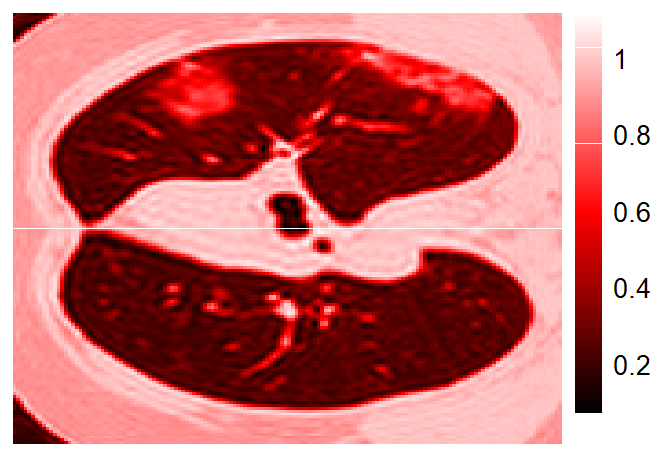}
		\end{minipage} }
     \subfigure[] {
		\begin{minipage}[htbp]{0.4\textwidth}
			\centering
			\includegraphics[height=5cm,width=6.9cm]{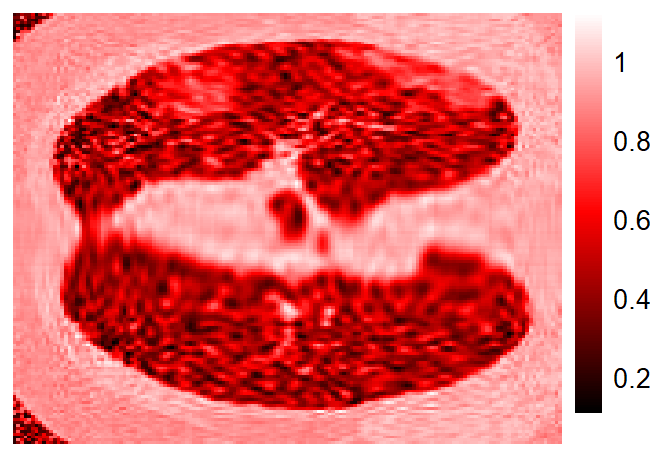}
		\end{minipage} }
		\caption{Four heatmaps of CT images. Subfigures (a) and (b) are down-sampled CT images of one positive case without/with contamination respectively. Subfigures (c) and (d) display the corresponding heatmaps of the reconstructed CT images using the estimated
common components by MRTS. }\label{commonrec}
	\end{figure}
We read these grayscale images by R package EBImage. Then we compress images into uniform dimensions (height,width)$=(150,150)$ by upsampling and downsampling. As a result, we have the dataset $\{Y_i,\Xb_i\}_{1 \leq i \leq 746}$, where $Y_i$ is the class label. Denote $Y_i=1$ as a positive case and $Y_i=0$ as a negative case, and $\Xb_i \in \RR^{150 \times 150}$ represents CT scan image of subject $i$. Then we construct the matrix  factor model and estimate the loading matrices as well as factor score matrices by different methods (e.g., MRTS) by setting $k_1=k_2=3$.  Then we get the dataset $\{Y_i,\vec(\widehat\Fb_i)\}_{1 \leq i \leq 746}$, where $\widehat\Fb_i$ is the estimated factor score matrix for subject $i$. We randomly select 70$\%$ of the samples as the training set and 30$\%$ of the samples as the testing set, and repeat this process 500 times.
In each replication we put the estimated factor scores into Support Vector Machine (SVM) for classification.
\begin{table*}[!h]
 \begin{center}
  \small
  \addtolength{\tabcolsep}{0pt}
  \caption{The means (standard errors) of classification metrics on real COVID-CT data with 500 replicates.}\label{classification-result}
   \renewcommand{\arraystretch}{1.5}
  \scalebox{0.7}{
    \begin{tabular*}{18.5cm}{ccccccc}
                 \toprule[1.2pt]
     &&MRTS&$\alpha$-PCA&PE&RMFA&MPCA$_F$\\

               \toprule[0.8pt]
               \multirow{3}*{Original datset}
               &Precision&\textbf{0.7335(0.0586)}&0.7243(0.0594)&0.7218(0.0615) &0.7224(0.0618)&0.7289(0.0598)  \\
                &AUC&\textbf{0.7299(0.0283)}&0.7242(0.0295)&0.7244(0.0297)&0.7260(0.0296)&0.7277(0.0289)    \\
                &Error Rate&\textbf{0.3646(0.0330)}&0.3750(0.0334)&0.3737(0.0331)&0.3730(0.0335)&0.3676(0.0340) \\
     \hline
     \multirow{3}*{Contaminated dataset}
               &Precision&\textbf{0.7244(0.0554)}&0.7120(0.0596)&0.7085(0.0597)&0.7084(0.0588)&0.7131(0.0580)\\
                &AUC&\textbf{0.7288(0.0287)}&0.7221(0.0293)&0.7213(0.0290)&0.7223(0.0290)&0.7225(0.0293) \\
                &Error Rate&\textbf{0.3653(0.0322)}&0.3769(0.0335)&0.3767(0.0339)&0.3764(0.0340)&0.3741(0.0343)\\
    \toprule[1.2pt]

  \end{tabular*}}
 \end{center}
\end{table*}


Meanwhile, to measure the robustness of different methods, we contaminated the original dataset in the following way.  We  replaced 30$\%$ entries of the matrix $\Xb_i$ with random numbers generated from uniform distribution $\mathcal{U}(0.8,1)$.   Subfigure (a) and Subfigure (b) in Figure \ref{commonrec} are heatmaps of CT images of one positive case without/with contamination respectively. In order to give a direct illustration of the effect of our MRTS method in extracting CT image information, we display the heatmaps of reconstructed CT images using the estimated common components in subfigures (c) and (d) of Figure \ref{commonrec}, from which
 we can see intuitively that the reconstructed image is very similar to the original image thereby indicating that the common components extracted by our method is a good approximation of the original image. The classification results are reported in Table \ref{classification-result}, from which we see that MRTS performs the best in terms of the metrics Precision, AUC and Error rate. MRTS, RMFA and MPCA$_F$ are all robust methods, and outperform the other two non-robust methods for both the original and contaminated  datasets.

\subsection{Portfolio returns dataset}
In this section, we illustrate the empirical performance of the MRTS method by analyzing a heavy-tailed financial portfolio dataset, which was also studied in \cite{wang2019factor}, \cite{Yu2021Projected}, and \cite{He2021Matrix}. The dataset contains monthly returns of 100 portfolios ranging from January
1964 to December 2019, which are  structured into a $10\times 10$ matrix at each time point,  with rows corresponding to 10 levels of market capital size (denoted as S1-S10) and columns corresponding to 10 levels of book-to-equity ratio (denoted as BE1-BE10). The dataset is available at the website \url{http://mba.tuck.dartmouth.edu/pages/faculty/ken.french/data_library.html}. We refer to \cite{He2021Matrix}  for further details on this dataset and the proprocessing procedure (stationarity and missing values).

\cite{He2021Matrix} showed the heavy-tailed property of the returns data  and concluded that robust analysis methods are more appropriate. Moreover, \cite{He2021Matrix} showed that it is reasonable to adopt a matrix factor model for this real dataset by using the test proposed in \cite{He2021Vector}.

For the preprocessed monthly returns dataset, MKER method
suggests that $(k_1,k_2) = (1,1)$ while the Rit-ER suggests that $(k_1,k_2) = (1,2)$. The difference between the estimates by different methods may be explained by the heavy-tailedness of the returns data. As overestimation is better than underestimation, and for better illustration, we take $(k_1,k_2) = (1,2)$.

\begin{table*}[!h]
 	\begin{center}
 		\small
 		\addtolength{\tabcolsep}{1pt}
 		\caption{Loading matrices for Fama--French data set, after varimax rotation and scaling by 30.}\label{tab5}
 		 \renewcommand{\arraystretch}{1}
 		\scalebox{1}{ 		 \begin{tabular*}{14.5cm}{cc|cccccccccc}
 				\toprule[1.2pt]
 				 \multicolumn{12}{l}{Size}\\
 				\toprule[1.2pt]
 				 Method&Factor&S1&S2&S3&S4&S5&S6&S7&S8&S9&S10\\\hline
 \multirow{2}*{MRTS}&1&\cellcolor {Lavender}-41&\cellcolor {Lavender}-42&\cellcolor {Lavender}-41&\cellcolor {Lavender}-38&\cellcolor {Lavender}-32&\cellcolor {Lavender}-23&\cellcolor {Lavender}-14&-5&8&\cellcolor {Lavender}24
 	\\
 	&2&\cellcolor {Lavender}11&6&-3&-6&\cellcolor {Lavender}-15&\cellcolor {Lavender}-26&\cellcolor {Lavender}-35&\cellcolor {Lavender}-45&\cellcolor {Lavender}-53&\cellcolor {Lavender}-44
 	\\\hline
 				 \multirow{2}*{RMFA}&1&\cellcolor {Lavender}-16&\cellcolor {Lavender}-15&\cellcolor {Lavender}-13&\cellcolor {Lavender}-11&-8&-6&-3&0&4&6\\
 				 &2&-6&-2&2&5&8&\cellcolor {Lavender}10&\cellcolor {Lavender}12&\cellcolor {Lavender}14&\cellcolor {Lavender}15&\cellcolor {Lavender}10\\\hline
 				 \multirow{2}*{PE}&1&\cellcolor {Lavender}-16&\cellcolor {Lavender}-15&\cellcolor {Lavender}-12&\cellcolor {Lavender}-10&-8&-5&-3&-1&4&7
 				\\
 				 &2&-6&-1&3&5&8&\cellcolor {Lavender}11&\cellcolor {Lavender}12&\cellcolor {Lavender}13&\cellcolor {Lavender}15&\cellcolor {Lavender}10
 				\\\hline
 	 			 \multirow{2}*{$\alpha$-PCA}&1&\cellcolor {Lavender}-14&\cellcolor {Lavender}-14&\cellcolor {Lavender}-13&\cellcolor {Lavender}-11&\cellcolor {Lavender}-9&-7&-4&-2&3&7
 	\\
 	&2&-4&-2&1&3&6&\cellcolor {Lavender}9&\cellcolor {Lavender}12&\cellcolor {Lavender}13&\cellcolor {Lavender}16&\cellcolor {Lavender}14
 	\\\hline
 \multirow{2}*{MPCA$_{F}$}&1&\cellcolor {Lavender}-47&\cellcolor {Lavender}-46&\cellcolor {Lavender}-40&\cellcolor {Lavender}-34&\cellcolor {Lavender}-29&\cellcolor {Lavender}-22&\cellcolor {Lavender}-12&-5&7&\cellcolor {Lavender}20
 	\\
 	&2&\cellcolor {Lavender}19&\cellcolor {Lavender}9&-6&\cellcolor {Lavender}-11&\cellcolor {Lavender}-20&\cellcolor {Lavender}-27&\cellcolor {Lavender}-36&\cellcolor {Lavender}-44&\cellcolor {Lavender}-51&\cellcolor {Lavender}-38
 \\

 				\bottomrule[1.2pt]		
\multicolumn{12}{l}{Book-to-Equity}\\
 				\toprule[1.2pt]
 				 Method&Factor&BE1&BE2&BE3&BE4&BE5&BE6&BE7&BE8&BE9&BE10\\\hline
 \multirow{2}*{MRTS}&1&\cellcolor {Lavender}17&8&-7&\cellcolor {Lavender}-18&\cellcolor {Lavender}-27&\cellcolor {Lavender}-35&\cellcolor {Lavender}-39&\cellcolor {Lavender}-41&\cellcolor {Lavender}-40&\cellcolor {Lavender}-37\\
 				 &2&\cellcolor {Lavender}53&\cellcolor {Lavender}55&\cellcolor {Lavender}41&\cellcolor {Lavender}30&\cellcolor {Lavender}19&\cellcolor {Lavender}10&0&-3&-3&-4\\\hline
 				 \multirow{2}*{RMFA}&1&6&1&-3&-6&\cellcolor {Lavender}-9&\cellcolor {Lavender}-11&\cellcolor {Lavender}-12&\cellcolor {Lavender}-13&\cellcolor {Lavender}-12&\cellcolor {Lavender}-11
 				\\
 				&2&\cellcolor {Lavender}19&\cellcolor {Lavender}17&\cellcolor {Lavender}12&\cellcolor {Lavender}9&5&3&0&-1&-1&0\\\hline
 				 \multirow{2}*{PE}&1&6&1&-4&-7&\cellcolor {Lavender}-10&\cellcolor {Lavender}-11&\cellcolor {Lavender}-12&\cellcolor {Lavender}-12&\cellcolor {Lavender}-12&\cellcolor {Lavender}-10
\\
&2&\cellcolor {Lavender}20&\cellcolor {Lavender}17&\cellcolor {Lavender}11&8&4&2&0&-1&-1&0
\\\hline
 				 \multirow{2}*{$\alpha$-PCA}&1&6&2&-4&-7&\cellcolor {Lavender}-10&\cellcolor {Lavender}-11&\cellcolor {Lavender}-12&\cellcolor {Lavender}-13&\cellcolor {Lavender}-12&\cellcolor {Lavender}-11
\\
&2&\cellcolor {Lavender}19&\cellcolor {Lavender}18&\cellcolor {Lavender}12&8&4&2&0&-1&-1&-1
\\\hline
\multirow{2}*{MPCA$_{F}$}&1&\cellcolor {Lavender}19&7&\cellcolor {Lavender}-9&\cellcolor {Lavender}-18&\cellcolor {Lavender}-28&\cellcolor {Lavender}-35&\cellcolor {Lavender}-38&\cellcolor {Lavender}-41&\cellcolor {Lavender}-41&\cellcolor {Lavender}-37
\\
&2&\cellcolor {Lavender}58&\cellcolor {Lavender}54&\cellcolor {Lavender}38&\cellcolor {Lavender}29&\cellcolor {Lavender}18&\cellcolor {Lavender}11&1&-2&-3&-4
\\

 				\bottomrule[1.2pt]		
 		\end{tabular*}}		
 	\end{center}
 \end{table*}

The estimated front and back loading matrices after varimax rotation and scaling are reported in Table \ref{tab5}, from which  we observe that the PE, $\alpha$-PCA  lead to very similar estimated loadings while the robust methods MRTS, RMFA, MPCA$_F$ lead to very similar estimated loadings.
To  compare the matrix factor analysis methods, we also adopt a rolling-validation procedure as in \cite{wang2019factor} and \cite{He2021Matrix}. For each year $t$ from 1996 to 2019, we recursively  use the $n$ (bandwidth) half years observations before year $t$ to train the matrix factor model, then we estimate the row/column loading matrices. The loadings are then used to train regression model to derive the  factor scores and the corresponding residuals of the 12 months in the current year. In detail, let $\Yb_t^i$ and $\hat\Yb_t^i$ be the observed and estimated price matrix of month $i$ in year $t$, and let $\bar {\Yb}_t$ be the mean price matrix. Further define
\[
\text{MSE}_t=\frac{1}{12\times10\times10}\sum_{i=1}^{12}\|\hat\Yb_t^i-\Yb_t^i\|_F^2~~\mbox{ and }~~ \rho_t=\frac{\sum_{i=1}^{12}\|\hat\Yb_t^i-\Yb_t^i\|_F^2}{\sum_{i=1}^{12}\|\Yb_t^i-\bar\Yb_t\|_F^2}
\]
as the mean squared pricing error and unexplained proportion of total variances, respectively.  For  year $t$ in the rolling-validation procedure, we measure the variation of loading space  by $v_t:=\mathcal{D}(\hat\Cb_t\otimes \hat\Rb_t,\hat\Cb_{t-1}\otimes \hat\Rb_{t-1})$.

We report the results of the means of MSE, $\rho$ and $v$ by RMFA, PE, ACCE and  $\alpha$-PCA  in Table \ref{tab6}. Firstly, we see that in all combinations considered, the robust methods MRTS, RMFA, MPCA$_F$ methods always have  lower pricing errors than the non-robust methods PE and $\alpha$-PCA. In terms of estimating  factor loading spaces,  MRTS, RMFA,  MPCA$_F$ methods  always perform more stably than the non-robust methods.
Since heavy tailedness is a well-known stylized fact of financial returns and stock-level predictor variables, the robust MRTS method is strongly recommended to cope with matrix-variate financial data.
\begin{table*}[!h]
	\begin{center}
		\small
		\addtolength{\tabcolsep}{0pt}
		\caption{Rolling validation with $k_1=1$ and $k_2=2$ for the Fama--French portfolios, the sample size of training set is $6n$. $\overline{MSE}$, $\bar \rho$, $\bar v$ are the mean pricing error, mean unexplained proportion of total variances and mean variation of the estimated loading space. }\label{tab6}
		 \renewcommand{\arraystretch}{1.5}
		\scalebox{0.7}{ 	
			 \begin{tabular*}{9cm}{ccccccccc}
				\toprule[1.2pt]
           \multirow{2}{*}&$\overline{MSE}$\\
				       &n&MRTS&$\alpha$-PCA&PE&RMFA&MPCA$_F$\\
\hline
&17&0.745& 0.749&0.746&0.744&\textbf{0.741}
\\
&19&0.747&0.751&0.748&0.746&\textbf{0.743}

\\
&21&0.747&0.751&0.748&0.746&\textbf{0.743}
\\

                  \toprule[0.8pt]

				\multirow{2}{*}&$\bar \rho$\\
				       &n&MRTS&$\alpha$-PCA&PE&RMFA&MPCA$_F$\\
\hline
&17&0.710&0.715&0.710&0.706&\textbf{0.701}
\\
&19&0.711&0.717&0.711&0.707&\textbf{0.701}

\\
&21&0.711&0.718&0.711&0.707&\textbf{0.701}
\\
                 \toprule[0.8pt]

				\multirow{2}{*}&$\bar v$\\
				       &n&MRTS&$\alpha$-PCA&PE&RMFA&MPCA$_F$\\

\hline
&17&0.093&0.182&0.094&0.086&\textbf{0.079}
\\
&19&0.079&0.177&0.082&0.073&\textbf{0.067}
\\
&21&0.073&0.172&0.073&0.067&\textbf{0.062}
	
\\
				
				\bottomrule[1.2pt]		
		\end{tabular*}}		
	\end{center}
\end{table*}

\section{Conclusion and Discussion}
In this paper we  propose a new type of Kendall's tau for random matrices, named as matrix Kendall's tau. We show that the row/column matrix Kendall's tau share the same eigenspace with the row/column scatter matrix for matrix-elliptical distribution, with the same descending order of the eigenvalues. The sample version of the row/column matrix Kendall's tau is a U-statistic with a bounded kernel (under operator norm) and enjoys the same distribution-free property as multivariate Kendall's tau. 
Secondly, we propose a Matrix-type Robust Two Step (MRTS) method to estimate the loading and factor spaces for MEFM by matrix Kendall's tau, which achieves faster convergence rates than the Manifold Principal Component Analysis (MPCA) for estimating the loading spaces. We also propose robust and consistent MKER estimators for the pair of factor numbers by exploiting the eigenvalue-ratios of the sample matrix Kendall's tau. Numerical results show that the proposed MRTS/MKER methods outperform the existing methods particularly in heavy-tailed cases. The current work can be extended along several directions. Firstly, the matrix Kendall's tau is a useful tool for matrix elliptical distribution, and provides a robust approach to  perform 2-Dimensional Principal Component Analysis (2D-PCA) and to estimate separable covariance matrices in high dimensions. We leave these extensions as future work. Secondly, for matrix factor model, the projection technique is attractive as it increases the signal-to-noise ratio and leads to more accurate estimators. The projection technique can be incorporated in the current framework. In detail, given observations $\{\Xb_t\}$,  we first obtain a projection matrix $\widehat\Cb$ that consists of the leading $k_2$
eigenvectors of $\widehat \Kb_c^X$, multiplied by $\sqrt{p_2}$. Then we project the original matrix observations to a lower dimensional space, i.e., $\Yb_t=\Xb_t\widehat\Cb/p_2$. We further construct row matrix Kendall's tau $\widehat \Kb_r^Y$ based on $\Yb_t$'s, and obtain a projected estimator $\widetilde \Rb$ by the leading $k_1$
eigenvectors of $\widehat \Kb_c^Y$, multiplied by $\sqrt{p_1}$. The projected estimator $\widetilde \Cb$ can be obtained by a similar procedure applied to matrix observations $\{\Xb_t^\top\}$.
The theoretical analysis is more challenging and will be left for future work.

\section*{Acknowledgements}

He's work is supported by  National Science Foundation (NSF) of  China (12171282,11801316), National Statistical Scientific Research Key Project (2021LZ09), Young Scholars Program of Shandong University, Project funded by
China Postdoctoral Science Foundation (2021M701997).

\bibliographystyle{model2-names}
\bibliography{Ref}

\setlength{\bibsep}{1pt}

\renewcommand{\baselinestretch}{1}
\setcounter{footnote}{0}
\clearpage
\setcounter{page}{1}
\title{
	\begin{center}
		\Large Supplementary Materials for ``Matrix Kendall's tau in High-dimensions: A  Robust Statistic for  Matrix  Factor Model"
	\end{center}
}
\date{}
\begin{center}
		
	\author{
	Yong He\footnotemark[1],~
	Yalin Wang\footnotemark[1]
,~Long Yu\footnotemark[2],~
	Wang Zhou\footnotemark[3],~
	Wen-Xin Zhou\footnotemark[4]
	}
\renewcommand{\thefootnote}{\fnsymbol{footnote}}
\footnotetext[1]{Institute of Financial Studies, Shandong University, China; e-mail:{\tt heyong@sdu.edu.cn}}
\footnotetext[2]{Shanghai University of Finance and Economics, China. e-mail:{\tt : fduyulong@163.com}}
\footnotetext[3]{National University of Singapore, Singapore. e-mail:{\tt wangzhou@nus.edu.sg }}
\footnotetext[4]{University of California, San Diego, USA. e-mail:{\tt wez243@ucsd.edu}}

\end{center}
\maketitle
\appendix
This document provides detailed proofs and additional simulation results to the main paper. The notation $\Xb$ and $\tilde\Xb$ are used repeatedly in the proof, but can represent different matrices in different lemmas.

\section{Proof of Proposition}
\subsection*{Proof of Proposition \ref{pro1}}
\begin{proof} We only show the results for $\Kb_{r}$ and the results for $\Kb_{c}$ can be derived in a similar way.

By Lemma \ref{lem}, it is equivalent to consider $\Kb_{r}=\EE\left(\frac{(\Xb-\Mb)(\Xb-\Mb)^\top}{\norm{\Xb-\Mb}_{F}^{2}}\right)$. Let $\Rb=\big(\bu_{1}(\bSigma),\ldots,\bu_{p}(\bSigma)\big)$ be the eigenvector matrix of $\bSigma$. Then, we have
\begin{equation}
\Rb^\top \frac{\Xb-\Mb}{\norm{\Xb-\Mb}_{F}}=\frac{\Rb^\top(\Xb-\Mb)}{\norm{\Rb^\top(\Xb-\Mb)}_{F}}=\frac{\Zb}{\norm{\Zb}_F},\nonumber
\end{equation}
where $\Zb=\Db\Ub\Bb^\top$ with $\Db=\Big(\text{diag}\big(\sqrt{\lambda_1(\bSigma)},\ldots,\sqrt{\lambda_m(\bSigma)}\big),\mathbf{0}\Big)^\top \in \RR^{p \times m}$. Therefore,
\begin{equation}
\Kb_{r}=\EE\left(\frac{(\Xb-\Mb)(\Xb-\Mb)^\top}{\norm{\Xb-\Mb}_{F}^{2}}\right)=\Rb\EE\left(\frac{\Zb\Zb^\top}{\norm{\Zb}_{F}^2}\right) \Rb^\top. \nonumber
\end{equation}

Next, we prove that $\EE\left(\frac{\Zb\Zb^\top}{\norm{\Zb}_{F}^2}\right)$ is a diagonal matrix.
For any matrix $\Pb=$diag$(\bv)$, where $\bv=(v_{1},\ldots,v_{m})^\top$ satisfies that $v_{j}=$ 1 or $-1$ for $j=1,\ldots,m$, we have
\begin{equation}
\Pb \frac{\Zb}{\norm{\Zb}_{F}}=\frac{\Pb\Zb}{\norm{\Pb\Zb}_{F}} \overset{d}{=} \frac{\Zb}{\norm{\Zb}_{F}} \Rightarrow \EE\left(\frac{\Zb\Zb^\top}{\norm{\Zb}_{F}^2}\right)=\Pb \EE\left(\frac{\Zb\Zb^\top}{\norm{\Zb}_{F}^2}\right) \Pb. \nonumber
\end{equation}
This result holds if and only if $\EE\left(\frac{\Zb\Zb^\top}{\norm{\Zb}_{F}^2}\right)$ is a diagonal matrix. In other words, $\Kb_r$ shares the same eigenvector space as $\bSigma$.

The last step is to show that the diagonals of $\EE\left(\frac{\Zb\Zb^\top}{\norm{\Zb}_{F}^2}\right)$ are decreasing. Recall that $\Zb=\Db\Ub\Bb^\top$, so
\begin{equation}
\EE\left(\frac{\Zb\Zb^\top}{\norm{\Zb}_{F}^2}\right)=\EE\left(\frac{\Db\Ub\Bb^\top\Bb\Ub^\top\Db^\top}{\norm{\Db\Ub\Bb^\top}_{F}^2}\right)=\EE \left( \frac{\Db\Ub\bOmega^{*}\Ub^\top\Db^\top}{\norm{\Db\Ub\Bb^\top}_{F}^2}\right), \nonumber
\end{equation}

\begin{equation}
\Rightarrow\begin{aligned}
\left[\EE\left(\frac{\Zb\Zb^\top}{\norm{\Zb}_{F}^2}\right)\right]_{jj}=\EE\left(\frac{\lambda_{j}(\bSigma)\Ub_{j,\cdot}\bOmega^{*}\Ub_{j,\cdot}^\top}{\lambda_{1}(\bSigma)\Ub_{1,\cdot}\bOmega^{*}\Ub_{1,\cdot}^\top+\cdots+\lambda_{m}(\bSigma^{*})\Ub_{m,\cdot}\bOmega^{*}\Ub_{m,\cdot}^\top}\right).\nonumber
\end{aligned}
\end{equation}
without loss of generality, we assume $\bOmega^{*}=\Bb^\top\Bb$ to be a diagonal matrix.
Then, for any $k \textless s$,
\begin{equation}
\frac{\left[\EE\left(\frac{\Zb\Zb^\top}{\norm{\Zb}_{F}^2}\right)\right]_{kk}}{\left[\EE\left(\frac{\Zb\Zb^\top}{\norm{\Zb}_{F}^2}\right)\right]_{ss}}=\frac{\EE \frac{\lambda_{k}(\bSigma)\Ub_{k,\cdot}\bOmega^{*}\Ub_{k,\cdot}^\top}{\lambda_{k}(\bSigma)E_{1}+\lambda_{s}(\bSigma)E_{2}+E_3}}{\EE \frac{\lambda_{s}(\bSigma)\Ub_{s,\cdot}\bOmega^{*}\Ub_{s,\cdot}^\top}{\lambda_{k}(\bSigma)E_{1}+\lambda_{s}(\bSigma)E_{2}+E_3}} \textless
\frac{\EE \frac{\lambda_{k}(\bSigma)\Ub_{k,\cdot}\bOmega^{*}\Ub_{k,\cdot}^\top}{\lambda_{k}(\bSigma)E_{1}+\lambda_{k}(\bSigma)E_{2}+E_3}}{\EE \frac{\lambda_{s}(\bSigma)\Ub_{s,\cdot}\bOmega^{*}\Ub_{s,\cdot}^\top}{\lambda_{s}(\bSigma)E_{1}+\lambda_{s}(\bSigma)E_{2}+E_3}}=
\frac{\EE \frac{\Ub_{k,\cdot}\bOmega^{*}\Ub_{k,\cdot}^\top}{E_{1}+E_{2}+E_3/\lambda_{k}(\bSigma)}}{\EE \frac{\Ub_{s,\cdot}\bOmega^{*}\Ub_{s,\cdot}^\top}{E_{1}+E_{2}+E_3/\lambda_{s}(\bSigma)}} \textless 1 .\nonumber
\end{equation}
where  $E_1=\Ub_{k,\cdot}\bOmega^{*}\Ub_{k,\cdot}^\top$, $E_2=\Ub_{s,\cdot}\bOmega^{*}\Ub_{s,\cdot}^\top$, $E_3=\sum_{t\notin \left\{k,s\right\}}\lambda_{t}(\bSigma)\Ub_{t,\cdot}\bOmega^{*}\Ub_{t,\cdot}^\top$. This completes the proof.
\end{proof}

\section{Proof of the Main Theorems}

\subsection*{Proof of Theorem \ref{th1}: $\hat{\Rb}$ and $\hat{\Cb}$ converge in Frobenius norm}
\begin{proof}
Define $\hat{\bLambda}$ as the diagonal matrix composed of the leading $k_1$ eigenvalues of $\hat{\Kb}_r$. Lemma \ref{eigenvalue} implies that $\hat{\bLambda}$ is asymptotically invertible and $\norm{\hat{\bLambda}^{-1}}_F=O_p(1)$. Because $\hat{\Rb}/\sqrt{p_1}$ is composed of the leading eigenvectors of $\hat{\Kb}_r$, we have
\begin{equation}\label{Kr definition}
\hat{\Kb}_r\hat{\Rb}=\hat{\Rb}\hat{\bLambda}.
\end{equation}
Expand $\hat{\Kb}_r$ by its definition
\begin{equation}
\begin{aligned}
\hat{\Kb}_r&=\frac{2}{T(T-1)}\sum_{1 \leq t \textless s \leq T}\frac{(\Xb_t-\Xb_s)(\Xb_t-\Xb_s)^\top}{\norm{\Xb_t-\Xb_s}_F^2}\\
&=\frac{2}{T(T-1)}\sum_{1 \leq t \textless s \leq T}\frac{\big(\Rb(\Fb_t-\Fb_s)\Cb^\top+(\Eb_t-\Eb_s)\big) \big(\Rb(\Fb_t-\Fb_s)\Cb^\top+(\Eb_t-\Eb_s)\big)^\top}{\norm{\Xb_t-\Xb_s}_F^2}.\nonumber
\end{aligned}
\end{equation}
Denote
$$
\begin{aligned}
&\Mb_1=\frac{2}{T(T-1)}\sum_{1 \leq t \textless s \leq T}\frac{(\Fb_t-\Fb_s)\Cb^\top\Cb(\Fb_t-\Fb_s)^\top}{\norm{\Xb_t-\Xb_s}_F^2},\\
&\Mb_2=\frac{2}{T(T-1)}\sum_{1 \leq t \textless s \leq T}\frac{(\Eb_t-\Eb_s)\Cb(\Fb_t-\Fb_s)^\top}{\norm{\Xb_t-\Xb_s}_F^2},\\
&\Mb_3=\frac{2}{T(T-1)}\sum_{1 \leq t \textless s \leq T}\frac{(\Fb_t-\Fb_s)\Cb^\top(\Eb_t-\Eb_s)^\top}{\norm{\Xb_t-\Xb_s}_F^2},\\
&\Mb_4=\frac{2}{T(T-1)}\sum_{1 \leq t \textless s \leq T}\frac{(\Eb_t-\Eb_s)(\Eb_t-\Eb_s)^\top}{\norm{\Xb_t-\Xb_s}_F^2}.
\end{aligned}
$$
Then,
$$\hat{\Kb}_r=\Rb\Mb_1\Rb^\top+\Mb_2\Rb^\top+\Rb\Mb_3+\Mb_4.$$
According to equation (\ref{Kr definition}), we have
$$\hat{\Rb}=(\Rb\Mb_1\Rb^\top+\Mb_2\Rb^\top+\Rb\Mb_3+\Mb_4)\hat{\Rb}\hat{\bLambda}^{-1}.$$
Let $\hat{\Hb}_R=\Mb_1\Rb^\top\hat{\Rb}\hat{\bLambda}^{-1}$, so
\begin{equation} \label{R-RH}
\hat{\Rb}-\Rb\hat{\Hb}_R=(\Mb_2\Rb^\top+\Rb\Mb_3+\Mb_4)\hat{\Rb}\hat{\bLambda}^{-1}.
\end{equation}
Lemma \ref{m1} shows that $\|\hat\Hb_R\|_F^2\le O_p(1)$.
Lemmas \ref{m2m3} and \ref{m4} show that
$$ \norm{\Mb_2}_F^2=O_p\Big(\frac{1}{Tp_1p_2}+\frac{1}{p_1^2p_2^2}\Big), $$
$$\norm{\Mb_3}_F^2=O_p\Big(\frac{1}{Tp_1p_2}+\frac{1}{p_1^2p_2^2}\Big),$$
while
\begin{equation}
\frac{1}{p_1}\norm{\Mb_4\hat{\Rb}}_F^2=O_p\Big(\frac{1}{Tp_2}+\frac{1}{p_1^2}\Big)+o_p(1) \times \frac{1}{p_1}\norm{\hat{\Rb}-\Rb\hat{\Hb}_R}_F^2.\nonumber
\end{equation}
Therefore, it is easy to prove that
\begin{equation}
\frac{1}{p_1}\norm{\hat\Rb-\Rb\hat{\Hb}_{R}}_F^2=O_p\Big(\frac{1}{Tp_2}+\frac{1}{p_1^2}\Big).\nonumber
\end{equation}
The proof of $\frac{1}{p_2}\norm{\hat\Cb-\Cb\hat{\Hb}_{C}}_F^2=O_p\Big(\frac{1}{Tp_1}+\frac{1}{p_2p_1^2}\Big)$ is similar. It remains to show the properties of $\hat\Hb_R$ and $\hat \Hb_C$. This is easy because
\[
\begin{split}
	 \hat\Hb_R^\top\Vb_1\hat\Hb_R=&\hat\Hb_R^\top\Big(\Vb_1-\frac{1}{p_1}\Rb^\top\Rb\Big)\hat\Hb_R+\frac{1}{p_1}(\Rb\hat\Hb_R)^\top(\Rb\hat\Hb_R-\hat\Rb)+\frac{1}{p_1}(\Rb\hat\Hb_R-\hat\Rb)^\top\hat\Rb+\frac{1}{p_1}\hat\Rb^\top\hat\Rb\\
	=&o_p(1)+\frac{1}{p_1}\hat\Rb^\top\hat\Rb=o_p(1)+\Ib_{k_1},
\end{split}
\]
where the $o_p(1)$ holds in terms of Frobenius norm. The proof of $\hat\Hb_C$ is similar and omitted.
\end{proof}

\subsection*{Proof of Theorem \ref{th2}: Convergence rate of $\hat{\Fb}_t$}
\begin{proof}
 By the least square optimization, we can get that $$\hat{\Fb}_t=\frac{1}{p_1p_2}\hat{\Rb}^\top\Xb_t\hat{\Cb}=\frac{1}{p_1p_2}\hat{\Rb}^\top\Rb\Fb_t\Cb^\top\hat{\Cb}+\frac{1}{p_1p_2}\hat{\Rb}^\top\Eb_t\hat{\Cb}.$$
Writing $\Rb=\hat{\Rb}\hat{\Hb}_R^{-1}-(\hat{\Rb}\hat{\Hb}_R^{-1}-\Rb)$ and $\Cb=\hat{\Cb}\hat{\Hb}_C^{-1}-(\hat{\Cb}\hat{\Hb}_C^{-1}-\Cb)$,  we obtain
\begin{equation}\label{facDEC}
\begin{aligned}
\hat{\Fb}_t-\hat{\Hb}_R^{-1}\Fb_t\hat{\Hb}_C^{-1\top}&=\frac{1}{p_1p_2}\hat{\Rb}^\top(\hat{\Rb}-\Rb\hat{\Hb}_R)\hat{\Hb}_R^{-1}\Fb_t\hat{\Hb}_C^{-1\top}(\hat{\Cb}-\Cb\hat{\Hb}_C)\hat{\Cb}+\frac{1}{p_1p_2}\hat{\Rb}^\top\Eb_t\hat{\Cb}\\
&-\frac{1}{p_1}\hat{\Rb}^\top(\hat{\Rb}-\Rb\hat{\Hb}_R)\hat{\Hb}_R^{-1}\Fb_t\hat{\Hb}_C^{-1\top}-\frac{1}{p_2}\hat{\Hb}_R^{-1}\Fb_t\hat{\Hb}_C^{-1\top}(\hat{\Cb}-\Cb\hat{\Hb}_C)^\top\hat{\Cb}.
\end{aligned}
\end{equation}
For $\hat{\Rb}^\top\Eb_t\hat{\Cb}$, further write $\hat{\Rb}=(\hat{\Rb}-\Rb\hat{\Hb}_R)+\Rb\hat{\Hb}_R$ and $\hat{\Cb}=(\hat{\Cb}-\Cb\hat{\Hb}_C)+\Cb\hat{\Hb}_C$. Then,
\begin{equation}
\begin{aligned}
\hat{\Fb}_t-\hat{\Hb}_R^{-1}\Fb_t\hat{\Hb}_C^{-1\top}&=\frac{1}{p_1p_2}\hat{\Rb}^\top(\hat{\Rb}-\Rb\hat{\Hb}_R)\hat{\Hb}_R^{-1}\Fb_t\hat{\Hb}_C^{-1\top}(\hat{\Cb}-\Cb\hat{\Hb}_C)\hat{\Cb}\\
&-\frac{1}{p_1}\hat{\Rb}^\top(\hat{\Rb}-\Rb\hat{\Hb}_R)\hat{\Hb}_R^{-1}\Fb_t\hat{\Hb}_C^{-1\top}-\frac{1}{p_2}\hat{\Hb}_R^{-1}\Fb_t\hat{\Hb}_C^{-1\top}(\hat{\Cb}-\Cb\hat{\Hb}_C)^\top\hat{\Cb}\\
&+\frac{1}{p_1p_2}(\hat{\Rb}-\Rb\hat{\Hb}_R)^\top\Eb_t(\hat{\Cb}-\Cb\hat{\Hb}_C)+\frac{1}{p_1p_2}(\hat{\Rb}-\Rb\hat{\Hb}_R)^\top\Eb_t\Cb\hat{\Hb}_C\\
&+\frac{1}{p_1p_2}\hat{\Hb}_R^\top\Rb^\top\Eb_t(\hat{\Cb}-\Cb\hat{\Hb}_C)+\frac{1}{p_1p_2}\hat{\Hb}_R^\top\Rb^\top\Eb_t\Cb\hat{\Hb}_C\\
&=\sum_{i=1}^7I_i.\nonumber
\end{aligned}
\end{equation}
Since $\frac{1}{p_1}\norm{\hat{\Rb}-\Rb\hat{\Hb}_R}_F^2=o_p(1)$ and $\frac{1}{p_2}\norm{\hat{\Cb}-\Cb\hat{\Hb}_C}_F^2=o_p(1)$ by Theorem \ref{th1}, term $I_1$ is dominated by $I_2$ and $I_3$.
By Lemma \ref{I2I3}, we have
$$\norm{I_2}_F^2=O_p\Big(\frac{1}{Tp_1p_2}\Big), \quad \norm{I_3}_F^2=O_p\Big(\frac{1}{Tp_1p_2}\Big).$$
For $I_4$, we have
\[
\|I_4\|_F^2\le \frac{1}{p_1^2p_2^2}\|\hat\Rb-\Rb\hat\Hb_R\|_F^2\|\Eb_t\|_F^2\|\hat\Cb-\Cb\hat\Hb_C\|_F^2=O_p\bigg(\Big(\frac{1}{Tp_2}+\frac{1}{p_1^2}\Big)\Big(\frac{1}{Tp_1}+\frac{1}{p_2^2}\Big)\bigg)=o_p\bigg(\frac{1}{p_1p_2}\bigg).
\]
For $I_5$, we have
\[
\|I_5\|_F^2\le \frac{1}{p_1^2p_2^2}\|\hat\Rb-\Rb\hat\Hb_R\|_F^2\|\Eb_t\Cb\|_F^2\|\hat\Hb_C\|_F^2=O_p\bigg(\frac{1}{Tp_2^2}+\frac{1}{p_1^2p_2}\bigg).
\]
Similarly,
\[
\|I_6\|_F^2=O_p\bigg(\frac{1}{Tp_1^2}+\frac{1}{p_1p_2^2}\bigg)\Rightarrow \|I_5\|_F^2+\|I_6\|_F^2=o_p\bigg(\frac{1}{p_1p_2}\bigg).
\]
On the other hand, Lemma \ref{I7} indicates that $\|I_7\|_F^2=O_p(1/(p_1p_2))$, which dominates in the errors. Then, we conclude that
$$\norm{\hat{\Fb}_t-\hat{\Hb}_R^{-1}\Fb_t\hat{\Hb}_C^{-1\top}}_F^2=O_p\Big(\frac{1}{p_1p_2}\Big),$$
which concludes the theorem.

\end{proof}
\subsection*{Proof of Theorem \ref{th3}: Convergence rate of $\hat{\Sbb}_t$}
By definition,
\begin{equation}
\begin{aligned}
\hat{\Sbb}_t-\Sbb_t=&\hat{\Rb}\hat{\Fb}_t\hat{\Cb}^\top-\Rb\Fb_t\Cb^\top=(\hat{\Rb}-\Rb\hat{\Hb}_R+\Rb\hat{\Hb}_R)\hat{\Fb}_t(\hat{\Cb}-\Cb\hat{\Hb}_C+\Cb\hat{\Hb}_C)^\top-\Rb\Fb_t\Cb^\top\\
=&(\hat{\Rb}-\Rb\hat{\Hb}_R)\hat{\Fb}_t(\hat{\Cb}-\Cb\hat{\Hb}_C)^\top+(\hat{\Rb}-\Rb\hat{\Hb}_R)\hat{\Fb}_t\hat{\Hb}_t^\top\Cb^\top
+\Rb\hat{\Hb}_R\hat{\Fb}_t(\hat{\Cb}-\Cb\hat{\Hb}_C)^\top\\
&+\Rb(\hat{\Hb}_R\hat{\Fb}_t\hat{\Hb}_C-\Fb_t)\Cb^\top.\nonumber
\end{aligned}
\end{equation}
By Theorems \ref{th1} and \ref{th2}, we have
$$\frac{1}{p_1p_2}\norm{\hat{\Sbb}_t-\Sbb_t}_F^2=O_p\Big(\frac{1}{p_1p_2}+\frac{1}{Tp_1}+\frac{1}{Tp_2}\Big).$$

\subsection*{Proof of Theorem \ref{k1k2Consistency}: Consistency of $\hat{k}_1$ and $\hat{k}_2$}
According to Lemma \ref{eigenvalue}, we have $\hat\lambda_j(\hat{\Kb}_r^X)\asymp 1, j \leq k_1$ and $\hat\lambda_j(\hat{\Kb}_r^X)=O_p(\delta_1), j \textgreater k_1$.
It is easy to check that for $j \textless k_1$ or $j \textgreater k_1$, $\hat\lambda_j(\hat{\Kb}_r^X)/\hat\lambda_{j+1}(\hat{\Kb}_r^X)= O(1)$, while for $j=k_1$ we have $\hat\lambda_j(\hat{\Kb}_r^X)/\hat\lambda_{j+1}(\hat{\Kb}_r^X) \rightarrow \infty$. Then $\hat{k}_1$ is consistent. The proof of the consistency of $\hat{k}_2$ can be similarly derived and is thus omitted.

\section{Technical Lemmas}

\begin{lemma}\label{lem}
	Let $\Xb \sim E_{p,q}(\Mb,\bSigma \otimes\bOmega,\psi)$ be a continuous random matrix, $\tilde{\Xb}$ be an independent copy of $\Xb$. Then we have
	\begin{equation}
		\Kb_{r}=\EE\left(\frac{(\Xb-\tilde{\Xb})(\Xb-\tilde{\Xb})^\top}{\norm{ \Xb-\tilde{\Xb}}_{F}^{2}}\right)=\EE\left(\frac{(\Xb-\Mb)(\Xb-\Mb)^\top}{\norm{ \Xb-\Mb}_{F}^{2}}\right).\nonumber
	\end{equation}
	\begin{equation}
		\Kb_{c}=\EE\left(\frac{(\Xb-\tilde{\Xb})^\top(\Xb-\tilde{\Xb})}{\norm{ \Xb-\tilde{\Xb}}_{F}^{2}}\right)=\EE\left(\frac{(\Xb-\Mb)^\top(\Xb-\Mb)}{\norm{ \Xb-\Mb}_{F}^{2}}\right).\nonumber
	\end{equation}
\end{lemma}
\begin{proof}
	By the difinition of the elliptical distribution, $\Xb \sim E_{p,q}(\Mb,\bSigma \otimes\bOmega,\psi)$, $\tilde{\Xb} \sim E_{p,q}(\Mb,\bSigma \otimes\bOmega,\psi)$. We have $\varphi_{\Xb-\tilde{\Xb}}(\Tb )=\psi^2[\tr(\Tb^\top\bSigma\Tb\bOmega)]$, implying that $\Xb-\tilde{\Xb} \sim E_{p,q}(\mathbf{0},\bSigma \otimes\bOmega,\psi^2)$. Hence, there exists a nonnegative random variable $r'$, such that $\Xb-\tilde{\Xb} \overset{d}{=}r'\Ab\Ub\Bb^\top$. Because $\Xb$ is continuous, we have $\PP(r'=0)=0$. Therefore,
	\begin{equation}
		\begin{aligned}
			 \Kb_r&=\EE\left(\frac{(\Xb-\tilde{\Xb})(\Xb-\tilde{\Xb})^\top}{\norm{\Xb-\tilde{\Xb}}_{F}^{2}}\right)=\EE\left(\frac{(r'\Ab\Ub\Bb^\top)(r'\Ab\Ub\Bb^\top)^\top}{\norm {r'\Ab\Ub\Bb^\top}_{F}^2}\right)\\
			 &=\EE\left(\frac{(\Ab\Ub\Bb^\top)(\Ab\Ub\Bb^\top)^\top}{\norm{(\Ab\Ub\Bb^\top)}_{F}^2}\right)=\EE\left(\frac{(r\Ab\Ub\Bb^\top)(r\Ab\Ub\Bb^\top)^\top}{\norm{(r\Ab\Ub\Bb^\top)}_{F}^2}\right)\\
			&=\EE\left(\frac{(\Xb-\Mb)(\Xb-\Mb)^\top}{\norm{\Xb-\Mb}_{F}^{2}}\right).\nonumber
		\end{aligned}
	\end{equation}
	The proof process for $\Kb_{c}$ is similar.
\end{proof}

\begin{lemma}\label{lemfree}
	Let $\Xb \sim E_{p,q}(\Mb,\bSigma \otimes\bOmega,\psi)$ be a continuous random matrix, $\bSigma=\Ab\Ab^\top$ and $\bOmega=\Bb\Bb^\top$. It takes another stochastic representation:
	\begin{equation}
		\Xb \overset{d}{=}\Mb + r\frac{\Zb}{\norm{\Ab^{\dagger}\Zb \Bb^{\dagger}}_{F}}. \nonumber
	\end{equation}
	where $\Zb \sim \mathcal{MN}(\mathbf{0},\bSigma,\bOmega)$, $\Ab^{\dagger}$ and $\Bb^{\dagger}$ are the Moore-Penrose pseudoinverse of $\Ab$ and $\Bb$ respectively.
\end{lemma}
\begin{proof}
	Let $\Xb=\Mb+r\Ab\Ub\Bb^\top$, $\text{rank}(\bSigma)=\text{rank}(\Ab)=m$, $\text{rank}(\bOmega)=\text{rank}(\Bb)=n$. Let $\Ub=\frac{\bvarepsilon}{\norm{\bvarepsilon}_{F}}$, so that $\text{Vec}(\Ub)=\frac{\text{Vec}(\bvarepsilon)}{\norm{\text{Vec}(\bvarepsilon)}_2}$, where $\text{Vec}(\bvarepsilon)$ is a standard normal vector in $\RR^{mn}$. Note that if $\Ab=\Vb_1\Db\Vb_2^\top$ is the singular value decomposition of $\Ab \in \RR^{p \times m}$ with $\Vb_1 \in \RR^{p \times m}$ and $\Db$,$\Vb_2 \in \RR^{m \times m}$, then $\Ab^{\dagger}=\Vb_{2}\Db^{-1}\Vb_{1}^\top$. Since $\text{rank}(\Ab)=m$, we have $\Ab^{\dagger}\Ab=\Ib_{m}$. $\Bb$ is similar. Let $\Zb=\Ab\bvarepsilon\Bb^\top$,
	\begin{equation}
		\text{Vec}(\Zb)=(\Bb \otimes \Ab)\text{Vec}(\bvarepsilon) \sim \mathcal{N}(\mathbf{0},\bOmega \otimes \bSigma) \Leftrightarrow \Zb \sim \mathcal{MN}(\mathbf{0},\bSigma,\bOmega) \nonumber
	\end{equation}
	then we have $\Xb=\Mb+r\Ab\frac{\bvarepsilon}{\norm{\bvarepsilon}_{F}}\Bb=\Mb+r\frac{\Zb}{\norm{\Ab^{\dagger}\Zb \Bb^{\dagger}}_{F}}$.
	The proof is complete.
\end{proof}

The following Lemma \ref{m2m3} to \ref{m4} establish the stochastic order of the following quantities:
$$
\begin{aligned}
	&\Mb_1=\frac{2}{T(T-1)}\sum_{1 \leq t \textless s \leq T}\frac{(\Fb_t-\Fb_s)\Cb^\top\Cb(\Fb_t-\Fb_s)^\top}{\norm{\Xb_t-\Xb_s}_F^2},\\
	&\Mb_2=\frac{2}{T(T-1)}\sum_{1 \leq t \textless s \leq T}\frac{(\Eb_t-\Eb_s)\Cb(\Fb_t-\Fb_s)^\top}{\norm{\Xb_t-\Xb_s}_F^2},\\
	&\Mb_3=\frac{2}{T(T-1)}\sum_{1 \leq t \textless s \leq T}\frac{(\Fb_t-\Fb_s)\Cb^\top(\Eb_t-\Eb_s)^\top}{\norm{\Xb_t-\Xb_s}_F^2},\\
	&\Mb_4=\frac{2}{T(T-1)}\sum_{1 \leq t \textless s \leq T}\frac{(\Eb_t-\Eb_s)(\Eb_t-\Eb_s)^\top}{\norm{\Xb_t-\Xb_s}_F^2}.
\end{aligned}
$$

\begin{lemma}\label{m2m3}
	Under Assumption \ref{joint_elliptical}-\ref{regular_noise}, as $\min \{T,p_1,p_2 \} \rightarrow \infty$, we have
	$$ \norm{\Mb_2}_F^2=O_p\Big(\frac{1}{Tp_1p_2}+\frac{1}{p_1^2p_2^2}\Big), $$
	$$\norm{\Mb_3}_F^2=O_p\Big(\frac{1}{Tp_1p_2}+\frac{1}{p_1^2p_2^2}\Big).$$
\end{lemma}
\begin{proof}
	Note that $\Mb_3$ is $\Mb_2$'s transpose, so we only show the result with $\Mb_2$. Similar to the proof in \cite{He2022large}, we assume $T$ is even and $\bar{T}=T/2$, otherwise we can delete the last observation. Given a permutation $\sigma$ of $\{1,\dots,T \}$, let $\Fb_t^{\sigma}$, $\Eb_t^{\sigma}$ and $\Xb_t^{\sigma}$ be the rearranged factors, errors, observations, and
	\begin{equation}
		 \Mb_2^\sigma=\frac{1}{\bar{T}}\sum_{s=1}^{\bar{T}}\frac{(\Eb_t^\sigma-\Eb_s^\sigma)\Cb(\Fb_t^\sigma-\Fb_s^\sigma)^\top}{\norm{\Xb_t^\sigma-\Xb_s^\sigma}_F^2}.\nonumber
	\end{equation}
	Denote $\mathcal{S}_T$ as the set containing all the permutations of $\{1,\dots,T \}$. After some elementary calculations,
	\begin{equation}
		\sum_{\sigma \in \mathcal{S}_T}\frac{T}{2}\Mb_2^\sigma=T\times (T-2)!\times \frac{T(T-1)}{2}\Mb_2.\nonumber
	\end{equation}
	That is,
	\begin{equation}
		\Mb_2=\frac{1}{T!}\sum_{\sigma \in \mathcal{S}_T}\Mb_2^\sigma \Rightarrow \EE \norm{\Mb_2}_F \leq \frac{1}{T!}\sum_{\sigma \in \mathcal{S}_T}\EE \norm{\Mb_2^\sigma}_F=\EE \norm{\Mb_2^\sigma}_F \leq \sqrt{\EE \norm{\Mb_2^\sigma}_F^2}.\nonumber
	\end{equation}
	Now we assume $\sigma$ is given. By the definition of joint matrix elliptical distribution, for any $s=1,\dots,\bar{T}$,
	\begin{equation}
		\left(\begin{array}{c}\vec(\Fb_{2s-1}-\Fb_{2s}) \\ \vec(\Eb_{2s-1}-\Eb_{2s})\end{array}\right)\overset{d}{=}r\left(\begin{array}{cc}\Ib_{k_2}\otimes \Ib_{k_1} & \zero \\ \zero & \bOmega_e^{1/2}\otimes \bSigma_e^{1/2}\end{array}\right) \frac{\bZ}{\|\bZ\|_2},\nonumber
	\end{equation}
	where $\bZ=(\vec(\Zb_1)^\top,\vec(\Zb_2)^\top)^\top \sim \mathcal{N}(\mathbf{0},\Ib_{k_1k_2+p_1p_2})$ and $\bZ$ is independent of $r$. $\Zb_1$ and $\Zb_2$ are $k_1\times k_2$ matrices, respectively. Hence,
	$$\Xb_s \overset{d}{=}\frac{\bSigma_e^{1/2}\Zb_2\bOmega_e^{1/2}\Cb\Zb_1^\top}{\norm{\Rb\Zb_1\Cb^\top+\bSigma_e^{1/2}\Zb_2\bOmega_e^{1/2}}_F^2}.$$ Note that $\Xb_s$ and $\Xb_t$ are independently and identically distributed when $s \neq t$, so
	\begin{equation}\label{EM2}
		\EE \norm{\Mb_2^\sigma}_F^2=\EE \Bigg\Vert\frac{1}{\bar{T}}\sum_{s=1}^{\bar{T}}\Xb_s\Bigg\Vert_F^2=\frac{1}{\bar{T}}\EE\norm{\Xb_1}_F^2+\frac{\bar{T}(\bar{T}-1)}{\bar{T}^2}\norm{\EE \Xb_1}_F^2.
	\end{equation}
	We first focus on the matrix $\EE \Xb_1$. Define
	\begin{equation}
		\begin{aligned}
			&\left(\begin{array}{c}\bu_1\\
				\bu_2 \end{array}\right)=\left( \begin{array}{cc}\Ib_{k_2}\otimes \Ib_{k_1} & -(\Cb \otimes \Rb)^\top(\Cb\Cb^\top \otimes \Rb\Rb^\top +\bOmega_e\otimes \bSigma_e)^{-1}\\ \zero & \Ib_{p_2}\otimes \Ib_{p_1} \end{array} \right)\left(\begin{array}{cc} \Ib_{k_2}\otimes \Ib_{k_1} & \zero\\ \Cb \otimes \Rb &\bOmega_e^{1/2} \otimes \bSigma_e^{1/2} \end{array} \right) \left(\begin{array}{c}\vec(\Zb_1)\\ \vec(\Zb_2) \end{array}\right)\\
			& \sim \mathcal{N}\bigg(\zero,\Big(\begin{array}{cc} \bSigma_{u_1} &\zero\\ \zero & \bSigma_x \end{array} \Big)\bigg),\nonumber
		\end{aligned}
	\end{equation}
	where $\bSigma_x=\Cb\Cb^\top \otimes \Rb\Rb^\top +\bOmega_e \otimes \bSigma_e$, $\bSigma_{u_1}=\Ib_{k_1k_2}-(\Cb \otimes \Rb)^\top \bSigma_x^{-1}(\Cb \otimes \Rb)$. $\bu_1$, $\bu_2$ are independent and
	\begin{equation}
		\left(\begin{array}{c}  \vec(\Zb_1) \\ \vec(\Zb_2) \end{array} \right)=\left(\begin{array}{cc} \Ib_{k_1k_2} & (\Cb \otimes \Rb)^\top \bSigma_x^{-1}\\ -(\bOmega^{1/2} \otimes \bSigma^{1/2})^{-1}(\Cb \otimes \Rb) & (\bOmega^{1/2} \otimes \bSigma^{1/2})\bSigma_x^{-1} \end{array}\right)\left(\begin{array}{c}\bu_1\\ \bu_2 \end{array} \right).\nonumber
	\end{equation}
	Since $k_1,k_2$ are fixed numbers, without loss generality we let $k_1=k_2=1$ in the following to simplify notation. For general fixed $k_1,k_2$, the error bounds are the same. As a result,
	\begin{equation}
		\begin{aligned}
			& \quad \quad \Xb_1 \overset{d}{=}\frac{\bSigma_e^{1/2}\Zb_2\bOmega_e^{1/2}\Cb Z_1^\top}{\norm{\Rb Z_1\Cb^\top+\bSigma_e^{1/2}\Zb_2\bOmega_e^{1/2}}_F^2} \\ &\Rightarrow \vec(\Xb_1)=\frac{(Z_1\Cb^\top\bOmega_e^{1/2})\otimes \bSigma_e^{1/2}\vec(\Zb_2)}{\norm{\vec(\Rb\Zb_1\Cb^\top+\bSigma_e^{1/2}\Zb_2\bOmega_e^{1/2})}_F^2}\\
			&=\frac{\big(u_1+(\Cb \otimes \Rb)^\top \bSigma_x^{-1}\bu_2\big)\big(-(\Cb^\top \otimes \Ib_{p_1})(\Cb \otimes \Rb)u_1+(\Cb^\top \otimes \Ib_{p_1})(\bOmega_e \otimes \bSigma_e)\bSigma_x^{-1}\bu_2\big)}{\norm{\bu_2}^2}.\nonumber
		\end{aligned}
	\end{equation}
	By Lemma S1 in \cite{He2022large}, $\norm{\bSigma_{u_1}}_F^2=O(\frac{1}{p_1^2p_2^2})$, $\norm{(\Cb \otimes \Rb)^\top \bSigma_x^{-2}(\Cb \otimes \Rb)}_F^2=O(\frac{1}{p_1^2p_2^2})$. $\bSigma_x$ and $\Cb \otimes \Rb$ are equivalent to $\bSigma_y$ and $\Lb$ respectively defined in \cite{He2022large}.
	
	Because $u_1$ and $\bu_2$ are zero-mean independent Gaussian vectors, we have
	$$\EE\frac{u_1\bu_2}{\norm{\bu_2}^2}=\zero,\quad  \EE\norm{\bu_2}^{-2} \leq \frac{1}{\lambda_{p_1p_2}(\bSigma_x)}\EE\frac{1}{\chi _{p_1p_2}^2} \asymp \frac{1}{p_1p_2}.$$
	Hence,
	\begin{equation}\label{eq1}
		\EE \vec (\Xb_1)=\text{Vec}\bigg(-\EE (\norm{\bu_2}^{-2})\big((\Cb^\top\Cb) \otimes \Rb \big)\bSigma_{u_1}+(\Cb^\top \otimes \Ib_{p_1})(\bOmega_e \otimes \bSigma_e)\bSigma_x^{-1}\EE \Big( \frac{\bu_2\bu_2^\top}{\norm{\bu_2}^2}\Big)\bSigma_x^{-1}(\Cb \otimes \Rb)\bigg).
	\end{equation}
	For the first term,
	\begin{equation}\label{eq2}
		\bigg\Vert-\EE (\norm{\bu_2}^{-2})\big((\Cb^\top\Cb) \otimes \Rb\big) \bSigma_{u_1}\bigg \Vert_F^2 \leq \norm{\bSigma_{u_1}}_F^2\norm{\Cb^\top\Cb}_F^2\norm{\Rb}_F^2\left(\EE(\norm{\bu_2}^{-2})\right)^2=O\Big(\frac{1}{p_1^3p_2^2}\Big).
	\end{equation}
	For the second term, according to Lemma S2 in \cite{He2022large}, we have
	$$\bigg \Vert \bSigma_x^{-1} \EE \Big(\frac{\bu_2\bu_2^\top}{\norm{\bu_2}^2}\Big)\bSigma_x^{-1}(\Cb \otimes \Rb)\bigg \Vert_F^2=O\Big(\frac{1}{p_1^3p_2^3}\Big).$$
	Therefore,
	\begin{equation}\label{eq3}
         \begin{aligned}
		&\bigg \Vert (\Cb^\top \otimes \Ib_{p_1})(\bOmega_e \otimes \bSigma_e)\bSigma_x^{-1}\EE \Big(\frac{\bu_2\bu_2^\top}{\norm{\bu_2}^2}\Big)\bSigma_x^{-1}(\Cb \otimes \Rb) \bigg \Vert_F^2 \\
&\leq \norm{\Cb^\top \otimes \Ib_{p_1}}_F^2\norm{\bOmega_e \otimes \bSigma_e}^2\bigg \Vert \bSigma_x^{-1} \EE \Big(\frac{\bu_2\bu_2^\top}{\norm{\bu_2}^2}\Big)\bSigma_x^{-1}(\Cb \otimes \Rb)\bigg \Vert_F^2=O\Big(\frac{1}{p_1^2p_2^2}\Big).
       \end{aligned}
	\end{equation}
	Combing ( \ref{eq1}),( \ref{eq2}) and (\ref{eq3}), we have $\norm{\EE\vec(\Xb_1)}_F^2=O\Big(\frac{1}{p_1^2p_2^2}\Big).$
	
	Now we consider $\EE\norm{\Xb_1}_F^2$. According to Lemma S2 in \cite{He2022large},$$\EE\norm{u_1}^2 \lesssim \frac{1}{p_1p_2},\quad \EE\norm{u_1}^4 \lesssim \frac{1}{p_1p_2}, \quad \EE\norm{\bu_2}^{-4}\asymp \frac{1}{p_1^2p_2^2},\quad \EE\norm{\bu}_2^{-8}\leq \frac{1}{\lambda_{p_1p_2}^4(\bSigma_x)}\EE \chi_{p_1p_2}^{-4} \asymp (p_1p_2)^{-4}.$$
	Then
	\begin{equation}
		\begin{aligned}
			\EE\norm{\Xb_1}_F^2&=\EE\bigg \Vert\frac{\bSigma_e^{1/2}\Zb_2\bOmega_e^{1/2}\Cb Z_1^\top}{\norm{\Rb Z_1\Cb^\top+\bSigma_e^{1/2}\Zb_2\bOmega_e^{1/2}}_F^2} \bigg \Vert_F^2=\EE \bigg( \frac{1}{\norm{\Rb Z_1\Cb^\top+\bSigma_e^{1/2}\Zb_2\bOmega_e^{1/2}}_F^4}\norm{\bSigma_e^{1/2}\Zb_2\bOmega_e^{1/2}\Cb Z_1^\top}_F^2 \bigg)\\
			&\leq C \sqrt{\EE(\norm{\Zb_2\bOmega_e^{1/2}\Cb}_F^4)\EE(\norm{\bu_2}^{-8}})=
			O\Big(\frac{1}{p_1p_2}\Big).\nonumber
		\end{aligned}
	\end{equation}
	
	As a result,
	$$\EE\norm{\Mb_2^\sigma}_F \lesssim \sqrt{\frac{1}{Tp_1p_2}+\frac{1}{p_1^2p_2^2}} \Rightarrow \norm{\Mb_2}_F^2=O_p\Big(\frac{1}{Tp_1p_2}+\frac{1}{p_1^2p_2^2}\Big).$$
	$\norm{\Mb_3}_F^2=O_p\Big(\frac{1}{Tp_1p_2}+\frac{1}{p_1^2p_2^2}\Big)$ can be derived similarly.
\end{proof}

\begin{lemma}\label{m1}
	Under Assumptions \ref{joint_elliptical}-\ref{regular_noise}, we have $\norm{\Mb_1}_F^2=O_p(\frac{1}{p_1^2})$ and $\norm{\hat{\Hb}_R}_F^2=O_p(1)$.
\end{lemma}
\begin{proof}
	Similarly to equation (\ref{EM2}) in Lemma \ref{m2m3}, it is easy to verify that
	$$\EE\norm{\Mb_1}_F \leq \sqrt{\frac{1}{\bar{T}}\EE\norm{\Xb}_F^2+\frac{\bar{T}(\bar{T}-1)}{\bar{T}^2}\norm{\EE\Xb}_F^2},$$
	where $\Xb \overset{d}{=}\frac{Z_1\Cb^\top\Cb Z_1^\top}{\norm{\Rb Z_1\Cb^\top+\bSigma_e^{1/2}\Zb_2\bOmega_e^{1/2}}_F^2}$, $\vec (\Xb)=\frac{\big(u_1+(\Cb \otimes \Rb)^\top\bSigma_x^{-1}\bu_2\big)\Cb^\top\Cb\big(u_1+(\Cb \otimes \Rb)^\top\bSigma_x^{-1}\bu_2\big)}{\norm{\bu_2}^2}$ with $Z_1,\Zb_2$, $u_1$, $\bu_2$ defined in Lemma \ref{m2m3}.
	Therefore,
	$$\norm{\EE\vec(\Xb)}_F^2 \lesssim (\EE\norm{\bu_2}_F^{-2})^2\norm{\Cb^\top\Cb}_F^2\norm{\bSigma_{u_1}}_F^2+\norm{\Cb \otimes \Rb}_F^2\bigg \Vert \bSigma_x^{-1} \EE \Big(\frac{\bu_2\bu_2^\top}{\norm{\bu_2}^2}\Big)\bSigma_x^{-1}(\Cb \otimes \Rb)\bigg \Vert_F^2\norm{\Cb^\top\Cb}_F^2=O\Big(\frac{1}{p_1^2}\Big).$$
	On the other hand,
	\begin{equation}
		\begin{aligned}
			\EE\norm{\Xb}_F^2 & \lesssim \EE \big(\norm{u_1}^4\big)\EE\big(\norm{\bu_2}^{-4}\big)\norm{\Cb^\top\Cb}_F^2 \le O\Big(\frac{1}{p_1^2}\Big).\nonumber
		\end{aligned}
	\end{equation}
	As a result, $\norm{\Mb_1}_F^2=O_p(\frac{1}{p_1^2})$ and by the definition of $\hat{\Hb}_R$ we have
	 $$\norm{\hat{\Hb}_R}_F^2\leq \norm{\Mb_1}_F^2\norm{\Rb^\top}_F^2\norm{\hat{\Rb}}_F^2\norm{\hat{\bLambda}^{-1}}_F^2=O_p(1),$$
	 which completes the proof.
\end{proof}

\begin{lemma}\label{m4}
	Under Assumption \ref{joint_elliptical}-\ref{regular_noise}, we have
	\begin{equation}
		\frac{1}{p_1}\norm{\Mb_4\hat{\Rb}}_F^2=O_p\Big(\frac{1}{Tp_2}+\frac{1}{p_1^2}\Big)+o_p(1) \times \frac{1}{p_1}\norm{\hat{\Rb}-\Rb\hat{\Hb}_R}_F^2.\nonumber
	\end{equation}
\end{lemma}
\begin{proof}
	By the decomposition $\hat{\Rb}=\hat{\Rb}-\Rb\hat{\Hb}_R+\Rb\hat{\Hb}_R$, we have
	\begin{equation}
		\frac{1}{p_1}\norm{\Mb_4\hat{\Rb}}_F^2 \lesssim \frac{1}{p_1}\norm{\Mb_4\Rb}_F^2\norm{\hat{\Hb}_R}_F^2+\norm{\Mb_4}_F^2\times \frac{1}{p_1}\norm{\hat{\Rb}-\Rb\hat{\Hb}_R}_F^2.
	\end{equation}
	We start with $\|\Mb_4\|_F^2$. Similarly to the proof of equation \ref{EM2} in Lemma \ref{m2m3},
	$$\EE\norm{\Mb_4}_F \leq \sqrt{\frac{1}{\bar{T}}\EE\norm{\Xb}_F^2+\frac{\bar{T}(\bar{T}-1)}{\bar{T}^2}\norm{\EE\Xb}_F^2},$$
	where $\Xb \overset{d}{=}\frac{\bSigma_e^{1/2}\Zb_2\bOmega_e\Zb_2^\top\bSigma_e^{1/2}}{\norm{\Rb Z_1\Cb^\top+\bSigma_e^{1/2}\Zb_2\bOmega_e^{1/2}}_F^2}.$
	Let $\Nb=\bSigma_e^{1/2}\Zb_2\bOmega_e^{1/2}$, so
	$$\vec(\Xb)=\frac{\sum_{j=1}^{p_2}\Nb_{\cdot,j}\otimes\Nb_{\cdot,j} }{\norm{\bu_2}^2}.$$
	Denote $\Eb_j\big((\bOmega_e^{1/2} \otimes \bSigma_e^{1/2})\vec(\Zb_2)\big)$ to extract the $\big((j-1)p_1+1\big)$-th to $jp_1$-th entries of $(\bOmega_e^{1/2} \otimes \bSigma_e^{1/2})\vec(\Zb_2)$, where
	$$\Eb_j=(\zero,\dots,\underbrace{\Ib_{p_1}}_{\mbox{$j$-th\ block}},\dots,\zero)_{p_1\times(p_1p_2)}.$$
	Then,
	\begin{equation}
		\begin{aligned}
			\vec(\Xb)&=\frac{\sum_{j=1}^{p_2}\bigg(\Eb_j\Big(\big(\bOmega_e^{1/2} \otimes \bSigma_e^{1/2}\big)\vec\big(\Zb_2\big)\Big)\otimes \Big(\Eb_j\big(\bOmega_e^{1/2} \otimes \bSigma_e^{1/2}\big)\vec\big(\Zb_2\big)\Big)\bigg)}{\norm{\bu_2}^2}\\
			&=\frac{\sum_{j=1}^{p_2}\bigg(\Eb_j\Big(-\big(\Cb \otimes \Rb\big)u_1+(\bOmega_e \otimes \bSigma_e)\bSigma_x^{-1}\bu_2\Big)\bigg)\otimes \bigg(\Eb_j\Big(-\big(\Cb \otimes \Rb\big)u_1+\big(\bOmega_e \otimes \bSigma_e\big)\bSigma_x^{-1}\bu_2\Big)\bigg)}{\norm{\bu_2}^2}.\nonumber
		\end{aligned}
	\end{equation}
	Therefore, $$\EE \vec(\Xb)=\sum_{j=1}^{p_2}\Bigg(\EE \big(\norm{\bu_2}^{-2}\big)\Big(\Eb_j\big(\Cb \otimes \Rb\big)\Big) \otimes \Big(\Eb_j\big(\Cb \otimes \Rb\big)\Big)\bSigma_{u_1} + \EE \bigg( \frac{\left(\Eb_j\left(\bOmega_e \otimes \bSigma_e \right)\bSigma_x^{-1}\bu_2\right)\otimes \left(\Eb_j\left(\bOmega_e \otimes \bSigma_e \right)\bSigma_x^{-1}\bu_2\right)}{\norm{\bu_2}^2}\bigg)\Bigg).$$
	For the first term,
	$$\bigg \Vert-\EE \big(\norm{\bu_2}^{-2}\big)\Big(\Eb_j\big(\Cb \otimes \Rb\big)\Big) \otimes \Big(\Eb_j\big(\Cb \otimes \Rb\big)\Big)\bSigma_{u_1}\bigg \Vert_F^2 \leq \norm{\Eb_j(\Cb \otimes \Rb)}_F^4\norm{\bSigma_{u_1}}_F^2\big(\EE(\norm{\bu_2}^{-2})\big)^2=O\left(\frac{1}{p_1^2p_2^4}\right),$$
	For the second term,
	$$\EE\Bigg(\frac{\big(\Eb_j(\bOmega_e \otimes \bSigma_e)\bSigma_x^{-1}\bu_2\big)\otimes\big(\Eb_j(\bOmega_e \otimes \bSigma_e)\bSigma_x^{-1}\bu_2\big)}{\norm{\bu_2}^2}\Bigg)=\text{Vec}\left(\big(\Eb_j(\bOmega_e\otimes \bSigma_e)\bSigma_x^{-1}\big)\EE\Big(\frac{\bu_2\bu_2^\top}{\norm{\bu_2}^2}\Big)\big(\Eb_j(\bOmega_e\otimes \bSigma_e)\bSigma_x^{-1}\big)^\top\right).$$
	Therefore,
	\begin{equation}\label{m4dec}
		\begin{aligned}
			&\bigg \Vert \EE\bigg(\frac{(\Eb_j(\bOmega_e \otimes \bSigma_e)\bSigma_x^{-1}\bu_2)\otimes(\Eb_j(\bOmega_e \otimes \bSigma_e)\bSigma_x^{-1}\bu_2)}{\norm{\bu_2}^2} \bigg)\bigg \Vert_F^2\\
			&\leq \norm{\Eb_j(\bOmega_e \otimes \bSigma_e)\bSigma_x^{-1}(\bOmega_e \otimes \bSigma_e)^\top\Eb_j^\top}_F^2\bigg \Vert \EE \bigg(\frac{\bSigma_x^{-1/2}\bu_2\bu_2^\top\bSigma_x^{-1/2}}{\norm{\bu_2}^2}\bigg) \bigg \Vert^2\\
			&\lesssim \frac{1}{p_1^2p_2^2}\norm{\Eb_j(\bOmega_e \otimes \bSigma_e)\bSigma_x^{-1}(\bOmega_e \otimes \bSigma_e)^\top\Eb_j^\top}_F^2\\
			&\asymp \frac{1}{p_1^2p_2^2}\norm{\Eb_j\bSigma_x\Eb_j^\top}_F^2=O\Big(\frac{1}{p_1p_2^2}\Big),
		\end{aligned}
	\end{equation}
	where we use the fact that $\bigg \Vert \EE \Big( \frac{\bSigma_x^{-1/2}\bu_2\bu_2^\top\bSigma_x^{-1/2}}{\norm{\bu_2}^2}\Big) \bigg \Vert^2=O\Big(\frac{1}{p_1^2p_2^2}\Big)$. To prove this, denote the spectral decomposition of $\bSigma_x$ as $\bGamma_x\bLambda_x\bGamma_x^\top$ where $\bGamma_x$ is an orthonormal matrix. Then
	$$ \bigg \Vert\EE \bigg(\bSigma_x^{-1/2}\frac{\bu_2\bu_2^\top}{\norm{\bu_2}^2}\bSigma_x^{-1/2} \bigg)\bigg \Vert ^2=\bigg \Vert \bLambda_x^{-1/2}\EE \bigg(\frac{\bGamma_x^\top\bu_2\bu_2^\top\bGamma_x}{\norm{\bGamma_x^\top\bu_2}^2}\bigg)\bLambda_x^{-1/2}\bigg \Vert^2.$$
	Denote $$\tilde{\bLambda}:=\EE \bigg(\frac{\bGamma_x^\top\bu_2\bu_2^\top\bGamma_x}{\norm{\bGamma_x^\top\bu_2}^2}\bigg),$$
	so that
	$$\tilde{\bLambda} \leq \bLambda_x^{1/2}\EE \bigg(\frac{\bLambda_x^{-1/2}\bGamma_x^\top\bu_2\bu_2^\top\bGamma_x\bLambda_x^{-1/2}}{\lambda_{p_1p_2}(\bSigma_x)\norm{\bLambda_x^{-1/2}\bGamma_x^\top\bu_2}^2}\bigg)\bLambda_x^{1/2}=\frac{1}{p_1p_2\lambda_{p_1p_2}(\bSigma_x)}\bLambda_x.$$
	Consequently,
	$$\bigg \Vert \EE\bigg( \frac{\bSigma_x^{-1/2}\bu_2\bu_2^\top\bSigma_x^{-1/2}}{\norm{\bu_2}^2}\bigg) \bigg \Vert^2=O\Big(\frac{1}{p_1^2p_2^2}\Big).$$
	Hence,
	 $$\norm{\EE\vec(\Xb)}_F^2=p_2^2\bigg(O\Big(\frac{1}{p_1^2p_2^4}\Big)+\Big(\frac{1}{p_1p_2^2}\Big)\bigg)=O\Big(\frac{1}{p_1}\Big).$$
	On the other hand,
	\begin{equation}
		\begin{aligned}
			\EE \norm{\Xb}_F^2&=\EE \bigg \Vert \frac{\bSigma_e^{1/2}\Zb_2\bOmega_e\bZ_2^\top\bSigma_e^{1/2}}{\norm{\Rb Z_1\Cb+\bSigma_e^{1/2}\Zb_2\bOmega_e^{1/2}}_F^2}\bigg \Vert_F^2\\
			&=\EE \bigg(\frac{1}{\norm{\Rb Z_1\Cb+\bSigma_e^{1/2}\Zb_2\bOmega_e^{1/2}}_F^4}\norm{\bSigma_e^{1/2}\bZ_2\bOmega_e\Zb_2^\top\bSigma_e^{1/2}}_F^2\bigg)\\
			&\leq C\sqrt{\EE\big(\norm{\Zb_2\bOmega_e\Zb_2^\top}_F^4\big)\EE\big(\norm{\bu_2}^{-8}\big)}=O\Big(\frac{1}{p_2}\Big).\nonumber
		\end{aligned}
	\end{equation}
	As a result,
	$$\norm{\Mb_4}_F^2=O\Big(\frac{1}{Tp_2}+\frac{1}{p_1}\Big)=o_p(1).$$
	
	Next, we consider $\|\Mb_4\Rb\|_F^2$. Indeed, the proof is very similar to that of $\|\Mb_4\|_F^2$. We have
	$$\EE\norm{\Mb_4\Rb}_F \leq \sqrt{\frac{1}{\bar{T}}\EE\norm{\tilde\Xb}_F^2+\frac{\bar{T}(\bar{T}-1)}{\bar{T}^2}\norm{\EE\tilde\Xb}_F^2},$$
	where $\tilde\Xb \overset{d}{=}\frac{\bSigma_e^{1/2}\Zb_2\bOmega_e\Zb_2^\top\bSigma_e^{1/2}\Rb}{\norm{\Rb Z_1\Cb^\top+\bSigma_e^{1/2}\Zb_2\bOmega_e^{1/2}}_F^2}.$ Define
	$$\tilde\Eb_j=(\zero,\dots,\underbrace{\Ib_{k_1}}_{\mbox{$j$-th\ block}},\dots,\zero)_{k_1\times(k_1p_2)}.$$
	Then,
	\begin{equation}
		\begin{aligned}
			\vec(\tilde\Xb)=&\frac{\sum_{j=1}^{p_2}\bigg(\Eb_j\Big(\big(\bOmega_e^{1/2} \otimes \bSigma_e^{1/2}\big)\vec\big(\Zb_2\big)\Big)\otimes \Big(\tilde\Eb_j\big(\bOmega_e^{1/2} \otimes \Rb^\top\bSigma_e^{1/2}\big)\vec\big(\Zb_2\big)\Big)\bigg)}{\norm{\bu_2}^2}\\
			=&\frac{\sum_{j=1}^{p_2}\bigg(\Big(\Eb_j\big(-(\Cb \otimes \Rb)u_1+(\bOmega_e \otimes \bSigma_e)\bSigma_x^{-1}\bu_2\big)\Big)\otimes \Big(\tilde\Eb_j\big(-(\Cb \otimes \Rb^\top\Rb)u_1+(\bOmega_e \otimes\Rb^\top \bSigma_e)\bSigma_x^{-1}\bu_2\big)\Big)\bigg)}{\norm{\bu_2}^2},\\
			\EE\big(\vec(\tilde\Xb)\big)=&\sum_{j=1}^{p_2}\Bigg(\EE \Big(\norm{\bu_2}^{-2}\Big)\Big(\Eb_j\big(\Cb \otimes \Rb\big)\Big) \otimes \Big(\tilde\Eb_j\big(\Cb \otimes\Rb^\top \Rb\big)\Big)\bSigma_{u_1} \\
			&+ \EE \bigg( \frac{\big(\Eb_j(\bOmega_e \otimes \bSigma_e)\bSigma_x^{-1}\bu_2\big)\otimes \big(\tilde\Eb_j(\bOmega_e \otimes \Rb^\top\bSigma_e)\bSigma_x^{-1}\bu_2\big)}{\norm{\bu_2}^2}\bigg)\Bigg).\nonumber
		\end{aligned}
	\end{equation}
	For the first term,  we see that
	\[
	\begin{split}
		\left\|\EE \big(\norm{\bu_2}^{-2}\big)\big(\Eb_j(\Cb \otimes \Rb)\big) \otimes \big(\tilde\Eb_j(\Cb \otimes\Rb^\top \Rb)\big)\bSigma_{u_1} \right\|_F^2\le \mathbb{E}\|\bu_2\|^{-4}\|\bSigma_{u_1}\|_F^2\|\Cb\|_F^4\|\Rb\|_F^2\|\Rb^\top\Rb\|_F^2=O\bigg(\frac{1}{p_1p_2^2}\bigg).
	\end{split}
	\]
	For the second term,
	\[
	\begin{split}
		&\EE\bigg(\frac{\big(\Eb_j(\bOmega_e \otimes \bSigma_e)\bSigma_x^{-1}\bu_2\big)\otimes\big(\tilde\Eb_j(\bOmega_e \otimes\Rb^\top \bSigma_e)\bSigma_x^{-1}\bu_2\big)}{\norm{\bu_2}^2}\bigg)\\
		=&\text{Vec}\Big(\big(\tilde\Eb_j(\bOmega_e\otimes\Rb^\top \bSigma_e)\bSigma_x^{-1}\big)\EE\bigg(\frac{\bu_2\bu_2^\top}{\norm{\bu_2}^2}\bigg)\big(\Eb_j(\bOmega_e\otimes \bSigma_e)\bSigma_x^{-1}\big)^\top\Big),
	\end{split}
	\]
	which implies that
	\[
	\begin{split}
		\bigg\|\EE\bigg(\frac{\big(\Eb_j(\bOmega_e \otimes \bSigma_e)\bSigma_x^{-1}\bu_2\big)\otimes\big(\tilde\Eb_j(\bOmega_e \otimes\Rb^\top \bSigma_e)\bSigma_x^{-1}\bu_2\big)}{\norm{\bu_2}^2}\bigg)\bigg\|_F^2\lesssim&\frac{1}{p_1^2p_2^2}\|\tilde\Eb_j(\bOmega_e\otimes \Rb^\top\bSigma_{e})\bSigma_x^{-1}(\bOmega_e\otimes \bSigma_e)^\top\Eb_j\|_F^2\\
		\lesssim &\frac{1}{p_1^2p_2^2}\|\Rb\|_F^2\le O\bigg(\frac{1}{p_1p_2^2}\bigg).
	\end{split}
	\]
	Therefore, we conclude that
	\[
	\|\mathbb{E}\tilde\Xb\|_F^2\le O\Big(\frac{1}{p_1}\Big).
	\]
	On the other hand,
	\[
	\mathbb{E}\|\tilde\Xb\|_F^2\le \mathbb{E}\|\Xb\|_F^2\|\Rb\|^2\le O\Big(\frac{p_1}{p_2}\Big).
	\]
	Consequently, we conclude that
	\[
	\frac{1}{p_1}\|\Mb_4\Rb\|_F^2\le O_p\bigg(\frac{1}{Tp_2}+\frac{1}{p_1^2}\bigg).
	\]

	Hence,
	\begin{equation}
		\frac{1}{p_1}\norm{\Mb_4\hat{\Rb}}_F^2=O_p\Big(\frac{1}{Tp_2}+\frac{1}{p_1^2}\Big)+o_p(1) \times \frac{1}{p_1}\norm{\hat{\Rb}-\Rb\hat{\Hb}}_F^2,\nonumber
	\end{equation}
which concludes the lemma.
\end{proof}

\begin{lemma}\label{I7}
	Under Assumptions \ref{joint_elliptical}-\ref{regular_noise}, we have
	 $$\left\|\frac{1}{p_1p_2}\hat{\Hb}_R^\top\Rb^\top\Eb_t\Cb\hat{\Hb}_C\right\|_F^2=O_p\Big(\frac{1}{p_1p_2}\Big).$$
\end{lemma}
\begin{proof}
	According to equation \ref{eq_joint_ellip_model}, we have $\Eb_t = \frac{r_t}{\|\bZ_t\|_2}\bSigma_e^{1/2} \Zb_t^{E}\bOmega_e^{1/2}$, then
	\begin{equation}
		\begin{aligned}
			\norm{\Rb^\top\Eb_t\Cb}_F^2&=\bigg \Vert\frac{r_t}{\|\bZ_t\|_2}\Rb^\top\bSigma_e^{1/2} \Zb_t^{E}\bOmega_e^{1/2}\Cb\bigg\Vert_F^2=\bigg \Vert\frac{r_t}{\|\bZ_t\|_2}(\Cb^\top\bOmega_e^{1/2}) \otimes (\Rb^\top\bSigma_e^{1/2})\vec(\Zb_t^{E})\bigg \Vert_F^2\\
			&=\bigg \Vert \frac{r_t}{\|\bZ_t\|_2}\left(\zero\quad (\Cb^\top\bOmega_e^{1/2}) \otimes (\Rb^\top\bSigma_e^{1/2})\right)\bZ_t\bigg\Vert_F^2=O_p(p_1p_2).
		\end{aligned}
	\end{equation}
	Since $\norm{\hat{\Hb}_R}_F^2=O_p(1)$, $\norm{\hat{\Hb}_C}_F^2=O_p(1)$, we have
	$$\Big \Vert \frac{1}{p_1p_2}\hat{\Hb}_R^\top\Rb^\top\Eb_t\Cb\hat{\Hb}_C\Big \Vert_F^2=O_p\Big(\frac{1}{p_1p_2}\Big).$$
\end{proof}

\begin{lemma}\label{I2I3}
	Under Assumptions \ref{joint_elliptical}-\ref{regular_noise}, we have
	$$\Big \Vert\frac{1}{p_1}\hat{\Rb}^\top(\hat{\Rb}-\Rb\hat{\Hb}_R)\Big \Vert_F^2=O_p\Big(\frac{1}{Tp_1p_2}\Big),\quad \Big \Vert\frac{1}{p_2}(\hat{\Cb}-\Cb\hat{\Hb}_C)^\top\hat{\Cb}\Big \Vert_F^2=O_p\Big(\frac{1}{Tp_1p_2}\Big).$$
\end{lemma}
\begin{proof}
	Note that $$\frac{1}{p_1}\hat{\Rb}^\top(\hat{\Rb}-\Rb\hat{\Hb}_R)=\frac{1}{p_1}(\hat{\Rb}-\Rb\hat{\Hb}_R)^\top(\hat{\Rb}-\Rb\hat{\Hb}_R)+\frac{1}{p_1}\hat{\Hb}_R^\top\Rb^\top(\hat{\Rb}-\Rb\hat{\Hb}_R)$$ while $\frac{1}{p_1}\norm{\hat\Rb-\Rb\hat{\Hb}_{R}}_F^2=O_p\left(\frac{1}{Tp_2}+\frac{1}{p_1^2}\right)$ and $\hat{\Hb}_R=O_p(1)$. Hence it suffices to bound $\Rb^\top(\hat{\Rb}-\Rb\hat{\Hb}_R)$.
	By equation (\ref{R-RH}),
	 $$\Rb^\top(\hat{\Rb}-\Rb\hat{\Hb}_R)=(\Rb^\top\Mb_2\Rb^\top+\Rb^\top\Rb\Mb_3+\Rb^\top\Mb_4)\hat{\Rb}\hat{\bLambda}^{-1}.$$
	We will bound the three error terms separately.
	
	Firstly, similar to the proof of equation \ref{EM2} in Lemma \ref{m2m3}, we have
	$\EE \norm{\Rb^\top\Mb_2}_F \lesssim \sqrt{\frac{1}{T}\EE\norm{\Xb}_F^2+\norm{\EE \Xb}_F^2}$, where $\Xb \overset{d}{=}\frac{\Rb^\top\bSigma_e^{1/2}\Zb_2\bOmega_e^{1/2}\Cb\Zb_1^\top}{\norm{\Rb\Zb_1\Cb^\top+\bSigma_e^{1/2}\Zb_2\bOmega_e^{1/2}}_F^2}$. We still assume $k_1=k_2=1$. Then,
	\begin{equation}
		\begin{aligned}
			 \vec(\Xb)&=\frac{Z_1(\Cb^\top\bOmega_e^{1/2})\otimes(\Rb^\top\bSigma_e^{1/2})\vec(\Zb_2)}{\norm{\vec(\Rb Z_1 \Cb^\top+\bSigma_e^{1/2}\Zb_2\bOmega_e^{1/2})}^2}\\
			&=\frac{\left(u_1+\left(\Cb \otimes \Rb \right)^\top\bSigma_x^{-1}\bu_2\right)\left(\Cb\bOmega_e^{1/2} \otimes \Rb\bSigma_e^{1/2}\right) \left(-(\bOmega_e^{1/2} \otimes \bSigma_e^{1/2})^{-1}(\Cb \otimes \Rb)u_1+(\bOmega_e^{1/2} \otimes \bSigma_e^{1/2})\bSigma_x^{-1}\bu_2\right)}{\norm{\bu_2}_F^2}\\
			&=\frac{-u_1(\Cb^\top\Cb \otimes \Rb^\top \Rb)u_1+u_1(\Cb \otimes \Rb)^\top(\bOmega_e \otimes \bSigma_e)\bSigma_x^{-1}\bu_2}{\norm{\bu_2}_F^2}\\
			&\frac{-(\Cb \otimes \Rb)^\top \bSigma_x^{-1}\bu_2(\Cb^\top\Cb \otimes \Rb^\top \Rb)u_1+(\Cb \otimes \Rb)^\top \bSigma_x^{-1}\bu_2(\Cb \otimes \Rb)^\top(\bOmega_e \otimes \bSigma_e)\bSigma_x^{-1}\bu_2}{\norm{\bu_2}_F^2}.\nonumber
		\end{aligned}
	\end{equation}
	Following similar technique as in the proof of Lemma \ref{m2m3},
	$$\norm{\EE \vec(\Xb)}_F^2 \lesssim \norm{\Cb \otimes \Rb}_F^4 \norm{\bSigma_{u_1}}_F^2(\EE \norm{\bu_2}^{-2})^2+\norm{\Cb \otimes \Rb}_F^2\norm{\bOmega_e \otimes \bSigma_e}^2\Big \Vert\bSigma_x^{-1}\EE \Big(\frac{\bu_2\bu_2^\top}{\norm{\bu_2}^2}\Big)\bSigma_x^{-1}(\Cb \otimes \Rb)\Big \Vert_F^2=O\Big(\frac{1}{p_1^2p_2^2}\Big).$$
	On the other hand,
	\begin{equation}
		\begin{aligned}
			\EE \norm{\Xb}_F^2 &\lesssim \mathbb{E}\big(\|\bu_2\|^{-4}\big)\|\Rb^\top\bSigma_e^{1/2}\Zb_2\bOmega_e^{1/2}\Cb Z_1\|_F^2\le O\Big(\frac{1}{p_1p_2}\Big).\nonumber
		\end{aligned}
	\end{equation}
	As a result,
	$$\norm{\Rb^\top\Mb_2}_F^2=O_p\Big(\frac{1}{Tp_1p_2}+\frac{1}{p_1^2p_2^2}\Big).$$
	Similarly, $\norm{\Mb_3\Rb}_F^2=O_p\Big(\frac{1}{Tp_1p_2}+\frac{1}{p_1^2p_2^2}\Big)$ and
	
	$$\norm{\Mb_3\hat{\Rb}}_F^2 \lesssim \norm{\Mb_3\Rb}_F^2+\norm{\Mb_3}_F^2\norm{\hat{\Rb}-\Rb\hat{\Hb}_R}_F^2=O_p\Big(\frac{1}{Tp_1p_2}+\frac{1}{T^2p_2^2}+\frac{1}{p_1^2p_2^2}\Big).$$

	For the last term, we can verify that
	$\EE \norm{\Rb^\top \Mb_4 \Rb}_F \lesssim \sqrt{\frac{1}{T}\EE \norm{\Xb}_F^2+\norm{\EE\Xb}_F^2}$, where $\Xb \overset{d}{=}\frac{\Rb^\top\bSigma_e^{1/2}\Zb_2\bOmega_e\bZ_2^\top\bSigma_e^{1/2}\Rb}{\norm{\Rb Z_1\Cb^\top+\bSigma_e^{1/2}\Zb_2\bOmega_e^{1/2}}_F^2}$ by slightly abusing the notation $\Xb$.
	Let $\Nb=\Rb^\top\bSigma_e^{1/2}\Zb_2\bOmega_e^{1/2}$, $\vec(\Nb)=\big(\bOmega_e^{1/2} \otimes (\Rb^\top\bSigma_e^{1/2})\big)\vec(\Zb_2)$. Then we have
	$$\vec(\Xb)=\frac{\sum_{j=1}^{p_2}\bN_{\cdot,j}\otimes \bN_{\cdot,j}}{\norm{\bu_2}^2}.$$
	Let $\Eb_j\Big(\big(\bOmega_e^{1/2} \otimes (\Rb^\top\bSigma_e^{1/2})\big)\vec(\Zb_2)\Big)$ extract the $\big((j-1)p_1+1\big)$-th to $jp_1$-th entries of $\big(\bOmega_e^{1/2} \otimes (\Rb^\top\bSigma_e^{1/2})\big)\vec(\Zb_2)$. Then,
	\begin{equation}
		\begin{aligned}
			&\vec(\Xb)=\frac{\sum_{j=1}^{p_2}\bigg(\Eb_j \Big(\big(\bOmega_e^{1/2} \otimes (\Rb^\top \bSigma_e^{1/2})\big)\vec(\Zb_2)\Big)\bigg)\otimes \bigg(\Eb_j \Big(\big(\bOmega_e^{1/2} \otimes (\Rb^\top \bSigma_e^{1/2})\big)\vec(\Zb_2)\Big)\bigg)}{\norm{\bu_2}^2},\nonumber
		\end{aligned}
	\end{equation}
where
\[
\Eb_j \big(\bOmega_e^{1/2} \otimes ( \Rb^\top \bSigma_e^{1/2})\vec(\Zb_2)  \big)=\Eb_j\left(-(\Cb \otimes \Rb^\top\Rb)u_1+(\Ib_{p_2} \otimes \Rb^\top)(\bOmega_e \otimes \bSigma_e )\bSigma_x^{-1}\bu_2\right).
\]
	Therefore,
	\begin{equation}
		\begin{aligned}
			\EE \vec(\Xb)&=\sum_{j=1}^{p_1} \EE(\norm{\bu_2}^{-2})\big(\Eb_j(\Cb \otimes \Rb^\top\Rb)\big)\otimes\big(\Eb_j(\Cb \otimes \Rb^\top\Rb)\big)\bSigma_{u_1}\\
			&+\sum_{j=1}^{p_1}\EE \Bigg(\frac{\Big(\Eb_j\big((\Ib_{p_2}\otimes \Rb^\top)(\bOmega_e \otimes \bSigma_e)\bSigma_x^{-1}\bu_2\big)\Big)\otimes \Big(\Eb_j\big((\Ib_{p_2}\otimes \Rb^\top)(\bOmega_e \otimes \bSigma_e)\bSigma_x^{-1}\bu_2\big)\Big) }{\norm{\bu_2}^2}\Bigg).\nonumber
		\end{aligned}
	\end{equation}
It is easy to verify that
	$$\bigg \Vert\EE(\norm{\bu_2}^{-2})\big(\Eb_j(\Cb \otimes \Rb^\top\Rb)\big)\otimes\big(\Eb_j(\Cb \otimes \Rb^\top\Rb)\big)\bSigma_{u_1}\bigg \Vert_F^2=O\Big(\frac{1}{p_2^4}\Big),$$
	and
	\begin{equation}
		\begin{aligned}
			&\Bigg \Vert\EE \Bigg(\frac{\Big(\Eb_j\big((\Ib_{p_2}\otimes \Rb^\top)(\bOmega_e \otimes \bSigma_e)\bSigma_x^{-1}\bu_2\big)\Big)\otimes \Big(\Eb_j\big((\Ib_{p_2}\otimes \Rb^\top)(\bOmega_e \otimes \bSigma_e)\bSigma_x^{-1}\bu_2\big)\Big) }{\norm{\bu_2}^2}\Bigg) \Bigg \Vert_F^2\\
			&\leq \norm{\big(\Eb_j(\Ib_{p_2}\otimes \Rb^\top)(\bOmega_e \otimes \bSigma_e)\big)\bSigma_x^{-1}\big(\Eb_j(\Ib_{p_2}\otimes \Rb^\top)(\bOmega_e \otimes \bSigma_e)\big)^\top}_F^2\bigg \Vert \EE \bigg(\frac{\bSigma_x^{-1/2}\bu_2\bu_2^\top\bSigma_x^{-1/2}}{\norm{\bu_2}^2}\bigg)\bigg \Vert^2\\
			&\leq \norm{\Eb_j(\Ib_{p_2}\otimes \Rb^\top)}_F^4\norm{\bOmega_e \otimes \bSigma_e}^4\norm{\bSigma_x^{-1}}^2\bigg \Vert \EE \bigg(\frac{\bSigma_x^{-1/2}\bu_2\bu_2^\top\bSigma_x^{-1/2}}{\norm{\bu_2}^2}\bigg)\bigg \Vert^2\\
			&=O\Big(\frac{1}{p_2^2}\Big).\nonumber
		\end{aligned}
	\end{equation}
	Then, $\norm{\EE\vec(\Xb)}_F^2=p_2^2O(\frac{1}{p_2^4}+\frac{1}{p_2^2})=O(1)$. On the other hand,
	$$\EE\norm{\Xb}_F^2 \leq C\sqrt{\EE(\norm{\Rb^\top\bSigma_e^{1/2}\Zb_2\bOmega\Zb_2^\top\bSigma_e^{1/2}\Rb}_F^4)\EE(\norm{\bu_2}^{-8})}=O(1).$$
	Hence, $\norm{\Rb^\top\Mb_4\Rb}_F^2=O_p(1)$ and
	$$\bigg\Vert\frac{1}{p_1}\Rb^\top\Mb_4\hat{\Rb}\bigg \Vert_F^2 \lesssim \frac{1}{p_1^2}\Big(\norm{\Rb^\top\Mb_4\Rb}_F^2\norm{\hat{\Hb}_R}_F^2+\norm{\Rb^\top\Mb_4}_F^2\norm{\hat{\Rb}-\Rb\hat{\Hb}_R}_F^2\Big)=O\Big(\frac{1}{p_1^2}+\frac{1}{T^2p_2^2}\Big).$$
	As a result,
	 $$\bigg\Vert\frac{1}{p_1}\hat{\Rb}^\top\big(\hat{\Rb}-\Rb\hat{\Hb}_R\big)\bigg\Vert_F^2=O_p\Big(\frac{1}{Tp_1p_2}+\frac{1}{p_1^2}+\frac{1}{T^2p_2^2}\Big)=O_p\Big(\frac{1}{Tp_1p_2}\Big).$$
	
Similarly, we can get $$\bigg \Vert\frac{1}{p_2}(\hat{\Cb}-\Cb\hat{\Hb}_C)^\top\hat{\Cb}\bigg \Vert_F^2=O_p\Big(\frac{1}{Tp_1p_2}\Big),$$
	which concludes the lemma.
\end{proof}

\begin{lemma}\label{eigenvalue}
	Under Assumptions \ref{joint_elliptical}-\ref{regular_noise}, for any constant $\epsilon>0$ we have
	\begin{equation}
		\left\{
		\begin{aligned}
			&\lambda_j(\hat{\Kb}_r) \asymp 1,\quad \quad j \leq k_1,\\
			&\lambda_j(\hat{\Kb}_r) \le  O_p(T^{-1/2+\epsilon}+p_2^{-1/2}),\quad j \textgreater k_1. \nonumber
		\end{aligned}
		\right.
	\end{equation}
	\begin{equation}
		\left\{
		\begin{aligned}
			&\lambda_j(\hat{\Kb}_c) \asymp 1,\quad \quad j \leq k_2,\\
			&\lambda_j(\hat{\Kb}_c) \le O_p(T^{-1/2+\epsilon}+p_1^{-1/2}),\quad j \textgreater k_2. \nonumber
		\end{aligned}
		\right.
	\end{equation}
\end{lemma}

\begin{proof}
	We start with the population Kendall's tau matrix $\Kb_r$. By definition,
	\begin{equation}\label{Kb_r decompose}
		 \Kb_r=\mathbb{E}\frac{(\Rb\Zb_1\Cb^\top+\bSigma_{e}^{1/2}\Zb_2\bOmega_{e}^{1/2})(\Rb\Zb_1\Cb^\top+\bSigma_{e}^{1/2}\Zb_2\bOmega_{e}^{1/2})^\top}{\|\Rb\Zb_1\Cb^\top+\bSigma_{e}^{1/2}\Zb_2\bOmega_{e}^{1/2}\|_F^2}.
	\end{equation}
	Write
	\[
	 \Kb_r^4:=\mathbb{E}\frac{\bSigma_{e}^{1/2}\Zb_2\bOmega_{e}\Zb_2^\top\bSigma_{e}^{1/2}}{\|\Rb\Zb_1\Cb^\top+\bSigma_{e}^{1/2}\Zb_2\bOmega_{e}^{1/2}\|_F^2},
	\]
	and we aim to provide an upper bound for $\|\Kb_r^4\|$. Using the formula $a^{-1}=b^{-1}-(ab)^{-1}(a-b)$,
	\[
	\begin{split}
		 \Kb_r^4=&\mathbb{E}\frac{\bSigma_{e}^{1/2}\Zb_2\bOmega_{e}\Zb_2^\top\bSigma_{e}^{1/2}}{\|\Rb\Zb_1\Cb^\top\|_F^2+\|\bSigma_{e}^{1/2}\Zb_2\bOmega_{e}^{1/2}\|_F^2}-\mathbb{E}\frac{\bSigma_{e}^{1/2}\Zb_2\bOmega_{e}\Zb_2^\top\bSigma_{e}^{1/2}\times 2\text{tr}(\Rb\Zb_1\Cb^\top\bOmega_{e}^{1/2}\Zb_2\bSigma_{e}^{1/2})}{\|\Rb\Zb_1\Cb^\top+\bSigma_{e}^{1/2}\Zb_2\bOmega_{e}^{1/2}\|_F^2(\|\Rb\Zb_1\Cb^\top\|_F^2+\|\bSigma_{e}^{1/2}\Zb_2\bOmega_{e}^{1/2}\|_F^2)}\\
		:=&\Kb_r^{41}+\Kb_r^{42}.
	\end{split}
	\]
	Note that the error term $\Kb_r^{42}$ satisfies
	\[
	\|\Kb_r^{42}\|\le \mathbb{E}\frac{ 2\text{tr}(\Rb\Zb_1\Cb^\top\bOmega_{e}^{1/2}\Zb_2\bSigma_{e}^{1/2})}{\|\Rb\Zb_1\Cb^\top+\bSigma_{e}^{1/2}\Zb_2\bOmega_{e}^{1/2}\|_F^2}\lesssim \sqrt{\mathbb{E}\|\bu_2\|^{-4}\times \mathbb{E}\text{tr}^2(\Rb\Zb_1\Cb^\top\bOmega_{e}^{1/2}\Zb_2\bSigma_{e}^{1/2})}\le O(p_2^{-1/2}),
	\]
	where $\bu_2$ are defined in Lemma \ref{m2m3}.
	Therefore, by Weyl's theorem, we can write
	\[
	\begin{split}
		\|\Kb_r^4\|\le& \|\Kb_r^{41}\|+O(p_2^{-1/2})\le \mathbb{E}\bigg\|\frac{\bSigma_{e}^{1/2}\Zb_2\bOmega_{e}\Zb_2^\top\bSigma_{e}^{1/2}}{\|\bSigma_{e}^{1/2}\Zb_2\bOmega_{e}^{1/2}\|_F^2}\bigg\|+O(p_2^{-1/2})\\
		\le &\lambda_{\min}(\bSigma_{e})^{-1}\lambda_{\min}(\Omega_{e})^{-1}\lambda_{\max}(\bSigma_{e})\mathbb{E}\frac{\Zb_2\bOmega_{e}\Zb_2^\top}{\|\Zb_2\|_F^2}+O(p_2^{-1/2})\le O(p_2^{-1/2}).
	\end{split}
	\]

	Next, consider the leading term in $\Kb_r$, denoted as
	\[
	 \Kb_r^1:=\mathbb{E}\frac{\Rb\Zb_1\Cb^\top\Cb\Zb_1^\top\Rb^\top}{\|\Rb\Zb_1\Cb^\top+\bSigma_{e}^{1/2}\Zb_2\bOmega_{e}^{1/2}\|_F^2}.
	\]
	Similarly, we can write
	\[
	\begin{split}
		 \Kb_r^1=&\mathbb{E}\frac{\Rb\Zb_1\Cb^\top\Cb\Zb_1^\top\Rb^\top}{\|\Rb\Zb_1\Cb^\top\|_F^2+\|\bSigma_{e}^{1/2}\Zb_2\bOmega_{e}^{1/2}\|_F^2}-\mathbb{E}\frac{\Rb\Zb_1\Cb^\top\Cb\Zb_1^\top\Rb^\top\times 2\text{tr}(\Rb\Zb_1\Cb^\top\bOmega_{e}^{1/2}\Zb_2\bSigma_{e}^{1/2})}{\|\Rb\Zb_1\Cb^\top+\bSigma_{e}^{1/2}\Zb_2\bOmega_{e}^{1/2}\|_F^2(\|\Rb\Zb_1\Cb^\top\|_F^2+\|\bSigma_{e}^{1/2}\Zb_2\bOmega_{e}^{1/2}\|_F^2)}\\
		:=&\Kb_r^{11}+\Kb_r^{12},
	\end{split}
	\]
	and claim that $\|\Kb_r^{12}\|\le O(p_2^{-1/2})$ while $\|\Kb_r^{11}\|\le C$ for some constant $C>0$. Then, by Cauchy-Schwartz inequality, the interaction terms in (\ref{Kb_r decompose}) satisfy
	\[
	 \bigg\|\mathbb{E}\frac{\Rb\Zb_1\Cb^\top\bOmega_{e}^{1/2}\Zb_2^\top\bSigma_{e}^{1/2}}{\|\Rb\Zb_1\Cb^\top+\bSigma_{e}^{1/2}\Zb_2\bOmega_{e}^{1/2}\|_F^2}\bigg\|\le \sqrt{\|\Kb_r^1\|^2\|\Kb_r^4\|^2}=o(1).
	\] Further note that the rank of $\Kb_r^1$ is at most $k_1$. Therefore, it remains to prove that $\lambda_{k_1}(\Kb_r^{11})\ge C^{-1}$ for some constant $C>0$.
	
	Due to the orthogonal invariant property of Gaussian distributions, we assume $\Cb^\top\Cb$ and $\Rb^\top\Rb$ to be diagonal matrices. By Assumptions \ref{strong_factor} and \ref{regular_noise}, we can write
	\[
	\begin{split}
		\lambda_{k_1}(\Kb_r^{11})\ge& \lambda_{k_1}\bigg(\mathbb{E}\frac{\Rb\Zb_1\Cb^\top\Cb\Zb_1^\top\Rb^\top}{c_2^2p_1p_2\|\Zb_1\|_F^2+C_2^2\|\Zb_2\|_F^2}\bigg)\ge \lambda_{k_1}(\Rb^\top\Rb)\Ab_{k_1k_1},\quad \text{with}\\
		\Ab=&\mathbb{E}\frac{\Zb_1\Cb^\top\Cb\Zb_1^\top}{c_2^2p_1p_2\|\Zb_1\|_F^2+C_2^2\|\Zb_2\|_F^2},
	\end{split}
	\]
	where  the second ``$\ge$'' comes from the facts that columns of $\Rb$ are orthogonal and  $\Ab_{ij}=0$ for $i\ne j$. See also the proof of Proposition \ref{pro1}.  Note that $p_1\Ab_{k_1k_1}$ is similar to $M_{jj}$ in the proof of Lemma 3.1, \cite{yu2019robust}. According to their results, we have $\lambda_{k_1}(\Kb_r^{11})\ge C^{-1}$ for some constant $C>0$.
	
	Next step, we consider the sample Kendall's tau matrix and show that $\|\hat\Kb_r-\Kb_r\|=o_p(p_2^{1/2})$.
	Assume $T$ is even and $\bar{T}=T/2$, otherwise we can delete the last observation. For any permutation $\sigma$ of $\{1,\dots,T \}$. Define $\bomega_s^\sigma=\Xb_{2s-1}-\Xb_{2s}/\norm{\Xb_{2s-1}-\Xb_{2s}}_F$ for $s \in \{1,\dots,\bar{T} \}$ and $\hat{\Kb}_r^\sigma=\bar{T}^{-1}\sum_{s=1}^{\bar{T}}\bomega_s^\sigma\bomega_s^{\sigma \top}$. Denote $\mathcal{S}_T$ as the set containing all the permutations of $\{1,\dots,T \}$, then we can get
	$$\sum_{\sigma \in \mathcal{S}_T}\bar{T}\hat{\Kb}_r^\sigma=T\times (T-2)!\times \frac{T(T-1)}{2}\hat{\Kb}_r \Rightarrow \hat{\Kb}_r=\frac{1}{T!}\sum_{\sigma \in \mathcal{S}_T}\hat{\Kb}_r^\sigma.$$
	because
	$$\norm{\hat{\Kb}_r-\Kb_r}=\bigg \Vert\frac{1}{\bar{T}}\sum_{\sigma \in \mathcal{S}_T}(\hat{\Kb}_r^\sigma-\Kb_r)\bigg\Vert \leq \frac{1}{\bar{T}}\sum_{\sigma \in \mathcal{S}_T}\norm{\hat{\Kb}_r^\sigma-\Kb_r}$$ and
	 $$\norm{\hat{\Kb}_r^\sigma-\Kb_r}=\bigg\|\bar{T}^{-1}\sum_{s=1}^{\bar{T}}\bomega_s^\sigma\bomega_s^{\sigma\top}-\EE\bomega_s^\sigma\bomega_s^{\sigma\top}\bigg\|=\bigg\|\bar{T}^{-1}\sum_{s=1}^{\bar{T}}(\bomega_s^\sigma\bomega_s^{\sigma\top}-\EE\bomega_s^\sigma\bomega_s^{\sigma\top})\bigg\|.$$
	
	Define $\text{Vec}(\Zb_t)=\big(\text{Vec}(\Zb_{1t})^\top,\text{Vec}(\Zb_{2t})^\top\big)^\top$ as $i.i.d.$ $(k_1k_2+p_1p_2)$ standard Gaussian random vectors, and $\Zb_{1t}$, $\Zb_{2t}$ are $k_1\times k_2$, $p_1\times p_2$ matrices, respectively. Let
	\[
\begin{split}
		 \Kb_t=&\frac{(\Rb\Zb_{1t}\Cb^\top+\bSigma_{e}^{1/2}\Zb_{2t}\bOmega_{e}^{1/2})(\Rb\Zb_{1t}\Cb^\top+\bSigma_{e}^{1/2}\Zb_{2t}\bOmega_{e}^{1/2})^\top}{\|\Rb\Zb_{1t}\Cb^\top+\bSigma_{e}^{1/2}\Zb_{2t}\bOmega_{e}^{1/2}\|_F^2}\\
		 =&\frac{\Rb\Zb_{1t}\Cb^\top\Cb\Zb_{1t}^\top\Rb^\top}{\|\Rb\Zb_{1t}\Cb^\top+\bSigma_{e}^{1/2}\Zb_{2t}\bOmega_{e}^{1/2}\|_F^2}+\frac{\Rb\Zb_{1t}\Cb^\top\bOmega_{e}^{1/2}\Zb_{2t}^\top\bSigma_{e}^{1/2}}{\|\Rb\Zb_{1t}\Cb^\top+\bSigma_{e}^{1/2}\Zb_{2t}\bOmega_{e}^{1/2}\|_F^2}\\
		 &+\bigg(\frac{\Rb\Zb_{1t}\Cb^\top\bOmega_{e}^{1/2}\Zb_{2t}^\top\bSigma_{e}^{1/2}}{\|\Rb\Zb_{1t}\Cb^\top+\bSigma_{e}^{1/2}\Zb_{2t}\bOmega_{e}^{1/2}\|_F^2}\bigg)^\top+\frac{\bSigma_{e}^{1/2}\Zb_{2t}\bOmega_{e}\Zb_{2t}^\top\bSigma_{e}^{1/2}}{\|\Rb\Zb_{1t}\Cb^\top+\bSigma_{e}^{1/2}\Zb_{2t}\bOmega_{e}^{1/2}\|_F^2}\\
		:=&\Kb_t^1+\Kb_t^2+\Kb_t^3+\Kb_t^4.
\end{split}
	\]
	Then, the large sample properties of $\norm{\hat{\Kb}_r-\Kb_r}$ is similar to those of $\|\bar T^{-1}\sum_{t=1}^{\bar T}\Kb_t-\mathbb{E}\Kb_1\|$. Similarly to $\|\Kb_r\|$, one can verify that
	\[
	\mathbb{E}\|\Kb_t^1\|\le C,\quad \mathbb{E}\|\Kb_t^j\|^2\le O(p_2^{-1}),j=2,3,4.
	\]
	Therefore,
	according to matrix concentration inequality (Theorem 5.4.1 in \cite{vershynin2018high}), we have
	$$\PP\bigg(\bigg\Vert\sum_{t=1}^{\bar{T}}(\Kb_t^1- \mathbb{E}\Kb_t^1)\bigg\Vert\geq s\bigg)\leq k_1\exp\Big\{-\frac{s^2}{cT}\Big\}\iff \PP\bigg(\bigg\Vert\frac{1}{\bar{T}}\sum_{t=1}^{\bar{T}}(\Kb_t^1- \mathbb{E}\Kb_t^1)\bigg\Vert\geq T^{-1/2+\epsilon}\bigg)\leq k_1\exp\Big\{-\frac{T^{2\epsilon}}{c}\Big\},$$
	for any $\epsilon>0$.
	Hence, when we have $\|\bar T^{-1}\sum_{t=1}^{\bar{T}}(\Kb_t^1- \mathbb{E}\Kb_t^1)\|=O_p(T^{-1/2+\epsilon})$ for any $\epsilon>0$. On the other hand, for $j=,2,3,4$,
	\[
	\PP\bigg(\bigg\Vert\sum_{t=1}^{\bar{T}}(\Kb_t^j- \mathbb{E}\Kb_t^j)\bigg\Vert\geq s\bigg)\leq p_1\exp\Big\{-\frac{p_2s^2}{cT}\Big\}\iff \PP\bigg(\bigg\Vert\frac{1}{\bar{T}}\sum_{t=1}^{\bar{T}}(\Kb_t^j- \mathbb{E}\Kb_t^j)\bigg \Vert\geq p_2^{-1/2}\bigg)\leq p_1\exp\Big\{-\frac{T}{c}\Big\}.
	\]
	That is, $\|\bar T^{-1}\sum_{t=1}^{\bar{T}}(\Kb_t^j- \mathbb{E}\Kb_t^j)\|=O_p(p_2^{-1/2})$ as long as $\log p_1=o(T)$. As a result,
	\[
	\bigg\Vert \frac{1}{\bar T}\sum_{t=1}^{\bar{T}}(\Kb_t- \mathbb{E}\Kb_t)\bigg\Vert=O_p(T^{-1/2+\epsilon}+p_2^{-1/2}),
	\]
	for any $\epsilon>0$, which concludes the lemma.
\end{proof}

\clearpage

\section{Additional Simulation Results}
\begin{table}[!h]
	 	\caption{Averaged estimation errors and standard errors of $\mathcal{D} (\hat{\mathbf{C}}, \mathbf{C})$ for Settings A under joint Normal distribution and joint $t$ distribution over 500 replications.}
	 	 \label{tab:main5}\renewcommand{\arraystretch}{1} \centering
	 	\selectfont
	 		 \scalebox{0.9}{\begin{tabular*}{18cm}{cccccccccccccccccccc}
	 				\toprule[2pt]
	 			    &\multirow{1}{*}{Evaluation}   &\multirow{1}{*}{$T$}
	 				&\multirow{1}{*}{$p_1$}        &\multirow{1}{*}{$p_2$}
                    &\multirow{1}{*}{MRTS}          &\multirow{1}{*}{RMFA}
                    &\multirow{1}{*}{$\alpha$-PCA} &\multirow{1}{*}{PE}
                    &\multirow{1}{*}{MPCA$_{F}$}   \cr
	 			    \cmidrule(lr){8-11} \\
	 				\midrule[1pt]
	 				
	 				 &&&& \multicolumn{5}{c}{\multirow{1}{*}{\textbf{Normal Distribution}}}\\
	 \cmidrule(lr){3-12}
	 				&$\mathcal{D} (\hat{\mathbf{C}}, \mathbf{C})$ &20 &20  &20 & 0.1192(0.0331) &0.0914(0.0151) &0.1128(0.0307) &0.0921(0.0157) &0.1400(0.0206)\\
	 				                                           &  &   &50  &50  &0.0623(0.0074)  &0.0569(0.0060) &0.0598(0.0070) &0.0568(0.0060) &0.0986(0.0093)  \\

   \cmidrule(lr){3-12}
	 				&$\mathcal{D} (\hat{\mathbf{C}}, \mathbf{C})$ &50 &20  &20  & 0.0807(0.0228) &0.0573(0.0099) &0.0771(0.0221) &0.0575(0.0099) &0.0866(0.0132)  \\
	 				                                           &  &   &50  &50  & 0.0389(0.0043) &0.0352(0.0033) & 0.0378(0.0041)&0.0351(0.0032) & 0.0610(0.0053) \\
	 				
	 				\cmidrule(lr){3-12}
	 				&$\mathcal{D} (\hat{\mathbf{C}}, \mathbf{C})$ &100 &20  &20 &0.0661(0.0256)& 0.0405(0.0067)& 0.0627(0.0245)& 0.0406(0.0069£©&
 0.0607(0.0088)  \\
	 				                                           &  &   &50  &50  &0.0279(0.0033)& 0.0246(0.0021)& 0.0271(0.0032)& 0.0246(0.0021)&
 0.0429(0.0033)   \\
                     \hline
	 				 &&&& \multicolumn{5}{c}{\multirow{1}{*}{\textbf{$t_1$ Distribution}}}\\
  \cmidrule(lr){3-12}
	 				&$\mathcal{D} (\hat{\mathbf{C}}, \mathbf{C})$ &20 &20  &20  &0.1294(0.0370)  &0.2298(0.1444) &0.4353(0.1753) &0.4246(0.1900) &0.1400(0.0221)\\
	 				                                           &  &   &50  &50  &0.0689(0.0084)  &0.1327(0.0976) &0.2845(0.1782) &0.2829(0.1872) &0.0989(0.0095)  \\

   \cmidrule(lr){3-12}
	 				&$\mathcal{D} (\hat{\mathbf{C}}, \mathbf{C})$ &50 &20  &20  & 0.0884(0.0256) &0.1884(0.1431) &0.4187(0.1871) &0.4069(0.2023) &0.0865(0.0121)  \\
	 				                                           &  &   &50  &50  &0.0428(0.0044)  &0.1070(0.0811) &0.2694(0.1771) &0.2641(0.1824) &0.0610(0.0049)  \\
	 				
	 				\cmidrule(lr){3-12}
	 				&$\mathcal{D} (\hat{\mathbf{C}}, \mathbf{C})$ &100 &20  &20 & 0.0685(0.0221)& 0.1767(0.1485) &0.4304(0.1918) &0.4232(0.2068)&
 0.0606(0.0085)  \\
	 				                                           &  &   &50  &50  &0.0308(0.0032)& 0.0939(0.0825)& 0.2678(0.1740)& 0.2636(0.1814)&
 0.0431(0.0033)   \\
                     \hline
	 				&&&& \multicolumn{5}{c}{\multirow{1}{*}{\textbf{$t_2$ Distribution}}}\\
	 				\cmidrule(lr){3-12}
	 				&$\mathcal{D} (\hat{\mathbf{C}}, \mathbf{C})$ &20 &20  &20  &0.1263(0.0368)  &0.1296(0.0549) &0.2655(0.1521) &0.2445(0.1603) &0.1402(0.0210)\\
	 				                                           &  &   &50  &50  &0.0667(0.0086)  &0.0783(0.0173) &0.1484(0.1011) &0.1416(0.1015) & 0.0996(0.0099) \\

   \cmidrule(lr){3-12}
	 				&$\mathcal{D} (\hat{\mathbf{C}}, \mathbf{C})$ &50 &20  &20  &0.0843(0.0236)  &0.0853(0.0255) &0.2153(0.1402) &0.1946(0.1464) &0.0861(0.0121)  \\
	 				                                           &  &   &50  &50  & 0.0415(0.0045) &0.0525(0.0171) &0.1193(0.0983) &0.1148(0.0992) &0.0613(0.0047)  \\
	 				\cmidrule(lr){3-12}
	 				&$\mathcal{D} (\hat{\mathbf{C}}, \mathbf{C})$ &100 &20  &20 & 0.0689(0.0270)& 0.0629(0.0209)& 0.1925(0.1338)& 0.1723(0.1399)&
 0.0610(0.0091)  \\
	 				                                           &  &   &50  &50  &0.0300(0.0034)& 0.0400(0.0240)& 0.1010(0.0862)& 0.0961(0.0856)&
 0.0429(0.0033)   \\
	 				
                     \hline
                     &&&& \multicolumn{5}{c}{\multirow{2}{*}{\textbf{$t_3$ Distribution}}}\\
                     \cmidrule(lr){3-12}
	 				&$\mathcal{D} (\hat{\mathbf{C}}, \mathbf{C})$ &20 &20  &20  & 0.1248(0.0329) &0.1111(0.0251) &0.1947(0.1000) &0.1669(0.0938) &0.1409(0.0220)\\
	 				                                           &  &   &50  &50  &0.0657(0.0079)  &0.0674(0.0100) &0.0974(0.0400) &0.0924(0.0376) &0.0988(0.0096)  \\

   \cmidrule(lr){3-12}
	 				&$\mathcal{D} (\hat{\mathbf{C}}, \mathbf{C})$ &50 &20  &20  & 0.0844(0.0233) &0.0703(0.0154) &0.1406(0.0788) &0.1192(0.0789) &0.0869(0.0125)  \\
	 				                                           &  &   &50  &50  & 0.0411(0.0045) &0.0438(0.0061) &0.0756(0.0488) &0.0715(0.0481) & 0.0616(0.0050) \\
	 				\cmidrule(lr){3-12}
	 				&$\mathcal{D} (\hat{\mathbf{C}}, \mathbf{C})$ &100 &20  &20 & 0.0677(0.0211)& 0.0508(0.0104)& 0.1149(0.0694)& 0.0944(0.0637)&
 0.0607(0.0079)  \\
	 				                                           &  &   &50  &50  & 0.0297(0.0034)& 0.0311(0.0042)& 0.0604(0.0410)& 0.0571(0.0432)&
 0.0429(0.0030)  \\
	 				\bottomrule[2pt]
 				\end{tabular*}}
     \end{table}
\begin{table}[!h]
	 	\caption{Averaged estimation errors and standard errors of $\mathcal{D} (\hat{\mathbf{C}}, \mathbf{C})$ for Scenario B under joint Normal distribution and joint $t$ distribution over 500 replications.}
	 	 \label{tab:main6}\renewcommand{\arraystretch}{1} \centering
	 	\selectfont
	 		 \scalebox{0.9}{\begin{tabular*}{18cm}{cccccccccccccccccccc}
	 				\toprule[2pt]
	 			    &\multirow{1}{*}{Evaluation}   &\multirow{1}{*}{$T$}
	 				&\multirow{1}{*}{$p_1$}        &\multirow{1}{*}{$p_2$}
                    &\multirow{1}{*}{MRTS}          &\multirow{1}{*}{RMFA}
                    &\multirow{1}{*}{$\alpha$-PCA} &\multirow{1}{*}{PE}
                    &\multirow{1}{*}{MPCA$_{F}$}   \cr
	 			    \cmidrule(lr){8-11} \\
	 				\midrule[1pt]
	 				
	 				 &&&& \multicolumn{5}{c}{\multirow{1}{*}{\textbf{Normal Distribution}}}\\
	 \cmidrule(lr){3-12}
	 				&$\mathcal{D} (\hat{\mathbf{C}}, \mathbf{C})$ &20 &20  &20 &0.1202(0.0331)  &0.0923(0.0151) &0.1139(0.0305) &0.0930(0.0157) &0.1410(0.0209)\\
	 				                                           &  &   &50  &50  &0.0627(0.0074)  &0.0575(0.0061) &0.0604(0.0071) &0.0574(0.0061) &0.0993(0.0094)  \\

   \cmidrule(lr){3-12}
	 				&$\mathcal{D} (\hat{\mathbf{C}}, \mathbf{C})$ &50 &20  &20  &0.0813(0.0231)  &0.0578(0.0100) &0.0778(0.0224) &0.0580(0.0100) &0.0872(0.0133)  \\
	 				                                           &  &   &50  &50  & 0.0392(0.0044) &0.0356(0.0033) &0.0382(0.0042) &0.0355(0.0033) &0.0613(0.0052)  \\
	 				
	 				\cmidrule(lr){3-12}
	 				&$\mathcal{D} (\hat{\mathbf{C}}, \mathbf{C})$ &100 &20  &20 & 0.0663(0.0255)& 0.0408(0.0067)& 0.0629(0.0245)&
 0.0409(0.0069)& 0.0609(0.0088)  \\
	 				                                           &  &   &50  &50  &0.0282(0.0033)& 0.0249(0.0021)& 0.0274(0.0032)&
  0.0248(0.0021)& 0.0430(0.0034)   \\
                     \hline
	 				 &&&& \multicolumn{5}{c}{\multirow{1}{*}{\textbf{$t_1$ Distribution}}}\\
  \cmidrule(lr){3-12}
	 				&$\mathcal{D} (\hat{\mathbf{C}}, \mathbf{C})$ &20 &20  &20  &0.1442(0.0410)  &0.2400(0.1498) &0.4359(0.1751) &0.4249(0.1894) &0.1516(0.0259)\\
	 				                                           &  &   &50  &50  &0.0771(0.0114)  &0.1385(0.1010) &0.2841(0.1776) &0.2827(0.1866) &0.1059(0.0121)  \\

   \cmidrule(lr){3-12}
	 				&$\mathcal{D} (\hat{\mathbf{C}}, \mathbf{C})$ &50 &20  &20  & 0.0976(0.0273) &0.1960(0.1450) &0.4193(0.1871) &0.4075(0.2025) &0.0938(0.0136)  \\
	 				                                           &  &   &50  &50  &0.0484(0.0061)  &0.1125(0.0870) &0.2696(0.1770) &0.2646(0.1826) &0.0651(0.0058)  \\
	 				
	 				\cmidrule(lr){3-12}
	 				&$\mathcal{D} (\hat{\mathbf{C}}, \mathbf{C})$ &100 &20  &20 & 0.0751(0.0236)& 0.1845(0.1517)& 0.4303(0.1917)&
 0.4234(0.2067)& 0.0655(0.0094)  \\
	 				                                           &  &   &50  &50  &0.0345(0.0036)& 0.0989(0.0888)& 0.2678(0.1740)&
 0.2639(0.1817)& 0.0459(0.0036)   \\
                     \hline
	 				&&&& \multicolumn{5}{c}{\multirow{1}{*}{\textbf{$t_2$ Distribution}}}\\
	 				\cmidrule(lr){3-12}
	 				&$\mathcal{D} (\hat{\mathbf{C}}, \mathbf{C})$ &20 &20  &20  &0.1303(0.0378)  &0.1333(0.0610) &0.2662(0.1520) &0.2449(0.1606) &0.1430(0.0220)\\
	 				                                           &  &   &50  &50  &0.0690(0.0092)  &0.0801(0.0189) &0.1490(0.1023) &0.1418(0.1013) & 0.1009(0.0102) \\

   \cmidrule(lr){3-12}
	 				&$\mathcal{D} (\hat{\mathbf{C}}, \mathbf{C})$ &50 &20  &20  &0.0864(0.0235)  &0.0877(0.0275) &0.2157(0.1404) &0.1947(0.1461) &0.0881(0.0123)  \\
	 				                                           &  &   &50  &50  &0.0430(0.0047)  &0.0540(0.0188) &0.1196(0.0986) &0.1151(0.0998) &0.0623(0.0049) \\
	 				
	 				\cmidrule(lr){3-12}
	 				&$\mathcal{D} (\hat{\mathbf{C}}, \mathbf{C})$ &100 &20  &20 &0.0704(0.0271)& 0.0649(0.0223)&0.1926(0.1336)&
  0.1722(0.1397)& 0.0623(0.0091)   \\
	 				                                           &  &   &50  &50  &0.0311(0.0035)& 0.0413(0.0257)& 0.1012(0.0863)&
  0.0962(0.0856)& 0.0437(0.0033)   \\
                     \hline
                     &&&& \multicolumn{5}{c}{\multirow{2}{*}{\textbf{$t_3$ Distribution}}}\\
                     \cmidrule(lr){3-12}
	 				&$\mathcal{D} (\hat{\mathbf{C}}, \mathbf{C})$ &20 &20  &20  &0.1272(0.0351)  &0.1129(0.0262) &0.1958(0.1008) &0.1679(0.0943) &0.1430(0.0224)  \\
	 				                                           &  &   &50  &50  & 0.0668(0.0082) &0.0683(0.0105) &0.0978(0.0400) &0.0927(0.0374) &0.0998(0.0099)  \\

   \cmidrule(lr){3-12}
	 				&$\mathcal{D} (\hat{\mathbf{C}}, \mathbf{C})$ &50 &20  &20  &0.0857(0.0235)  &0.0715(0.0160) &0.1409(0.0791) &0.1193(0.0788) & 0.0879(0.0126) \\
	 				                                           &  &   &50  &50  &0.0419(0.0046)  &0.0446(0.0065) &0.0760(0.0502) &0.0718(0.0487) &0.0623(0.0049)  \\
	 				\cmidrule(lr){3-12}
	 				&$\mathcal{D} (\hat{\mathbf{C}}, \mathbf{C})$ &100 &20  &20 & 0.0685(0.0213)& 0.0518(0.0109)& 0.1152(0.0693)&
 0.0945(0.0636)& 0.0615(0.0079)  \\
	 				                                           &  &   &50  &50  &0.0303(0.0034)& 0.0317(0.0044)& 0.0605(0.0411)&
  0.0572(0.0432)& 0.0433(0.0030)  \\
	 				\bottomrule[2pt]
 				\end{tabular*}}
     \end{table}

\begin{table}[!h]
	 	\caption{Mean squared error and its standard under Scenario B over 500 replications.}
	 	 \label{tab:main7}\renewcommand{\arraystretch}{1} \centering
	 	\selectfont
	 		 \scalebox{1}{\begin{tabular*}{17cm}{ccccccccccccccccccccccccccccccccc}
	 				\toprule[2pt]
	 				&\multirow{2}{*}{Distribution}  &\multirow{2}{*}{$p_1$}
	 				&\multirow{2}{*}{MRTS}  &\multirow{2}{*}{RMFA}
                    &\multirow{2}{*}{$\alpha$-PCA}  &\multirow{2}{*}{PE}
                     &\multirow{2}{*}{MPCA$_{F}$}   \cr
	 				\cmidrule(lr){8-15} \\
	 				\midrule[1pt]
	 				 \multicolumn{8}{c}{\multirow{1}{*}{$T=20,p_2=p_1 $}}\\
	 				\cmidrule(lr){2-17}
	 				&Normal  &20   &0.0461(0.0066) &0.0376(0.0037) &0.0438(0.0061) &0.0378(0.0038) &0.0567(0.0068)  \\
	 				      &  &50  &0.0107(0.0009) &0.0096(0.0007)& 0.0101(0.0008)& 0.0096(0.0007( & 0.0216(0.0026) \\
	 				
	 				\cmidrule(lr){2-17}
	 				&$t_3$  &20   &0.1446(0.1434)& 0.1542(0.2754)& 0.3455(0.9260)& 0.3291(0.9691)& 0.1675(0.1603)\\
	 				&  &50    &0.0319(0.0229)& 0.0370(0.0392)& 0.0730(0.1243)& 0.0716(0.1375)& 0.0595(0.0405) \\
	 			    \hline
	 				\multicolumn{8}{c}{\multirow{1}{*}{$T=50,p_2=p_1 $}}\\
	 				\cmidrule(lr){2-17}
	 				&Normal  &20   &0.0340(0.0043)& 0.0285(0.0024)& 0.0330(0.0040)& 0.0286(0.0025)& 0.0358(0.0033)   \\
	 				      &  &50  &0.0065(0.0004)& 0.0060(0.0003)& 0.0063(0.0004)& 0.0060(0.0003)& 0.0107(0.0008) \\
	 				
	 				\cmidrule(lr){2-17}
	 				&$t_3$  &20   &0.1056(0.0829)& 0.1033(0.1173)& 0.2794(0.9220)& 0.2676(0.9615)& 0.1073(0.0838)\\
	 				&  &50    &0.0205(0.0155)& 0.0251(0.0451)& 0.0648(0.2269)& 0.0628(0.2375)& 0.0319(0.0226)  \\
	 			\hline
	 			     \multicolumn{8}{c}{\multirow{1}{*}{$T=100,p_2=p_1 $}}\\
                     \cmidrule(lr){2-17}
                     &Normal  &20   &0.0301(0.0035)& 0.0254(0.0017)& 0.0293(0.0033)& 0.0254(0.0017)& 0.0290(0.0019)  \\
	 				      &  &50  &0.0051(0.0003)& 0.0048(0.0002)& 0.0050((0.0003)& 0.0048(0.0002)& 0.0071(0.0004)  \\
	 				
	 				\cmidrule(lr){2-17}
	 				&$t_3$  &20   &0.0905(0.0512)& 0.0868(0.0998)& 0.1968(0.6717)& 0.1837(0.7117)& 0.0862(0.0502) \\
	 				&  &50    & 0.0163(0.0123)& 0.0188(0.0263)& 0.0516(0.1748)& 0.0500(0.1813)& 0.0219(0.0153) \\
	 				\bottomrule[2pt]
	 		\end{tabular*}}
	 \end{table}

\begin{table}[!h]
	 	\caption{The frequencies of exact estimation and underestimation of the numbers of factors under Scenario B over 500 replications. }
	 	 \label{tab:main8}\renewcommand{\arraystretch}{1} \centering
	 	\selectfont
	 		 \scalebox{0.8}{\begin{tabular*}{20cm}{ccccccccccccccccccccccccccccccccc}
	 				\toprule[2pt]
	 		   		&\multirow{2}{*}{Distribution}  &\multirow{2}{*}{$p_1$}
                     &\multirow{2}{*}{MKER} & \multirow{2}{*}{Rit-ER}
	 				&\multirow{2}{*}{IterER}  &\multirow{2}{*}{$\alpha$-PCA-ER}
                    &\multirow{2}{*}{iTIP-IC}  &\multirow{2}{*}{iTIP-EC}
                     &\multirow{2}{*}{TCorTh}  \cr
	 				\cmidrule(lr){8-15} \\
	 				\midrule[1pt]
	 				 \multicolumn{10}{c}{\multirow{1}{*}{$T=20,p_2=p_1 $}}\\
	 				\cmidrule(lr){2-17}
	 				&Normal  &20   &0.664(0.046)&0.990(0.000) &0.990(0.000)&0.594(0.076)&0.000(1.000) &0.068(0.454)&0.668(0.050)   \\
	 				      &  &50  &1.000(0.000)&1.000(0.000) &1.000(0.000)&1.000(0.000)&0.000(1.000)&0.182(0.198) &1.000(0.000)   \\

	 				\cmidrule(lr){2-17}
	 				&$t_1$  &20   &0.568(0.09)&0.312(0.264)&0.220(0.484)&0.064(0.776) &0.082(0.288)&0.036(0.678)&0.278(0.378) \\
	 				&  &50    &0.992(0.000)&0.664(0.098) &0.534(0.306)&0.296(0.544) &0.114(0.388)&0.150(0.450)&0.700(0.162)   \\
	 				
	 				\cmidrule(lr){2-17}
                    &$t_2$  &20   &0.614(0.048)&0.702(0.072) &0.630(0.174)&0.224(0.470)&0.024(0.926) &0.094(0.570)&0.398(0.224) \\
	 				&  &50    &0.998(0.000)&0.908(0.002) &0.880(0.066)&0.682(0.166)&0.018(0.952)&0.164(0.338)&0.920(0.028)  \\
	 				\cmidrule(lr){2-17}
	 				&$t_3$  &20   &0.636(0.050)&0.866(0.018) &0.806(0.066)&0.364(0.286) &0.004(0.996)&0.076(0.494)&0.486(0.152) \\
	 				&  &50    &0.998(0.000)&0.990(0.000) &0.978(0.014)&0.882(0.048) &0.002(0.996)&0.152(0.252)&0.990(0.000)  \\
	 			    \hline
	 				\multicolumn{10}{c}{\multirow{1}{*}{$T=50,p_2=p_1 $}}\\
	 				\cmidrule(lr){2-17}
	 				&Normal  &20   &0.772(0.010)&1.000(0.000) &1.000(0.000)&0.732(0.028)&0.000(1.000) &0.124(0.374)&0.954(0.000)   \\
	 				      &  &50  &1.000(0.000)&1.000(0.000) &1.000(0.000)&1.000(0.000)&0.000(1.000)&0.234(0.130) &1.000(0.000)  \\

	 				\cmidrule(lr){2-17}
	 				&$t_1$  &20   &0.708(0.022)&0.424(0.194) &0.292(0.468)&0.110(0.738) &0.072(0.102)&0.082(0.642)&0.558(0.058) \\
	 				&  &50    &1.000(0.000)&0.686(0.062) &0.558(0.294)&0.354(0.492) &0.164(0.128)&0.258(0.356)&0.696(0.014)  \\
	 				
	 				\cmidrule(lr){2-17}
                    &$t_2$  &20   &0.802(0.02)&0.834(0.026) &0.736(0.126)&0.316(0.338) &0.034(0.878)&0.122(0.458)&0.790(0.014) \\
	 				&  &50    &1.000(0.000)&0.936(0.004) &0.912(0.056)&0.804(0.136) &0.018(0.944)&0.258(0.264)&0.930(0.004)   \\
	 				\cmidrule(lr){2-17}
	 				&$t_3$  &20   &0.758(0.022)&0.952(0.000) &0.908(0.034)&0.522(0.168) &0.002(0.994)&0.114(0.476)&0.878(0.012) \\
	 				&  &50    &1.000(0.000) &0.990(0.000) &0.988(0.006)&0.956(0.030)&0.000(1.000) &0.236(0.238)&0.992(0.000)   \\
	 			 \hline
	 				\multicolumn{10}{c}{\multirow{1}{*}{$T=100,p_2=p_1 $}}\\
	 				\cmidrule(lr){2-17}
	 				&Normal  &20   &0.824(0.010)&1.000(0.000)&1.000(0.000)&0.786(0.016)&0.000(1.000)&0.226(0.298)&0.994(0.000)  \\
	 				      &  &50  & 1.000(0.000)&1.000(0.000)&1.000(0.000)&1.000(0.000)&0.000(1.000)&0.374(0.084)&1.000(0.000) \\

	 				\cmidrule(lr){2-17}
	 				&$t_1$  &20   &0.790(0.012)&0.394(0.218)&0.280(0.440)&0.100(0.722)&0.030(0.004)&0.168(0.532)&0.458(0.016)\\
	 				&  &50    &1.000(0.000)&0.702(0.052)&0.568(0.308)&0.350(0.512)&0.134(0.020)&0.364(0.280)&0.466(0.002) \\
	 				
	 				\cmidrule(lr){2-17}
                    &$t_2$  &20   &0.802(0.012)&0.870(0.016)&0.780(0.110)&0.406(0.300)&0.038(0.860)&0.216(0.384)&0.870(0.002) \\
	 				&  &50    &1.000(0.000)&0.962(0.002)&0.938(0.032)&0.854(0.096)&0.038(0.896)&0.424(0.178)&0.900(0.000) \\
	 				\cmidrule(lr){2-17}
	 				&$t_3$  &20   &0.820(0.006)&0.966(0.000)&0.952(0.014)&0.636(0.080)&0.000(0.996)&0.220(0.370)&0.960(0.000)\\
	 				&  &50    &1.000(0.000)&0.992(0.000)&0.992(0.006)&0.968(0.016)&0.000(1.000)&0.412(0.120)&0.980(0.000) \\
	 				\bottomrule[2pt]
	 		\end{tabular*}}
	 \end{table}

\end{document}